\documentclass[12pt]{article}
\usepackage[left=1in,right=1in,top=1in,bottom=1in]{geometry}
\usepackage{parskip}
\DeclareMathAlphabet\mathbfcal{OMS}{cmsy}{b}{n}
    \newcommand{\mbc}[1]{\ensuremath{\mathbfcal{#1}}}
    
\usepackage{blkarray}

\usepackage{amsmath, amsthm, amsfonts, amssymb, mathtools,  mathrsfs}
\usepackage{thmtools}
   
\usepackage{tikz, graphicx}
    \usetikzlibrary{positioning}
    \usetikzlibrary{arrows, arrows.meta, decorations.pathmorphing, decorations.markings}
    \newcommand{\cf}[1]{\textsf{#1}}
    \tikzset{%
        fwdrxn/.style={thick, 
        arrows={-left to},  
        text=ratecnst, font=\scriptsize},
        revrxn/.style={very thick, arrows={-Stealth[length=5pt,width=5pt,left]}},
        newt/.style={turq, opacity=0.085}
    }
    \tikzset{%
        EdgeShift/.style n args={2}{transform canvas={xshift={#1}, yshift={#2}}}, %
        Jnode/.style={circle, draw=black, thick,  fill=JnodeColor, fill opacity = 1, inner sep=0pt,minimum size=0.48cm, outer sep=2pt, font=\footnotesize},
        posEdge/.style={very thick, arrows={-Stealth[length=7.5pt,width=7.5pt]} },
        negEdge/.style={very thick, arrows={-Bar[length=0pt,width=7.5pt]} },
        neutEdge/.style={very thick, -{Circle[open, scale=0.75]} } 
    }
    
    \tikzset{
        o/.style={very thick, shorten >={#1} , decoration={markings,
                mark={at position 1 with {\draw[xshift=-{0.5*#1}] circle [radius={2pt}];}}  }, postaction=decorate},
        o/.default=5pt
    }
    \definecolor{JnodeColor}{RGB}{235,230,237}
    \definecolor{ActivEdgeGreen}{RGB}{34,167,132}
    \definecolor{InhibEdgeRed}{RGB}{153,8,3}
    \definecolor{JnodeColorExternal}{RGB}{168, 165, 162}
    \definecolor{vblue}{RGB}{34, 69, 151}

\usepackage{color, xcolor}

    \definecolor{Teal}{RGB}{0,153,115}
    \definecolor{DarkTurq}{RGB}{0,153,115}
    \definecolor{Orange}{RGB}{255,95,5}
    \definecolor{NavyBlue}{RGB}{19,41,75}
    \definecolor{OrangeBlood}{RGB}{255,80,48}
    \definecolor{Pantina}{RGB}{0,126,142}  
    \definecolor{Harvest}{RGB}{252,179,22}  
    \definecolor{Industrial}{RGB}{29,88,167}  
    \definecolor{Arches}{RGB}{0,159,212} 
    \definecolor{Berry}{RGB}{92,14,65}  
    \definecolor{Storm1}{RGB}{113,115,114} 
    \definecolor{Storm2}{RGB}{156,154,157} 
    \definecolor{Storm3}{RGB}{200,198,199} 

    \colorlet{ratecnst}{Storm1}

\usepackage[margin=1.5cm]{caption}
    \usepackage{subcaption}
    
\usepackage{tabularx}
    \newcolumntype{C}{>{\centering\arraybackslash}X}

\usepackage[linkbordercolor=cyan, pdfborder={0 0 0}]{hyperref}
    \hypersetup{
        colorlinks=true,
        linkcolor=Industrial,
        filecolor=Industrial,      
        urlcolor=Industrial,
        citecolor=Industrial,
    }

\usepackage[titletoc,title]{appendix}
    \usepackage{titletoc}
\usepackage[compress]{cite}
\usepackage{bibunits}

\usepackage[nameinlink]{cleveref}
    
    \crefname{ex}{Example}{Examples}
    \crefname{thm}{Theorem}{Theorems} 
    \crefname{lem}{Lemma}{Lemmas}
    \crefname{prop}{Proposition}{Propositions}
    \crefname{cor}{Corollary}{Corollaries} 
    \crefname{conj}{Conjecture}{Conjectures} 
    \crefname{defn}{Definition}{Definitions}
    \crefname{rmk}{Remark}{Remarks} 
    \crefname{equation}{Eq.}{Eqs.} 
    \crefname{figure}{Fig.}{Figs.}

    \Crefname{sisection}{Appendix}{Appendix}
    \Crefname{sisubsection}{Appendix}{Appendix}
    \Crefname{sifig}{\scshape{Fig.}}{\scshape{Figs.}}
    \Crefname{sitab}{\scshape{Table}}{\scshape{Tables}}
    \Crefname{sieq}{Eq.}{Eqs.}
    
	\newtheorem{thm}{Theorem}[section]
	\newtheorem{lem}[thm]{Lemma}
	\newtheorem{prop}[thm]{Proposition}
	\newtheorem{cor}[thm]{Corollary}
	
	\newtheorem*{thm*}{Theorem}
	\newtheorem*{cor*}{Corollary}
	\theoremstyle{definition} 
		\newtheorem{defn}[thm]{Definition}
		\newtheorem{ex}[thm]{Example}
    	\newtheorem{rmk}[thm]{Remark}	
    \newcommand{\df}[1]{{\bf\emph{{#1}}}}

\newcommand{\eq}[1]{\begin{align*}#1\end{align*}}
	\newcommand{\eqn}[1]{\begin{align}#1\end{align}}  
\newcommand{\ds}{\displaystyle}	
\renewcommand{\epsilon}{\varepsilon}	
\renewcommand{\phi}{\varphi}
  
\newcommand\mc[1]{\mathcal{#1}}

\newcommand\mrm[1]{\mathrm{#1}}
 
\newcommand{\rr}{\ensuremath{\mathbb{R}}}

\DeclareMathOperator{\ran}{ran}
\DeclareMathOperator{\sgn}{sign}

\DeclareMathOperator{\diag}{diag}
\DeclareMathOperator{\Tr}{tr}

\newcommand\abs[1]{\left|{#1}\right|}	
\newcommand\norm[1]{\left\lVert{#1}\right\rVert}

\newcommand{\vv}[1]{{\boldsymbol{#1}}} 
    \newcommand{\xx}{\vv x}
\newcommand{\ehat}{\ensuremath{\hat{\mathbf{e}}}}

\newcommand{\mm}[1]{\mathbf{#1}} 
\newcommand{\rrp}{\rr_{\geq 0}}
\newcommand{\rrpp}{\rr_{>0}}

\usepackage{bbm}  
	\newcommand{\ones}{\mathbbm{1}}
\newcommand{\kk}{\ensuremath{\kappa}}

\newcommand{\deriv}[2]{\partial_{#2}{#1}}
\newcommand{\neutedge}{
    {\protect\tikz\protect
        \draw[EdgeShift={0pt}{-1pt}, -{Circle[open=false, length=1.5mm]}, thick]  
        (0,0) node[outer sep=1pt] {} -- 
        (0.45,0) node[outer sep=1pt] {};
    }}
\newcommand{\posedge}{
    {\protect\tikz \protect
        \draw[EdgeShift={0pt}{-0pt}, -stealth, thick]  
        (0,0) node[outer sep=1pt] {} -- 
        (0.4,0) node[outer sep=1pt] {};
    }}
\newcommand{\negedge}{
    {\protect\tikz \protect
        \draw[EdgeShift={0pt}{-0.5pt}, -{|[width=2.5mm]}, thick]  
        (0,0) node[outer sep=1pt] {} -- 
        (0.4,0) node[outer sep=1pt] {};
    }}
    \newcommand{\ppFBL}{\ensuremath{
        \mathrel{%
            \begin{tikzpicture}[baseline=-1.2ex]
            \draw[-stealth,thick] (0,0) -- (0.4,0);
            \draw[stealth-, EdgeShift={0pt}{-4pt}, thick] (0,0) -- (0.4,0);
            \end{tikzpicture}%
        }%
    }}
    
    \newcommand{\pnFBL}{\ensuremath{%
        \mathrel{%
            \begin{tikzpicture}[baseline=-1.2ex]
            \draw[-stealth, thick] (0,0) -- (0.4,0);
            \draw[{|[width=1.75mm]}-, EdgeShift={0pt}{-4pt}, thick] (0,0) -- (0.4,0);
            \end{tikzpicture}%
        }%
    }}
    
    \newcommand{\neutFBL}{\ensuremath{%
        \mathrel{%
            \begin{tikzpicture}[baseline=-1.2ex]
            \draw[-{Circle[open=false, length=1.5mm]}, thick] (0,0) -- (0.4,0);
            \draw[{Circle[open=false, length=1.5mm]}-, EdgeShift={0pt}{-4pt}, thick] (0,0) -- (0.4,0);
            \end{tikzpicture}%
        }%
    }}
\usepackage{mathtools}
    \newcommand{\FRxn}[1]{
        \xrightharpoonup{#1}
    }
    \newcommand{\RRxn}[2]{
        \xrightleftharpoons[#2]{#1}
    }

\newcommand{\figdraftboxwidth}{
}

\usepackage{multirow}
\usepackage{array}
\usepackage{enumitem}
    \setlist[enumerate]{itemsep=-5pt, topsep=-5pt}
    \setlist[itemize]{itemsep=-8pt, topsep=-8pt}
\usepackage{multicol}

\usepackage{nicematrix}

\newcommand{\uInt}{I}
\newcommand{\X}{\mathbf{X}}
\usepackage{authblk}\usepackage[symbol]{footmisc}

\title{On non-monotonic steady state response curves}

\author[1]{
    Polly Y. Yu%
}\author[2]{
    Eduardo D. Sontag%
}
\affil[1]{\small Department of Mathematics, University of Illinois Urbana-Champaign,
{\tt\small  pollyyu@illinois.edu}}
\affil[2]{\small Departments of Bioengineering and Electrical \& Computer Engineering, Northeastern University,
{\tt\small e.sontag@northeastern.edu}}
\date{}

\begin{document}
\maketitle
\renewcommand*{\thefootnote}{\arabic{footnote}}

\begin{abstract}
    Non-monotonic steady state response curves can readily arise from an incoherent feedforward loop (IFFL), which has also been conjectured to be a necessary mechanism. In a previous work, we showed this conjecture to be false for systems where the control parameter directly influences only one state variable. Here, we extend this result to arbitrary systems. We show that a stable non-monotonic response requires either an IFFL or a combination of positive and negative feedback loops. Finally, we apply our theoretical findings to three network-structured biological contexts: signaling pathways, gene regulation, and ecological consumer-resource models. For the first two contexts, we identify necessary motifs for non-monotonic responses, while for consumer-resource models, we provide an alternative proof of the Competitive Exclusion Principle. 
\end{abstract}

\begin{bibunit}[unsrt]
\section{Introduction}
\label{sec:intro}

Living organisms are constantly responding to their environments, and the magnitude of the response often depends on the signal strength. Traditionally in toxicology and pharmacology, dose-response curves are assumed to be monotonic and sigmoidal, reflecting the intuition that a stronger signal generates a stronger response until saturation. However, non-monotonic responses are also observed in biology, e.g., the effect of metal ions on rabbit myosin activity~\cite{NimuraMiuraShinobuImura1987}, phytochemicals on tumor cell proliferation~\cite{Jodynis-LiebertKujawska2020}, and antigens on T-cells activation~\cite{LeverLimKrugerNguyenEtAl2016}. In pharmacology, \emph{biphasic}%
    \footnote{\lq\lq Biphasic\rq\rq\ response normally refers to a two-phase non-monotonicity, i.e., either a bell-shaped or inverted bell-shaped response. Throughout this paper, we use the term biphasic to mean non-monotonicity, with no implication on the number of extrema.
    } 
response have clinical implications, as the optimal dosage may occur at an unexpected point. We are interested in uncovering the network architecture behind non-monotonic responses.

Many biological processes are represented as networks, analogous to wiring diagrams in engineering~\cite{TysonChenNovak2003}. Particularly in cellular signaling and gene regulatory networks, which are often presented as signed directed graphs, certain network motifs are often associated with specific dynamical behaviors. Among these motifs, incoherent feedforward loops (IFFLs) have been proposed as the primary mechanism responsible for biphasic responses~\cite{KimKwonCho2008, LeverLimKrugerNguyenEtAl2016}.

In our preliminary conference paper~\cite{YuSontag-biphasic1}, we showed that a stable biphasic response does not necessarily require an IFFL, but can also be generated by a combination of positive and negative feedback loops (PFBL and NFBL). However, we were only able to prove this for systems where the independent parameter $u$ only appears in a single equation, i.e., $\dot{\vv x} = \vv f(\vv x) + \ehat_1 g(u)$. 

In this work, we extend our findings to an arbitrary autonomous system $\dot{\xx} = \vv f(\xx; u)$ with a single control parameter $u$. We establish that a necessary condition for stable biphasic response remains either an IFFL, or a combination of PFBL and NFBL, where the PFBL is vertex-disjoint from an input--output path and the NFBL is downstream of the input and upstream of the output. We apply our developed framework to three biological contexts where network representations naturally arise. 
First, we revisit a model of cellular signaling in the activation of T-cells~\cite{LeverLimKrugerNguyenEtAl2016}. We also give an explicit example with no IFFLs that exhibits biphasic response; this definitively settles the question posed in \cite{LeverLimKrugerNguyenEtAl2016}. Next, we consider transcriptional regulation models, where activation and inhibition come from activatory/inhibitory Hill functions. 
Finally, stepping away from steady state responses, we give an alternative proof of the Competitive Exclusion Principle for consumer-resource models.

\subsection{Outline}
\label{sec:outline}

In \Cref{sec:prelim}, we lay out the minimal assumptions for the system under consideration and define its \emph{J-graph}, a signed directed graph based on the Jacobian matrix. Standing independently, \Cref{sec:symbolic_minor} contains technical lemmas from combinatorial matrix theory that identify graph-theoretic conditions for when a minor of a signed symbolic matrix contains terms of opposite signs. In \Cref{sec:main-results}, we turn to the steady state response curve $x_n^*(u)$, and prove our main results on stable non-monotonic response. Finally, \Cref{sec:application} contains applications, and we conclude with a discussion in \Cref{sec:conclusion}.

\subsection{Notation}
\label{sec:notation}

The following notations are used throughout this paper.
\begin{itemize}
    \item $g(s) \equiv 0$ means $g$ is the identically zero function on its domain; 
    \item $[n] = \{ 1,2,\ldots, n\}$, 
    \item $\hat{i} = [n]\setminus\{i\}$ for any $i \in [n]$, 
    \item $\ehat_i$ is the $i$th standard orthonormal vector of $\rr^n$,  
    \item  $\mm J(\gamma, \delta)$ is the submatrix of the matrix  $\mm J$ with rows indexed by $\gamma$ and columns by $\delta$, e.g., $\mm J(\hat i, \hat j)$ is submatrix obtained by deleting row $i$ and column $j$,
    \item $\mm J[\gamma, \delta] = \det \mm J(\gamma, \delta)$,  
    \item $\mm J(\gamma) = \mm J(\gamma, \gamma)$ and $\mm J[\gamma] = \mm J[\gamma, \gamma]$,
    \item $S_\gamma$ is the permutation group on the set $\gamma$, with $S_n = S_{[n]}$, 
    \item $\mm A \odot \mm B$ is the component-wise multiplication of two matrices of the same size. 
\end{itemize}

We say a continuous function $g(s)$ is \emph{strictly non-negative} (respectively \emph{strictly non-positive}) if and only if $g(s) \geq 0$ (respectively $g(s) \leq 0$) for all $s$, and there exists $s_0$ such that $g(s_0) \neq 0$. We say $g(s)$ is \emph{non-vanishing} if $g(s) \neq 0$ for all $s$.

\section{Preliminaries}
\label{sec:prelim}

Consider the autonomous system 
\eqn{
\label{eq:syst_main}
    \dot{\vv x} = \vv f(\xx;  u), \quad \xx(0) = \vv a \in \rr^n, 
}
with control parameter $u$ and output $x_n$. For a continuously differentiable branch of steady states $\xx^*(u)$, we seek conditions under which $x_n^*(u)$ is non-monotonic, that is, when the local sensitivity $d_u{x_n^*(u)}$ is not sign-definite. 

Suppose $\vv f \colon \X \times \uInt  \to \rr^n$ is $C^1$ on $(\xx,u) \in \X \times \uInt$, where $\X \subseteq \rr^n$ is an open non-empty set and $\uInt\subseteq \rr$ is an open non-empty interval. Denote by $\phi(t, \vv a, u)$ the solution to \eqref{eq:syst_main}. 
We make the following standing assumption:
\begin{align*}
    \forall u_0 \in \uInt \,\, 
    \exists \vv x_0\in \X 
        \text{ such that } 
        \vv f(\vv x_0; u_0) = \vv 0  
    \text{ and } 
    D_{\xx}\vv f (\vv x_0; u_0) \text{ has full rank}.
\end{align*}
By the Implicit Function Theorem, this implies: 
\begin{align}
\begin{gathered}
\label{eq:assumption-A0}
    \forall u_0 \in \uInt \,\, 
    \exists \text{ an open interval } \uInt_0 \subseteq \uInt 
        \text{ containing } u_0 \,\, 
    \exists \text{ a } C^1 \text{ function }  
        \xx^* \colon \uInt_0 \to \X 
        \\
        \text{ such that }
    \forall u \in \uInt_0  
        \text{ we have }
        \vv f(\xx^*(u); u) = \vv 0 
        \text{ and } 
        D_{\xx}\vv f (\vv x^*(u); u) \text{ has full rank},
\end{gathered}\tag{A0}
\end{align}
where $\xx^*$ is the locally unique branch with $\xx^*(u_0) = \xx_0$.

While \eqref{eq:assumption-A0} guarantees the existence of a steady state branch, it does not ensure stability. When stability is required, we assume the following property:
\begin{align}
\begin{gathered}
\label{eq:assumption-A1}
    \text{For any } \xx^* \text{ and } \uInt_0 \text{ as in \eqref{eq:assumption-A0}, }
    \forall u \in \uInt_0 \,\, 
    \exists \text{ an open ball } B_{\epsilon(u)}(\xx^*(u)) \subseteq \X 
        \\
        \text{ for some } \epsilon(u) > 0 
        \text{ such that } 
    \phi(t,\vv a, u)  \xrightarrow{t\to\infty} \xx^*(u)
        \text{ for any } \vv a \in B_{\epsilon(u)}(\xx^*(u)). 
        \\ 
    \text{ Moreover, the mapping } u \mapsto \epsilon(u) \text{ is continuous}. 
\end{gathered} \tag{A1}
\end{align}
Alternatively, we may assume the stronger and more easily checked property: 
\begin{align}
\begin{gathered}
\label{eq:assumption-A2}
    \forall u_0 \in \uInt \,\, 
    \exists \vv x_0\in \X 
        \text{ such that } 
        {\vv f}(\vv x_0; u_0) = \vv 0  
        \text{ and } 
        D_{\xx}\vv f (\vv x_0; u_0) \text{ is Hurwitz}.  
\end{gathered}\tag{A2}
\end{align}
In \Cref{sec:A2-implies-A1}, we prove that \eqref{eq:assumption-A2} implies \eqref{eq:assumption-A1}, which guarantees that for any $u$ sufficiently close to $u_0$, we have $\xx^*(u_0) \in B_{\epsilon(u)}(\xx^*(u))$, ensuring that $\phi(t, \xx^*(u_0), u) \xrightarrow{t\to\infty} \xx^*(u)$. In other words, if the system is initially at steady state $\xx^*(u_0)$, after a small perturbation in $u$, the system converges to the new steady state $\xx^*(u)$. See \Cref{fig:A1-implies-convergence} for a schematic illustration and \Cref{sec:A1-implies-convergence} for a proof of convergence.

\begin{figure}[h!t]
\centering
\begin{tikzpicture}
\clip (-3, -2.75) rectangle (3.7, 3);
    \draw[fill=Arches!10, very thick, dashed] (0,0) ellipse (0.35 and 1.18);
    \node at (1,1.5) [left] {\footnotesize $B_{\epsilon(u_1)}(\xx^*(u_1))$}; 
        
    \begin{scope} 
        \clip (-2, -2.2) rectangle (3.7, 3);
        \draw[ultra thick, Orange] (-2,-2) .. controls (-1.5,-1.8) and  (-1,-1.5).. (0,0) .. controls (1,1) .. (2.4, 1.7) node [right] {$\xx^*(u)$};
    \end{scope}
    \filldraw[black] (0,0) circle (2pt);
    
    \draw[black, -{Stealth[scale=1.25]}] (-0.7,-0.97)-- (-0.1,-0.97); 
    \filldraw[Industrial] (-0.7,-0.97) circle (1.5pt); 
    \node [Industrial] at (-0.7,-0.97) [left] {\footnotesize $\xx^*(u_0)$};
    \draw[Industrial, thick, -stealth] 
        (-0.1,-0.97) to[out=10, in=320, looseness=0.5]
        (0.05,0.6) to[out=130,in=80, looseness=1] 
        (-0.2,0.18) to[out=-100,in=160, looseness=0.65] 
        (-0.1,-0.45) to[out=-20,in=-80, looseness=1]
        (0.05,-0.1); 
        
    \begin{scope}
        \draw[thin, ->, EdgeShift={0pt}{-65pt}] (-2.5,0)--(3,0) node [right] {$u$};
        \draw[thin, ->, EdgeShift={0pt}{-65pt}] (-2.5,0)--(-2.5,5) node [left] {$\xx$};
        \draw[thin] (-0.7,-2.4)--(-0.7,-2.2); 
        \draw[thin] (0,-2.4)--(0,-2.2);
        \node[EdgeShift={0pt}{-67pt}] (u1) at (-0.7,0) [below] {\footnotesize $u_0$};
        \node[EdgeShift={0pt}{-67pt}] (u2) at (0,0) [below] {\footnotesize $u_1$};
    \end{scope} 
\node at (0,-2.5) {};
\end{tikzpicture}
\caption{Assumption~\eqref{eq:assumption-A1} guarantees that following a small perturbation in $u$, the previous steady state $\xx^*(u_0)$ lies within the basin of attraction of the new steady state $\xx^*(u_1)$.}
\label{fig:A1-implies-convergence}
\end{figure}
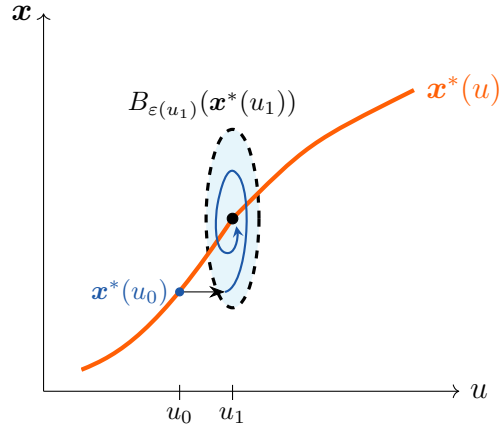

We also make several assumptions regarding the Jacobian matrix $\mm J(\xx; u) = D_{\xx} \vv f(\xx; u)$ of \eqref{eq:syst_main}. First, we assume that each entry of $\mm J(\xx; u)$ has a fixed sign on $\X \times \uInt_0$, i.e., each partial derivative $\deriv{f_j}{i}$ is either identically zero, strictly non-negative, or strictly non-positive. Moreover, the diagonal entries of $\mm J$ are assumed to be strictly non-positive. We also assume for each $i \in [n]$ that $\deriv{f_i}{u}$ is either identically zero, strictly non-negative, or strictly non-positive, with at least one such derivative not identically zero.

We refer to $u$ as the control parameter and treat it as the \emph{input} of the system%
    \footnote{Here, we treat the control parameter $u$ as the input. This contrasts with \cite{deOliveiraMadeiraAntoneli2024, GolubitskyStewartAntoneliHuangWang2020, YuSontag-biphasic1, AraujoLiotta2018}, where the \lq\lq input\rq\rq\ is any state variable $x_i$ for which $\deriv{f_i}{u} \not\equiv 0$.
    }%
. Without loss of generality, let $x_n$ be the \emph{output}. Then the system \eqref{eq:syst_main} is an \emph{input--output system}.  
The \df{(steady state) response} of the system is $x^*_n(u)$, while its derivative $d_ux^*_n(u)$ is the \df{(local) sensitivity}.

\subsection{Graph representation of an input--output system}
\label{sec:graph_of_syst}

We represent the input--output system \eqref{eq:syst_main} as a signed directed graph $G$ with $n+1$ nodes: one for each state variable $x_i$ and one for the control parameter $u$. The edges of $G$ encode the signs of the entries in the Jacobian matrix $\mm J = D_{\xx} \vv f$ and the vector $\deriv{\vv f}{u}$. We call $G$ the \emph{J-graph}, to emphasize its origin in the Jacobian matrix.

Numerous other papers have represented ODE systems as signed directed graphs, the so-called \emph{influence graph}, see \cite{Soule2003, KaufmanSouleThomas2007, Banaji2010, GolubitskyStewartAntoneliHuangWang2020} for example. The key difference between the J-graph in this work and those in some of earlier formulations, including our own~\cite{YuSontag-biphasic1}, is the \emph{explicit} inclusion of the control parameter $u$ as a distinct node. This seemingly innocuous change greatly simplifies the proofs. 

Formally, the \df{J-graph} of \eqref{eq:syst_main} is defined as follows. A positive directed edge exists from $x_i$ to $x_j$ (denoted $i \posedge j$) if and only if $\deriv{f_j}{i}$ is strictly non-negative, i.e., $\deriv{f_j}{i} \geq 0$ and not identically zero. A negative edge ($i \negedge j$) exists if and only if $\deriv{f_j}{i}$ is strictly non-positive. Similarly, a directed edge exists from $u$ to $x_j$ if and only if $\deriv{f_j}{u}$ is non-zero, with its sign determined by the sign of $\deriv{f_j}{u}$. To indicate that an edge exists from $x_i$ to $x_j$ without specifying its sign, we use the notation $i \neutedge j$. Finally, to designate $x_n$ as the output, we draw an outgoing half-edge from its node. Note that under the assumption of strictly non-positive diagonal, the J-graph of \eqref{eq:syst_main} contains negative self-loops at every state node, which we may omit from drawing for simplicity. 

More generally, the J-graph can be defined for any square signed matrix~\cite{vandenDriessche2018,BrualdiRyser1991}. Let $\mbc A$ be an $m \times m$ \emph{signed symbolic matrix}, i.e., $[\mbc A]_{j,i} \in \{0, a_{ji}, -a_{ji}\}$ where each $a_{ji}$ is a positive indeterminate. The \df{J-graph of $\mbc A$} has $m$ nodes, containing a positive edge $i \posedge j$ (respectively negative edge $i \negedge j$) if and only if $[\mbc A]_{j,i} = a_{ji}$ (respectively $[\mbc A]_{j,i} = -a_{ji}$). In this manner, $\mbc A$ serves as the adjacency matrix of its J-graph. 

We identify a subgraph of $G = (V,E)$ either by a subset of edges together with all incident vertices, or as an induced subgraph defined by a subset of vertices. In the latter case, say given any $\gamma\subseteq V$, the induced subgraph $G[\gamma]$ has vertex set $\gamma$ and edge set $\{(i,j) \in E \colon i,j \in \gamma\}$. In other words, $G[\gamma]$ is the subgraph obtained by keeping the nodes in $\gamma$ and all edges between them. 

More generally, given any $\gamma, \delta \subseteq V$, not necessarily disjoint, we define $G[\gamma, \delta]$ to be the subgraph containing all edges from $\delta$ to $\gamma$, that is, $G[\gamma, \delta]$ has vertex set $\gamma \cup \delta$, and edge set $\{(i,j) \in E \colon i  \in \delta \text{ and } j \in \gamma \}$. Of particular relevance is $G[\hat{i}, \hat{j}]$, obtained by removing all edges leaving $j$ and all edges entering $i$. See \Cref{ex:symbolic_graph} and \Cref{fig:symbolic_graph} for an example.

Our notation for induced subgraph is chosen to be consistent with indexing in the Jacobian matrix, where each column records the outgoing edges from a given node, and rows record incoming edges. More generally, induced subgraphs directly correspond to submatrices. Given any $\gamma, \delta \subseteq [m]$, the induced subgraph $G[\gamma]$ corresponds to the principal submatrix $\mbc A(\gamma)$, while $G[\gamma, \delta]$ corresponds to $\mbc A(\gamma, \delta)$. In particular, $G[\hat{i}, \hat{j}]$ corresponds to the submatrix $\mbc A(\hat{i}, \hat{j})$ obtained by deleting row $i$ and column $j$ from $\mbc A$.

\begin{figure}[th]
\centering 
\begin{subfigure}[t]{0.3\textwidth}
    \centering\figdraftboxwidth
    \includegraphics[]{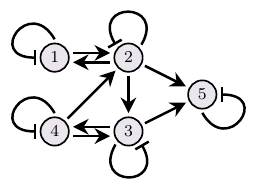}
    \caption{}
    \label{fig:symbolic_graph-a}
\end{subfigure}
\begin{subfigure}[t]{0.3\textwidth}
    \centering\figdraftboxwidth
    \includegraphics[]{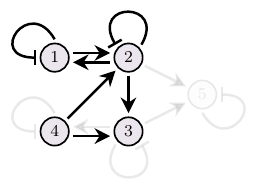}
    \caption{}
    \label{fig:symbolic_graph-b}
\end{subfigure}
\begin{subfigure}[t]{0.3\textwidth}
    \centering\figdraftboxwidth
    \includegraphics[]{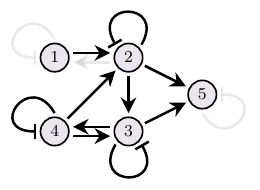}
    \caption{}
    \label{fig:symbolic_graph-c}
\end{subfigure}
\caption{
    (a) The J-graph of the symbolic matrix $\mbc A$ from \Cref{ex:symbolic_graph}. 
    (b) The induced subgraph $G[\gamma, \delta]$ where $\gamma = \{1,2,3\}$ and $\delta = \{1,2,4\}$. 
    (c) The induced subgraph $G[\hat 1, \hat 5]$.
}
\label{fig:symbolic_graph}
\end{figure}

\begin{ex}
\label{ex:symbolic_graph}
    Consider the signed symbolic matrix 
    \eq{
        \mbc A = 
        \begin{pmatrix*}[r]
            -a_{11} & a_{12} & 0 & 0 & 0\\
            a_{21} & -a_{22} & 0 & a_{24} & 0 \\ 
            0 & a_{32} & -a_{33} & a_{34} & 0 \\ 
            0 & 0 & a_{43} & -a_{44} & 0  \\ 
            0 & a_{52} & a_{53} & 0 & -a_{55}
        \end{pmatrix*}, 
    }
    where each $a_{ij}$ is a positive indeterminate. Its J-graph $G$ is shown in \Cref{fig:symbolic_graph-a}. 
    
    Let $\gamma = \{1,2,3\}$ and $\delta = \{1,2,4\}$. The induced subgraph $G[\gamma, \delta]$, shown in \Cref{fig:symbolic_graph-b}, is obtained by keeping only the edges that originate from a node in $\delta$ and terminating at a node in $\gamma$. This subgraph is associated with the submatrix 
    \eq{
        \mbc A(\gamma, \delta) = \begin{pmatrix*}[r]
            -a_{11} & a_{12}  & 0 \\
            a_{21} & -a_{22}  & a_{24}  \\ 
            0 & a_{32}  & a_{34}  
        \end{pmatrix*},
    }
    obtained by deleting from $\mbc A$, rows indexed by $\gamma^c$ and columns indexed by $\delta^c$. The induced subgraph $G[\hat 1, \hat 5]$,  shown in \Cref{fig:symbolic_graph-c} and associated with the submatrix $\mbc A(\hat 1, \hat 5)$, is obtained by deleting all edges that leave the node $5$ and all edges that enter the node $1$.
\end{ex}

Next, we review terminology related to signed directed graphs. A \df{walk} is a non-empty sequence of directed edges connecting a sequence of vertices. A \df{path} is a simple walk, i.e., a walk with no repeating vertices. A \df{cycle} is a closed path. The \df{length} of a path or a cycle is the number of edges in it. The \df{sign} of a subset of edges $E'$ is defined to be $\sgn(E') = \prod_{e \in E'} \sgn(e)$.

A subset of nodes $\chi$ is said to be \df{downstream of the node $i$} if there exists a walk from $i$ to $v$ for every $v \in \chi$. It is said to be \df{upstream}%
    \footnote{\lq\lq Upstream\rq\rq\ and \lq\lq downstream\rq\rq\ do not in general define a partial ordering of the nodes, only a pre-ordering.} 
\df{of the node $j$} if there is a walk from every node in $\chi$ to $j$.  Similarly, a subset of edges is downstream (respectively upstream) of a node $v$ if the set of incident nodes is downstream (respectively upstream) of $v$. 
For example, in \Cref{fig:symbolic_graph-a}, the cycle $3 \ppFBL 4$ is downstream of $1$, but not of $5$.

A \df{feedback loop} is a cycle of length two or more. It is a \df{positive feedback loop} (\df{PFBL}) if its sign is positive; otherwise it is a \df{negative feedback loop} (\df{NFBL}). A \df{feedforward loop from $i$ to $j$} is a pair of distinct paths $P_1 \neq P_2$, both originating from node $i$, and terminating at node $j \neq i$. We allow $P_1$ and $P_2$ to share some common internal nodes and edges. Without reference as to where it starts or ends, an IFFL is by default from input to output. A feedforward loop is said to be \df{coherent} (\df{CFFL}) if $\sgn(P_1) = \sgn(P_2)$; otherwise it is \df{incoherent} (\df{IFFL}). For example, in \Cref{fig:symbolic_graph-a}, $1 \posedge 2 \posedge 5$ and $1 \posedge 2 \posedge 3 \posedge 5$ form a CFFL. An \df{input--output path} (\df{I/O path}) is a path from the input node to the output node. When $G$ is the J-graph of \eqref{eq:syst_main}, we take $u$ as input, and $n$ as output.

A \df{vertex-disjoint cycle cover} of $G$ (or \df{disjoint cycle cover} for short) is a subgraph $H$ consisting of cycles, such that every node of $G$ lies in exactly one cycle, and no two cycles share a common vertex. Some cycles in $H$ may be self-loops. Two example disjoint cycle covers of the graph in \Cref{fig:symbolic_graph-a} are shown in \Cref{fig:symbolic_cover-a,fig:symbolic_cover-b}. 
In the case of the J-graph of a symbolic matrix $\mbc A$, it is well known that each disjoint cycle cover (also called a \emph{Hamiltonian hooping}~\cite{Soule2003}) corresponds to a non-zero term in the expansion of $\det(\mbc A)$; see \Cref{lem:nonzeroterms-subgraphs} or \cite{Soule2003, Harary1962}.

Each non-zero term in a non-principal first minor corresponds to a different kind of vertex-disjoint graph covering (\Cref{fig:symbolic_cover-c}). 

\begin{defn}
\label{def:IO-cover}
    Let $G$ be the J-graph with input node $i$ and output node $j$. A subgraph $H$ is a \df{disjoint I/O cover} if it is a vertex-disjoint collection of an I/O path and cycles, such that together they cover all the nodes of $G$. 
\end{defn}

\begin{figure}[t!]
\centering 
\begin{subfigure}[t]{0.3\textwidth}
    \centering\figdraftboxwidth
    \includegraphics[]{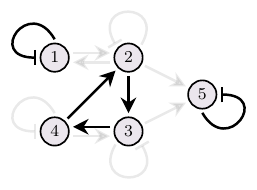}
    \caption{}
    \label{fig:symbolic_cover-a}
\end{subfigure}
\begin{subfigure}[t]{0.3\textwidth}
    \centering\figdraftboxwidth
    \includegraphics[]{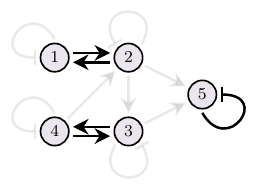}
    \caption{}
    \label{fig:symbolic_cover-b}
\end{subfigure}
\begin{subfigure}[t]{0.37\textwidth}
    \centering\figdraftboxwidth
    \includegraphics[]{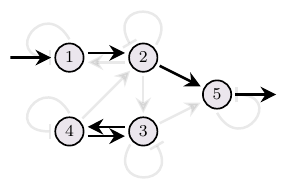}
    \caption{}
    \label{fig:symbolic_cover-c}
\end{subfigure}
\caption{
    (a-b) Two disjoint cycle covers of the J-graph from \Cref{fig:symbolic_graph-a}. 
    (c) A disjoint I/O cover with $1$ as input and $5$ as output.
}
\label{fig:symbolic_graph-cover}
\end{figure}

\section{Subgraphs and minors of a signed symbolic matrix}
\label{sec:symbolic_minor}

In this section, we collect several results concerning minors of a signed symbolic matrix and its J-graph. They culminate in \Cref{cor:symbolic-minor-neg_diag}, which states that for a  matrix with negative diagonal, either an IFFL or a PFBL disjoint from an I/O path is necessary and sufficient for a non-principal first minor to contain two terms of opposite signs.

Mixed signs in minors have been investigated in combinatorial matrix theory~\cite{BassettMaybeeQuirk1968,FiedlerPtak1966, BrualdiRyser1991} and \emph{sign-solvability} of a linear system $\mm A \xx = \vv b$, particularly when the matrix has negative diagonal, see, e.g., \cite[Chapter~3.2]{BrualdiShader1995}. The motivation here is to ask given a solution to $\mm A \xx = \vv b$, can the sign of $x_j$ be determined by the sign pattern of $\mm A$ and $\vv b$ alone. Since generically linear equations can be solved using Cramer's rule, the question reduces to one about cofactors. 

Note that \Cref{cor:symbolic-minor-neg_diag} is a combination of Theorem~3.2.1 and Corollary~3.2.3 in \cite{BrualdiShader1995}. However, for our last application in \Cref{sec:application-CRM}, the diagonal of the Jacobian graph is non-positive. Since the proofs are relatively elementary, and the insight connecting the matrix with its J-graph are important, we derive the results from the beginning, with fewer constraints on the diagonal entries.

It is worth mentioning that a key assumption throughout this section is that entries of the matrix are algebraically independent. Thus, there is no cancellation when computing minors. Clearly this is false for general ODE systems (e.g., \Cref{ex:PA-suff_notNecessary}); nonetheless, these algebraic results can provide necessary conditions for the non-monotonicity of response curves.

Recall for any permutation $\sigma \in S_m$, the \df{disjoint cycle decomposition} is the decomposition of $\sigma$ into disjoint cycles that partition $[m]$. Its \df{non-trivial cycle decomposition} drops the fixed points. For example, $\sigma = (1)(24)(356)(7) \in S_7$ is a disjoint cycle decomposition, while $(24)(356)$ is its non-trivial cycle decomposition. By a slight abuse of notation, we identify cycles with their underlying subsets (e.g., $(356)$ with $\{3,5,6\}$). 

\begin{lem}
\label{lem:nonzeroterms-subgraphs}
    Let $\mbc A$ be an $m \times m$ signed symbolic matrix with J-graph $G$. Let $\gamma \neq \delta \subseteq [m]$ where $|\gamma| = |\delta|$, and let $\Phi \colon \delta^c \to \gamma^c$ be a bijection.
    \begin{enumerate}[label={(\alph*)}]
        \item\label{lem:nonzeroterms-subgraphs-a} The collection of non-zero terms in the principal minor $\mbc A[\gamma]$ is in bijection with the collection of disjoint cycle covers of $G[\gamma]$. 

        \item\label{lem:nonzeroterms-subgraphs-b}  The collection of non-zero terms in the non-principal minor $\mbc A[\gamma, \delta]$ is in bijection with the collection of disjoint cycle covers of $G[\gamma, \delta] \cup \{ i \posedge \Phi(i) \colon i \in \delta^c \}$. 
        
        \item\label{lem:nonzeroterms-subgraphs-c} Let $i$ be the input, and $j \neq i$ be the output. The collection of non-zero terms in the non-principal first minor $\mbc A[\hat{i}, \hat{j}]$ is in bijection with the collection of disjoint I/O covers of $G$.
    \end{enumerate}
\end{lem}
\begin{proof}
\begin{enumerate}[label={(\alph*)}, leftmargin=0pt, itemindent=21pt]
\item 
    Without loss of generality, let $\gamma = [m]$. By the Leibniz formula, $\det(\mbc A) = \sum_{\sigma \in S_m} T_\sigma$, where the monomial $T_\sigma = \sgn(\sigma) \prod_{i=1}^m [\mbc A]_{\sigma(i), i}$ is non-zero if and only if $[\mbc A]_{\sigma(i),i} \neq 0$ for all $i$, i.e., $i \neutedge \sigma(i) \in G$. 
    For any $\sigma \in S_m$ such that $T_\sigma \neq 0$, let $\sigma = \tau_1\cdots \tau_q$ be its disjoint cycle decomposition. Each $\tau_j$ is associated to a cycle $C_j$ in $G$:  any $\tau_j = (k)$ is associated to the self-loop on the node $k$, and any $\tau_j = (k_1k_2\cdots k_\ell)$ with $\ell \geq 2$  is associated to the feedback loop ${k_1} \neutedge \cdots \neutedge {k_\ell} \neutedge {k_1}$. Because the disjoint cycle decomposition partitions $[m]$, the collection $\{ C_j\}_{j=1}^q$ forms a disjoint cycle cover of $G$. 

\item 
    Let $\delta^c = \{ \ell_1 , \cdots , \ell_r\}$, and $\gamma^c = \{ k_j = \Phi(\ell_j) \colon \ell_j \in \delta^c\}$. Consider the matrix $\mbc B$ where 
    \eq{
        [\mbc B]_{ji} = \left\{
        \begin{array}{cl}   
            {[\mbc A]}_{ji} & \text{ if } i \in \delta, \, j \in \gamma,  \\
            1 & \text{ if } i \in \delta^c,\,  j = \Phi(i) , \\
            0 & \text{ otherwise}. 
        \end{array}
        \right.
    }
    Expanding along the columns in $\delta^c$, we arrive at $\det(\mbc B) = (-1)^{\sum_j (k_j + \ell_j)} \mbc A[\gamma, \delta]$, so the non-zero terms in $\mbc A[\gamma, \delta]$ are in a one-to-one correspondence with those in $\det(\mbc B)$. 
    The J-graph of $\mbc B$ is 
    $G[\gamma, \delta] \cup \{ i \posedge \Phi(i) \colon i \in \delta^c \}$. The out-degree of $\ell_j \in \delta^c$ is one, so every cycle that contains $\ell_j$ must contain the edge $\ell_j \neutedge k_j$. Then \ref{lem:nonzeroterms-subgraphs-b} follows from \ref{lem:nonzeroterms-subgraphs-a}. 

\item     
    This is immediate from \ref{lem:nonzeroterms-subgraphs-b}, because every disjoint cycle cover of $G[\hat i, \hat j] \cup \{ j \posedge i\}$ is a vertex-disjoint collection of cycles, where exactly one of them is $P \cup \{j \posedge i\}$ with $P$ being a path from $i$ to $j$. 
    \qedhere 
\end{enumerate}
\end{proof}

\begin{rmk}
\label{rmk:cycles-permutation-cover}
    The main idea behind \Cref{lem:nonzeroterms-subgraphs}, which we return to repeatedly, is the connection between a non-zero term in a minor (e.g., $T_\sigma$ in $\det(\mbc A)$ for some $\sigma \in S_m$) and a graph cover of $G$ (e.g., a disjoint cycle cover $E_\sigma$). The cycles in the disjoint cycle decomposition of $\sigma \in S_m$ correspond to the connected components of $E_\sigma$. In the case of a non-principal first minor, we require  $\sigma(j) = i$. 
\end{rmk}

\begin{figure}[h!t]
\centering 
\begin{subfigure}[t]{0.22\textwidth}
    \centering\figdraftboxwidth
    \raisebox{0.65cm}{\includegraphics{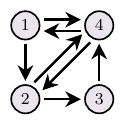}}
    \caption{}
    \label{fig:ex:cycles-permutation-cover-principal-Jgraph}
\end{subfigure}
\begin{subfigure}[t]{0.375\textwidth}
    \centering\figdraftboxwidth
    \includegraphics{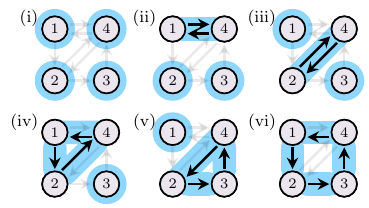}
    \caption{}
    \label{fig:ex:cycles-permutation-cover-principal-covers}
\end{subfigure}
\hspace{0.45cm}
\begin{subfigure}[t]{0.25\textwidth} 
    \centering\figdraftboxwidth
    \includegraphics{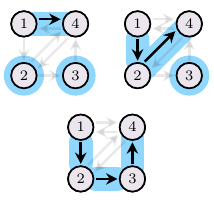}
    \caption{}
    \label{fig:ex:cycles-permutation-cover-nonprincipal-covers}
\end{subfigure}
\caption{
    (a) The J-graph $G$ of the matrix $\mbc A$ from \Cref{ex:cycles-permutation-cover}, with self-loops omitted for simplicity.
    (b) The six disjoint cycle covers of $G$, corresponding to the six non-zero terms in $\det(\mbc A)$. 
    (c) The three disjoint I/O covers corresponding to the three non-zero terms in $\mbc A[\hat 1,\hat 4]$. 
}
\label{fig:ex:cycles-permutation-cover}
\end{figure}

\begin{ex}
\label{ex:cycles-permutation-cover}
    Consider the signed symbolic matrix 
    \eq{
        \mbc A = \begin{pmatrix*}[r]
            -a_{11} & 0 & 0 & a_{14} \\ 
            a_{21} & -a_{22} & 0 & a_{24} \\ 
            0 & a_{32} & -a_{33} & 0 \\ 
            a_{41} & a_{42} & a_{43} & -a_{44}
        \end{pmatrix*}, 
    }
    whose J-graph is shown in \Cref{fig:ex:cycles-permutation-cover-principal-Jgraph}, with self-loops omitted for simplicity. Its determinant  
        $\det(\mbc A) = a_{11} a_{22} a_{33} a_{44} -a_{14}a_{41} a_{22} a_{33} - a_{24}a_{42}a_{11}a_{33} - a_{21} a_{42} a_{14} a_{33} - a_{32} a_{43} a_{24}  a_{11} -  a_{21} a_{32} a_{43}a_{14}$ 
    contains six terms, each of which corresponds to a disjoint cycle covers by \Cref{lem:nonzeroterms-subgraphs}. These are shown in \Cref{fig:ex:cycles-permutation-cover-principal-covers}. For example, the term $a_{14}a_{41} a_{22} a_{33}$ corresponds to the cycle cover $(14)(2)(3)$, shown in \hyperref[fig:ex:cycles-permutation-cover-principal-covers]{(ii)}. Since $\mbc A$ has negative diagonal, $T_\mrm{id} = (-1)^4 \prod_i a_{ii}$ is a term in $\det(\mbc A$), corresponding to self-loops on all nodes (\hyperref[fig:ex:cycles-permutation-cover-principal-covers]{Fig.~4(b, i)}). 

    Let node $1$ be the input and $4$ be the output. The terms in the non-principal first minor $\mbc A[\hat 1, \hat 4] = a_{22} a_{33} a_{41} + a_{21} a_{33} a_{42} + a_{21} a_{32} a_{43}$  correspond to the disjoint I/O covers in \Cref{fig:ex:cycles-permutation-cover-nonprincipal-covers}. 
\end{ex}

We next turn our attention to the sign of each non-zero term. 

\begin{lem} 
\label{lem:nonzeroterms-sign}
    Let $\mbc A$ be an $m \times m$ signed symbolic matrix with J-graph $G$. 
    \begin{enumerate}[label={(\alph*)}]
        \item \label{lem:nonzeroterms-sign-a}
        Let $\gamma \subseteq [m]$, and $T_\sigma$ be a non-zero term in $\mbc A[\gamma]$. Let $E_\sigma$ be the corresponding disjoint cycle cover described in \Cref{rmk:cycles-permutation-cover}. Then $\sgn(T_\sigma) = (-1)^{|\gamma|+p}$, where $p$ is the number of positive cycles in $E_\sigma$. 
        \item \label{lem:nonzeroterms-sign-b}
        With $i$ as input and $j \neq i$ as output, let $T_\sigma$ be a non-zero term in $\mbc A[\hat i, \hat j]$, corresponding to the disjoint I/O cover $E_\sigma$. Then $\sgn(T_\sigma) = (-1)^{m+i+j-1} (-1)^{p}\sgn(P)$, where $P$ is the I/O path, and $p$ is the number of positive cycles in $E_\sigma$. 
        
    \end{enumerate}
\end{lem} 
\begin{proof}
\begin{enumerate}[label={(\alph*)}, leftmargin=0pt, itemindent=21pt]
\item  
    Without loss of generality, let $\gamma = [m]$. For any $\sigma \in S_m$ such that $T_\sigma \neq 0$ in $\det(\mbc A)$, let $\sigma = \tau_1 \cdots \tau_q$ be its disjoint cycle decomposition, and $C_j$ be the cycle in $E_\sigma$ associated to $\tau_j$ (see \Cref{rmk:cycles-permutation-cover}). Note that some cycles have length $1$. 
    Sorting the Leibniz formula by cycles, the sign of this term is 
    \eq{
        \sgn(T_\sigma)  
        &= \sgn(\sigma) \prod_{i=1}^m \sgn\left( [\mbc A]_{\sigma(i), i}\right) 
        = \prod_{j=1}^q \sgn(\tau_j) \sgn(C_j)
        \\&= \prod_{j=1}^q (-1)^{|\tau_j|+1} \sgn(C_j)
        = (-1)^{m+q} \prod_{j=1}^q \sgn(C_j). 
    }
    Since $\sgn(C_j) = -1$ if and only if $C_j$ is a negative cycle, the sign of $T_\sigma$ is determined by the size of the matrix and the number of positive cycles in $E_\sigma$. 

\item 
    By \Cref{lem:nonzeroterms-subgraphs}\ref{lem:nonzeroterms-subgraphs-c} and \Cref{rmk:cycles-permutation-cover}, the I/O path $P$ in $E_\sigma$ comes from a cycle $C$ in a disjoint cycle cover of $G[\hat{i}, \hat{j}] \cup \{ j \posedge i\}$. Since $\sgn(P) = \sgn(C)$, the desired result follows from \ref{lem:nonzeroterms-sign-a} along with a factor of $(-1)^{i+j}$ from \Cref{lem:nonzeroterms-subgraphs}\ref{lem:nonzeroterms-subgraphs-b}. 
    \qedhere 
\end{enumerate} 
\end{proof}

For the purpose of looking for mixed signs in $\mbc A[\hat i, \hat j]$, note that in \ref{lem:nonzeroterms-sign-b}, $m,i,j$ are fixed. Therefore, we only need to focus on  $(-1)^p \sgn(P)$, the number of PFBLs, and the sign of the I/O path.

When the diagonal of $\mbc A$ is strictly negative, any vertex-disjoint collection of cycles and I/O paths can be completed into a vertex-disjoint graph cover, by including self-loops at the remaining uncovered nodes. This yields a simple characterization of when a minor contains mixed signs.   

The following result can be found in existing literature. \Cref{cor:symbolic-minor-neg_diag}\ref{cor:symbolic-minor-neg_diag-principal} can be found in \cite{BassettMaybeeQuirk1968} and \cite[Theorem~3.2.1]{BrualdiShader1995}; \ref{cor:symbolic-minor-neg_diag-nonprincipal-noIOpath} in \cite[Theorem~3.2.5(ii)]{BrualdiShader1995}, and \ref{cor:symbolic-minor-neg_diag-nonprincipal} in \cite[Corollary~3.2.3]{BrualdiShader1995}.

\begin{cor}
\label{cor:symbolic-minor-neg_diag}
    Let $\mbc A$ be an $m \times m$ signed symbolic matrix with negative diagonal and J-graph $G$. 
    \begin{enumerate}[label={(\alph*)}]
        \item\label{cor:symbolic-minor-neg_diag-principal} 
            For any $\gamma \subseteq [m]$, the polynomial $\mbc A[\gamma]$ has mixed signs if and only if $G[\gamma]$ contains a PFBL.
        \item\label{cor:symbolic-minor-neg_diag-nonprincipal-noIOpath} 
            Let $i$ be the input, and $j \neq i$ be the output. The non-principal first minor $\mbc A[\hat i, \hat j]$ is identically zero if and only if $G$ has no I/O path.
        \item\label{cor:symbolic-minor-neg_diag-nonprincipal} 
            Let $i$ be the input and $j \neq i$ be the output. Then $\mbc A[\hat i, \hat j]$ has mixed signs if and only if $G$ contains either an IFFL from $i$ to $j$ or a PFBL that is vertex-disjoint from an I/O path.
    \end{enumerate}
\end{cor}
\begin{proof}
\begin{enumerate}[label={(\alph*)}, leftmargin=0pt, itemindent=21pt]
\item 
    Because the diagonal of $\mbc A$ is negative, $T_{\mrm{id}}$ is a term in $\mbc A[\gamma]$, corresponding to the disjoint cycle cover consisting of negative self-loops on every node. The sign of this term is $(-1)^{|\gamma|}$, so there are mixed signs in $\mbc A[\gamma]$ if and only if there exists $\sigma \in S_{\gamma}$ for which $\sgn(T_\sigma)= (-1)^{|\gamma|+1}$. By \Cref{lem:nonzeroterms-sign}\ref{lem:nonzeroterms-sign-a}, this is possible if and only if $G[\gamma]$ has a PFBL $C$, since $C$ with self-loops on the remaining nodes is a disjoint cycle cover with one PFBL. 

\item 
    In a J-graph with self-loops on every node, any I/O path $P$ immediately gives rise to a disjoint I/O cover $E_{P}$, consisting of $P$ and self-loops on nodes not traversed by $P$. Then \ref{cor:symbolic-minor-neg_diag-nonprincipal-noIOpath} follows from \Cref{lem:nonzeroterms-subgraphs}\ref{lem:nonzeroterms-subgraphs-c}. 

\item 
    For an I/O path $P$, let $E_{P}$ be the disjoint I/O cover made of $P$ and self-loops, and let $T_{P}$ be the corresponding non-zero term in $(-1)^{m+i+j}\mbc A[\hat i, \hat j]$. Because $E_{P}$ has no positive cycles, we have $\sgn(T_{P}) = -\sgn(P)$. 
    Suppose $G$ has an IFFL made of I/O paths $P$ and $Q$; then $\sgn(T_{Q}) = -\sgn(T_{P})$.  
    If instead of an IFFL, $G$ has a PFBL $C$ that is vertex-disjoint from $P$, then the disjoint I/O cover consisting of $P$, $C$, and self-loops on the remaining nodes form a disjoint I/O cover. The sign of the corresponding term $T_{P,C}$ in $(-1)^{m+i+j}\mbc A[\hat i, \hat j]$ is $\mrm{sgn}(T_{P,C}) =  +\sgn(P)$. Thus, $T_P$ and $T_{P,C}$ are two terms with opposite signs in $\mbc A[\hat i, \hat j]$. 

    Conversely, suppose the expansion of $(-1)^{m+i+j}\mbc A[\hat i, \hat j]$ has mixed signs, say $\sgn(T_\sigma) = - \sgn(T_\eta)$. Let $E_\sigma$ be the corresponding disjoint I/O cover, where $E_\sigma$ contains the I/O path $P$ and exactly $p$ PFBLs. Similarly, let $E_\eta$ be the I/O cover with the I/O path $Q$ and $q$ PFBLs. 
    If $\sgn(P) \neq \sgn(Q)$, then $P$ and $Q$ form an IFFL. If $\sgn(P) = \sgn(Q)$, by \Cref{lem:nonzeroterms-sign}\ref{lem:nonzeroterms-sign-b}, $p-q$ is odd, so there must be a PFBL in either $E_\sigma$ or $E_\eta$, disjoint from either $P$ or $Q$ respectively. 
    \qedhere 
\end{enumerate}
\end{proof}

\begin{ex}
\label{ex:cycles-permutation-cover--continuation}
    We saw in \Cref{ex:cycles-permutation-cover} that $\det(\mbc A)$ has six terms, only one of which is positive. This is consistent with the six disjoint cycle covers in \Cref{fig:ex:cycles-permutation-cover-principal-covers}, five of which contain exactly one PFBL. 
    
    Similarly, the non-principal first minor $\mbc A[\hat 1, \hat 4]$ has three positive terms, consistent with the fact that all I/O paths are positive (no IFFL for the J-graph), and there is no PFBL vertex-disjoint from an I/O path (\Cref{fig:ex:cycles-permutation-cover-nonprincipal-covers}). 
\end{ex}

\section{Main result}
\label{sec:main-results}

In this section, we characterize when a response curve $x^*_n(u)$ is non-monotonic. The local sensitivity $d_u x^*_n(u)$ can be expressed in terms of a non-principal first minor of a matrix, whose J-graph is essentially that of the system. This allows us to relate the sensitivity to disjoint I/O covers of the J-graph.

We begin by introducing classifications of response curves based on their shapes. Then, using Cramer's rule, we derive an algebraic expression for sensitivity, before applying \Cref{cor:symbolic-minor-neg_diag} to obtain graph motifs associated to each type of responses. Finally, assuming a stable steady state branch, we identify one additional motif associated to non-monotonicity.

\begin{defn}
\label{def:response-types}
    Consider an ODE system $\dot{\xx} = \vv f(\xx; u)$ with control parameter $u$ and output variable $x_n$. Let $\xx^*$ be a differentiable branch of steady states defined on the interval $\uInt_0$. We say the branch $\xx^*$ exhibits 
    \begin{enumerate}[label={(\alph*)}]
    \item 
        \df{perfect adaptation} if  $x_n^*$ is constant on $\uInt_0$;  
    \item 
        \df{quasi-adaptation} if there exists $u_0 \in \uInt_0$ such that $d_u{x_n^*(u_0)} = 0$; 
    \item\label{def:biphasic} 
        \df{biphasic response} if there exist $u_1$, $u_2 \in \uInt_0$ such that $[d_u{x_n^*(u_1)}][d_u{x_n^*(u_2)}] < 0$; 
    \item 
        \df{stable biphasic response} if it is biphasic and $\xx^*(u)$ is locally attracting for every $u \in \uInt_0$. 
    \end{enumerate}
    We say the system exhibits perfect adaptation (quasi-adaptation, etc.) if there exists a steady state branch with the corresponding property. 
\end{defn} 

A biphasic response, also known as \emph{hormesis}, or \emph{bell-shaped}, \emph{U-shaped} or \emph{J-shaped} response~\cite{Calabrese2013,LeverLimKrugerNguyenEtAl2016}, protects the system from overstimulation or overactivation. Although a biphasic response technically involves exactly two phases, \Cref{def:response-types}\ref{def:biphasic} defines a non-monotonic (or multiphasic) response, with no implication on the number of extrema. We elect to use the more standard terminology from pharmacology.

Perfect adaptation, or \emph{zero-sensitivity}, ensures that the system's output is resilient to perturbations. Its relation to the underlying network structure has been investigated in, e.g., \cite{AraujoLiotta2018, OkadaMochizuki2016, GuptaKhammash2022, Khammash2021, Feliu2019, Sontag2014}.  
Quasi-adaptation, also called \emph{infinitesimal homeostasis}~\cite{GolubitskyStewart2017, WangHuangAntoneliGolubitsky2021,ReedBestGolubitskyStewartNijhout2017}, is related to the notion of \emph{homeostasis}, an approximately constant response over an interval. 

Stable biphasic response clearly implies biphasic response, which, as we argue in \Cref{rmk:biphasicImpliesQA} below, implies quasi-adaptation under assumption~\eqref{eq:assumption-A0}. However, as the following examples illustrate, these are distinct phenomena (\Cref{fig:previous-examples}).

\begin{figure}[h]
\centering 
\begin{tabularx}{\linewidth}{CCC}
    \begin{subfigure}{0.9\linewidth}
        \includegraphics{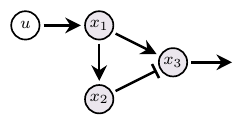}
        
        \includegraphics[height=1.5in]{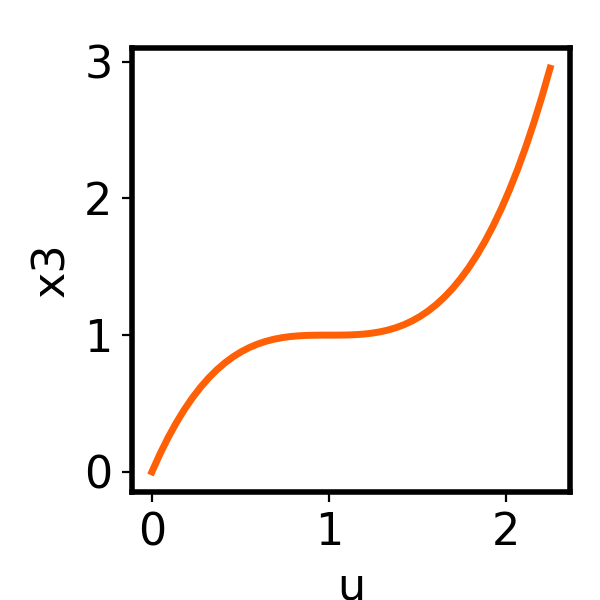} 
        \caption{}
        \label{fig:quasiadaptive-Jgraph}
        \label{fig:quasiadaptive-ss_curves}
    \end{subfigure}
    & 
    \begin{subfigure}{0.9\linewidth}
        \includegraphics{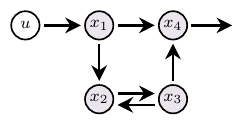}

        \includegraphics[height=1.5in]{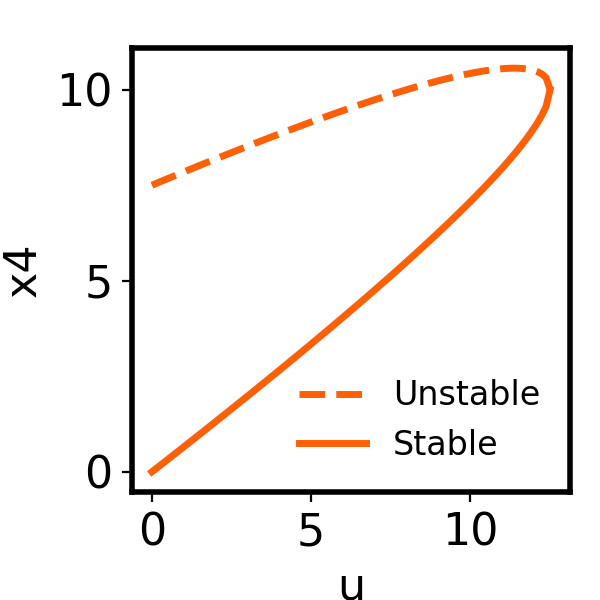} 
        \caption{}
        \label{fig:biphasic_unstable-Jgraph}
        \label{fig:biphasic_unstable-ss_curves}
    \end{subfigure}
    & 
    \begin{subfigure}{0.9\linewidth}
        \includegraphics{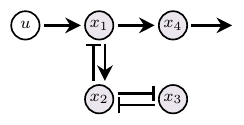}

        \includegraphics[height=1.5in]{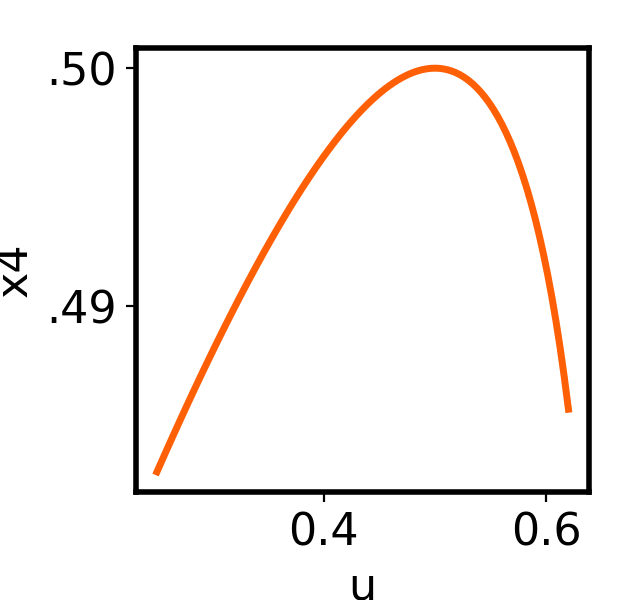} 
        \caption{}
        \label{fig:biorxiv2-Jgraph}
        \label{fig:biorxiv2-ss_curve}
    \end{subfigure}
\end{tabularx}
\caption{({\sc Top}) J-graphs of systems with self-loops omitted for simplicity and ({\sc Bottom}) their corresponding response curves. (a) Quasi-adaptation (\Cref{ex:quasiadaptive}), (b) biphasic response (\Cref{ex:biphasic_unstable}), and  (c) stable biphasic response (\Cref{ex:biorxiv2}). 
}
\label{fig:previous-examples}
\end{figure}

\begin{ex}
\label{ex:quasiadaptive}
    The following system, with $x_3$ as output, is quasi-adaptive, but not biphasic (\Cref{fig:quasiadaptive-ss_curves}):
    \eq{
        \dot x_1 &= u - x_1  \\ 
        \dot x_2 &= x_1^2 - x_2 \\ 
        \dot x_3 &= x_1^3 + 3x_1 - 3 x_2  -x_3.
    }
    At steady state, 
        $x_1^* = u$, 
        $x_2^* = u^2$, and 
        $x_3^* =  (u-1)^3 + 1$. 
    So $x_3^*(u)$ has a vanishing derivative at $u=1$. 
\end{ex}

\begin{ex}
\label{ex:biphasic_unstable}
    This example is biphasic, thus quasi-adaptive, but not stable biphasic (\Cref{fig:biphasic_unstable-ss_curves}). 
    Consider the system on $\rrpp^4$ with $u \in (0,12.5)$ and output $x_4$: 
    \begin{align*}
        \dot x_1 &= u - 2 x_1      \\
        \dot x_2 &=   x_1 -2 x_2^2 + x_3   + {x_3^2}/{4}   \\
        \dot x_3 &=   x_2^2  - {7x_3}/{4}     \\
        \dot x_4 &=  x_1    + {3x_3}/{4}   - x_4. 
    \end{align*} 
    This system has two branches of steady states: a stable monotonic branch and an unstable non-monotonic branch. A full analysis can be found in \cite[Example~III.9, Appendix~C]{YuSontag-biphasic1}. 
\end{ex}

\begin{ex}
\label{ex:biorxiv2}
    This example is stable biphasic (\Cref{fig:biorxiv2-ss_curve}). 
    Consider the system on $\rrpp^4$ with $u \in (0.4,0.6)$ and output $x_4$:
    \begin{align*}
        \dot x_1 &= h(x_2) - 2x_1 + u \\ 
        \dot x_2 &=  h(x_3) + x_1 - x_2  \\ 
        \dot x_3 &= f(x_2) - x_3\\ 
        \dot x_4 &= x_1 - x_4 , 
\intertext{where $h(s) = \frac{1}{2} e^{2(1-\sigma(s))}$, $f(s) = e^{1-\sigma(s)}$, and} 
    \sigma(s) &= \left\{\begin{array}{cl}%
    {0.8}&{0\leq s<0.8} \\ %
    {s  }&{0.8\leq s\leq 1.2} \\%
    {1.2}&{1.2\leq s}.
    \end{array}\right.  \nonumber 
    \end{align*}
In the region where $\sigma$ is smooth, the response curve is not monotonic. A full analysis can be found in \cite{Sontag2020_biphasic} and \cite[Example~III.8, Appendix~B]{YuSontag-biphasic1}. 
\end{ex}

\subsection{A formula for local sensitivity}
\label{sec:main-algebra}

Cramer's rule grants us access to the sensitivity. To streamline the proof, we augment the Jacobian matrix. Then the sign of $d_u{x_n^*(u)}$ is determined by a non-principal first minor of the augmented matrix. 

\begin{defn}
\label{def:augmented-Jacobian} 
    Consider an ODE system $\dot{\xx} = \vv f(\xx; u)$ with $C^1$ right-hand side, control parameter $u$, and a steady state branch $\xx^*$ on the interval $\uInt_0$. Suppose along $\xx^*$, the Jacobian matrix $\mm J = D_{\xx}\vv f$ is non-singular and its entries are signed. 
    The \df{augmented Jacobian matrix} is the $(n+1)\times(n+1)$ matrix 
    \eqn{\label{eq:Cramer-Jacobian-A_mtx} 
        \mm A = \begin{pNiceArray}{c|ccc}
            -1 &  & \vv 0^\top & \\
            \hline 
            \Block{3-1}<>{\deriv{\vv f}{u}} & \Block{3-3}<>{\mm J} & & \\
             & && \\
             & \hphantom{w} &&\hphantom{w}
        \end{pNiceArray}.
    } 
\end{defn}

Note that each entry of $\mm A$ is signed along $\xx^*$, and $\mm A$ is singular precisely when $\mm J$ is singular. The J-graph of $\mm A$ coincides with that of the system, except for the self-loop on $u$. 
The loop $u\negedge u$ and the entry $[\mm A]_{11} = -1$ can be thought of as a technicality, so we can apply \Cref{cor:symbolic-minor-neg_diag}. An alternative perspective is to recast the system into one whose Jacobian matrix is precisely $\mm A$. For example, one such augmented system is $\dot{x}_0 = \tilde{u} - u_0$, and $\dot{\xx} = \vv f(\xx; x_0)$. 

Regardless of interpretation, the entry $[\mm A]_{11}$ and the self-loop at $u$ are inconsequential for the purpose of calculating sensitivity, since we only need a non-principal minor of $\mm A$.

\begin{lem}
\label{lem:Cramer}
    Consider the ODE system in \Cref{def:augmented-Jacobian}, with a steady state branch $\xx^*$, Jacobian matrix $\mm J$, and augmented Jacobian matrix $\mm A$. For any $1 \leq j \leq n$, the local sensitivity is given by 
    \eqn{\label{eq:Cramer} 
        d_u{x_j^*(u)} 
        = \frac{(-1)^{j} \mm A^*[\hat 1, \widehat{j+1}]}{\det (\mm J^*)}, 
    }
    where asterisks denote evaluation along $(\xx^*(u); u)$. 
\end{lem}
\begin{proof}
    Implicitly differentiate $\vv 0 = \vv f(\xx^*(u); u)$ to obtain the matrix equation $\vv 0 = \mm J^* \,( \deriv{\xx^*}{u}) + \vv b$ where $\vv b = \deriv{\vv f}{u}^*$. Let $\mm J^*(j \leftarrow \vv b)$ be the matrix obtained by replacing the $j$th column of $\mm J^*$ by $\deriv{\vv f}{u}^*$. By Cramer's rule, 
    \eq{ 
        d_ux_j^*(u) &= \frac{-\det(\mm J^*(j \leftarrow \vv b))}{\det(\mm J^*)},   
    } 
    which is equivalent to \eqref{eq:Cramer} after swapping columns of $\mm J^*(j \leftarrow \vv b)$. 
\end{proof}

\begin{rmk}
\label{rmk:biphasicImpliesQA} 
    Our earlier claim that biphasic response implies quasi-adaptation is a consequence of \Cref{lem:Cramer}. Under assumption~\eqref{eq:assumption-A0}, the Jacobian matrix is never singular along $\xx^*(u)$, so $\det(\mm J^*)$ is constant in sign. 
    Thus, for $x^*_j(u)$ to be non-monotonic, it must have a critical point $u_0$, which must be a root of $\mm A^*[\hat 1, \widehat{n+1}]$ as a function of $u$. 
\end{rmk}

\begin{rmk}
\label{rmk:old_vs_new_perspective}
    When the control parameter only directly influences one state variable, say $u$ only appear in $\dot{x}_1$ as in \cite{WangHuangAntoneliGolubitsky2021, YuSontag-biphasic1, AraujoLiotta2018},  \Cref{eq:Cramer} simplifies to 
    \eqn{\label{eq:Cramer-single-input} 
        d_u{x^*_j}(u) =  \frac{(-1)^j \deriv{f_1}{u}}{\det (\mm J^*)} \mm J^*[\hat 1, \hat j].
    }
    As explained in \cite{AraujoLiotta2018, YuSontag-biphasic1}, there appears to be two cases: when $j = 1$, the sensitivity is dictated by a principal minor of $\mm J$, whereas for $j\neq 1$, it is a non-principal minor. In this paper, the sensitivity is always given by a non-principal first minor of the augmented matrix $\mm A$. This streamlines the proofs. We may assume, without loss of generality now, that the output variable is $x_n$. 
\end{rmk}

\begin{rmk}
\label{rmk:crit_pt-mixed_signs} 
    Although \Cref{eq:Cramer} requires the steady state branch $\xx^*$, we work with a signed symbolic matrix $\mbc A$ consistent with the sign pattern of $\mm A$. A necessary condition for $\mm A^*[\hat 1, \widehat{n+1}] = 0$ to have a solution is for $\mbc A[\hat 1, \widehat{n+1}]$ to have a factor that can vanish. 

    These factors can be further categorized. The first type requires one or more $a_{ij}$ to vanish (as a function of $u$ after the appropriate substitution). For example, if $a_{22}$ divides $\mbc A[\hat 1, \widehat{n+1}]$, then a root exists when $a_{22}$ vanishes; alternatively, a root of $a_{22}a_{33}+a_{23}a_{32}$ can occur if say $a_{22}$ and $a_{23}$ simultaneously vanish. The second type arises from a factor with mixed signs. For example, if $a_{22}a_{33}-a_{23}a_{32}$ is a factor, then a root exists when $a_{22}a_{33} = a_{23}a_{32}$. Given the assumption of $a_{ij} \geq 0$, biphasic response necessarily requires a factor with mixed signs.   
\end{rmk}

\begin{rmk}
\label{rmk:trimgraph}
    We note that the entire matrix $\mm A$ may not be needed in \eqref{eq:Cramer}, since there may be cancellations between the numerator and denominator in \eqref{eq:Cramer}. In \Cref{sec:app-trimgraph}, we outline which entries of $\mm A$ may be removed when considering whether $x^*_j$ is monotonic. 
\end{rmk}

\subsection{Graph-theoretic condition for sign-sensitivity}
\label{sec:main-graph}

Next, we relate sensitivity to motifs in the J-graph. The first step is algebraic in nature; \Cref{prop:main-alg} is based on the symbolic minor $\mbc A[\hat 1, \widehat{n+1}]$. 
The second step involves transient non-monotonicity (\Cref{prop:main-stable_biphasic}).   
Combining the algebraic and dynamical aspects, \Cref{cor:main} states that either an IFFL or a combination of a PFBL and an NFBL is necessary for stable biphasic response.

\begin{prop}
\label{prop:main-alg}
    Consider the system \eqref{eq:syst_main} under assumption~\eqref{eq:assumption-A0} with output variable $x_n$, and J-graph $G$. 
    \begin{enumerate}[label={(\alph*)}]
        \item\label{prop:main-alg-perfect}
            If $G$ has no I/O path, the system exhibits perfect adaptation along any steady state branch. 
            
        \item\label{prop:main-alg-QA} 
            Suppose $G$ has an I/O path, and $\deriv{f_j}{i}$ is non-vanishing along the steady state branch $\xx^*$  for all $i \neutedge j \in G$, where $i \in [n]\cup\{u\}$ and $j \in [n]$.  
            If the system exhibits quasi-adaptation, then either $G$ has an IFFL from $u$ to $x_n$, or it has a PFBL that is vertex-disjoint from an I/O path.
            
        \item\label{prop:main-alg-biphasic} 
             If the system exhibits biphasic response, then either $G$ has an IFFL from $u$ to $x_n$, or it has a PFBL that is vertex-disjoint from an I/O path.
    \end{enumerate}
\end{prop}
\begin{proof} \begin{enumerate}[label={(\alph*)}]
\item
    Let $\mbc A$ be the signed symbolic matrix consistent with the augmented Jacobian matrix $\mm A$ from \eqref{eq:Cramer-Jacobian-A_mtx}. The J-graph of \eqref{eq:syst_main} is also that of $\mbc A$, except for the missing self-loop on $u$. 
    By \Cref{cor:symbolic-minor-neg_diag}\ref{cor:symbolic-minor-neg_diag-nonprincipal-noIOpath}, the absence of any I/O path implies $\mbc A[\hat 1, \widehat{n+1}] \equiv 0$, so $\mm A[\hat 1, \widehat{n+1}]$ is identically zero. The result follows from \Cref{lem:Cramer}. 

\item 
    Suppose $G$ has an I/O path, so as a polynomial, $\mbc A[\hat 1, \widehat{n+1}]$ is not identically zero. Because each $\deriv{f_j}{i}$ is non-vanishing, that the relevant minor of the augmented Jacobian matrix along $\xx^*$ can vanish implies $\mbc A[\hat 1, \widehat{n+1}]$ has mixed signs; see \Cref{rmk:crit_pt-mixed_signs}. Thus, \ref{prop:main-alg-QA} follows from \Cref{cor:symbolic-minor-neg_diag}\ref{cor:symbolic-minor-neg_diag-nonprincipal}. 

\item 
    Since biphasic response requires the relevant minor of the augmented Jacobian matrix to take strictly positive and strictly negative values somewhere, this means $\mbc A[\hat 1, \widehat{n+1}]$ must have mixed signs. The claim follows from \Cref{cor:symbolic-minor-neg_diag}\ref{cor:symbolic-minor-neg_diag-nonprincipal}.  \qedhere 
\end{enumerate} 
\end{proof}

\begin{rmk}
\label{rmk:Cramer-at-ss}
    The assumptions for \Cref{prop:main-alg} can be relaxed slightly. Instead of asking $\mm A(\xx;u)$ to be signed on all of $\X \times \uInt_0$, we only need it to be signed along $\xx^*$, since \Cref{lem:Cramer} is a local result. 
\end{rmk}

We give three examples below to demonstrate that the sufficient (respectively necessary) condition in \Cref{prop:main-alg} is not necessary (respectively sufficient). This is mainly due to entries in $\mm J$, and thus in $\mm A$, not being algebraically independent.

\begin{ex}
\label{ex:PA-suff_notNecessary}
    Consider the reaction network
    $\cf{0} \FRxn{\,\,\kk_1\,\,}  \cf{X}_1 \FRxn{\,\, u\,\,} \cf{X}_2 \FRxn{\,\,\kk_2\,\,} \cf{0}$,  
    with rate constants $\kk_1$, $\kk_2$, $u > 0$. 
    Under mass-action kinetics, 
    its dynamics is governed by $\dot{x}_1 = \kk_1 - u x_1$, $\dot{x}_2 = ux_1 - \kk_2 x_2$. 
    The augmented Jacobian matrix is 
    \eq{
        \mm A = \begin{pNiceArray}{r|rc} 
            -1   & 0  & 0  \\
            \hline 
            -x_1 & -u & 0 \\ 
            x_1  & u  & -\kk_2 
        \end{pNiceArray}, 
    }
    so the J-graph is a 3-node IFFL ($u \negedge x_1 \posedge x_2$, $u \posedge x_2$), along with self-loops on the nodes for $x_1$, $x_2$. 
    
    By \Cref{lem:Cramer}, the sensitivity is determined by the non-principal minor $\mm A[\hat 1, \hat 3]$:  
    \eq{
        d_u x_2^* \propto \det 
         \begin{pNiceArray}{r|r}
            -x_1^* &  -u  \\ 
            x_1^* &  u 
        \end{pNiceArray} 
        \equiv 0.
    }
    In other words, $x_2^*$ exhibits perfect adaptation, despite the presence of I/O paths. A direct calculation confirms this:  $x_1^* = \kk_1/u$ and $x_2^* = \kk_1/\kk_2$. Moreover, the steady state is asymptotically stable since $\det(\mm J) > 0$ and $\Tr(\mm J) < 0$. In this case, perfect adaptation is a consequence of algebraic dependence among the entries of $\mm A$. 
\end{ex}

\begin{ex}
\label{ex:IFFL-not_suff}   
    The following system on $\rrpp^3$ with $u > 0$ has a monotonic response despite an IFFL architecture ($u \posedge x_1 \negedge x_2 \posedge x_3$, $x_1 \posedge x_3$): 
    \eq{
        \dot{x}_1 &= u - x_1 \\ 
        \dot{x}_2 &= 1 - x_2 - x_1x_2 \\ 
        \dot{x}_3 &= 1 + x_1 x_2 - x_3 . 
    }
    At steady state, $x_1^* = u$, $x_2^* = (1+u)^{-1}$, and $x_3^* = 1 + u/(1+u)$, which is strictly increasing.  
\end{ex}

\begin{figure}[h!t]
\centering 
\begin{subfigure}[b]{0.35\textwidth}
    \centering\figdraftboxwidth
    \raisebox{0.8cm}{\includegraphics{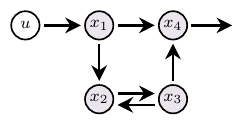}}
    \caption{}
    \label{fig:ex-param_dep_biphasic-J-graph}
\end{subfigure}
\begin{subfigure}[b]{0.35\textwidth}
    \centering\figdraftboxwidth
    \includegraphics[height=1.5in]{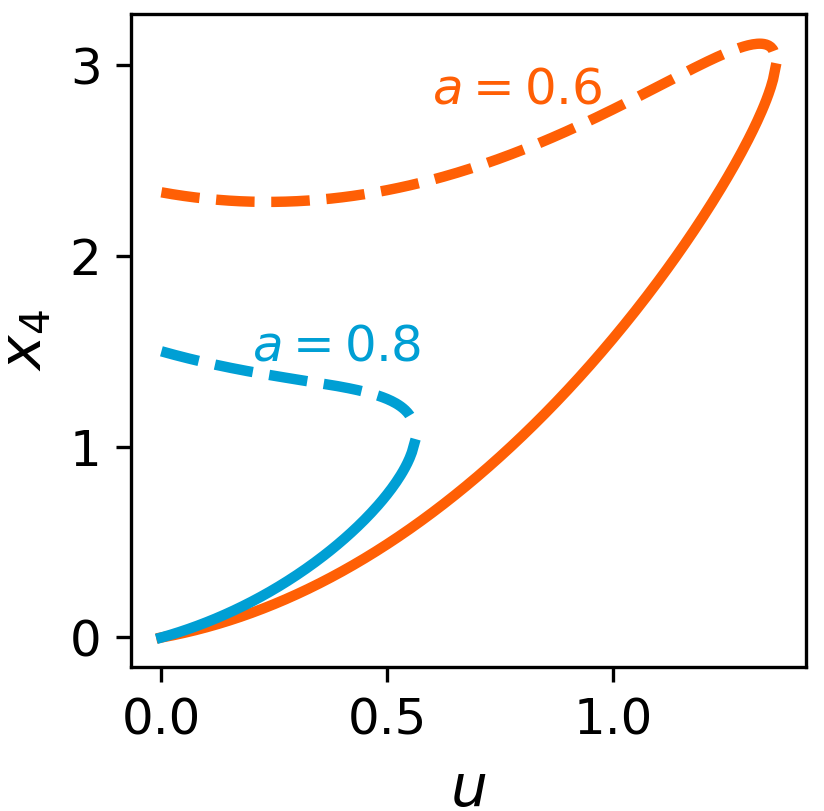} 
    \caption{}
    \label{fig:ex-param_dep_biphasic-response_curve}
\end{subfigure}
\caption{
    (a) The J-graph from \Cref{ex:PFBL-not_suff} with self-loops omitted for simplicity. (b) The positive branches of $x_4^*(u)$ for $a = 0.6$, $0.8$. The stable branch is always monotonic, while the unstable branch is monotonic when $a = 0.8$ and exhibits biphasic response when $a=0.6$. 
}
\label{fig:ex-param_dep_biphasic}
\end{figure}

\begin{ex}
\label{ex:PFBL-not_suff}
    This system, though having a PFBL disjoint from an I/O path (\Cref{fig:ex-param_dep_biphasic-J-graph}), can exhibit biphasic response, depending on the choice of parameter. 
    Let $0 < a < 2$ be a constant to be specified later, and $u \in (0,u_c)$, where $u_c = (a-2)^2/4a^2$. Consider the following system on $\rrpp^4$:
    \eq{
        \dot x_1 &= u - x_1 \\ 
        \dot x_2 &= x_1 - 2 x_2^2  + x_3 + x_3^2 \\ 
        \dot x_3 &= a x_2^2 -x_3  \\ 
        \dot x_4 &= x_1^2 + x_3 -x_4  .
    }
    It has four branches of steady states, two of which are positive
    , but only one of these is asymptotically stable. 
    The unstable positive branch exhibits biphasic response only when $a$ is sufficiently small ($0 < a <  \sqrt{3} - 1 \approx 0.73$). For example, when $a = 0.6$, the unstable response curve $x_4^*$ has a local minimum and a local maximum (\Cref{fig:ex-param_dep_biphasic-response_curve}). However, when $a = 0.8$, both branches are monotonic. A full analysis of this system can be found in \Cref{app:ex-PFBL-not_suff}. 
\end{ex}

Experimentally, only \emph{stable} biphasic response can be observed. For this, a PFBL alone is not enough; stable biphasic response requires an NFBL in addition to a PFBL. 
We apply \cite[Theorem~1]{AscensaoDattaHanciogluSontagEtAl2016}, which relies on monotone systems theory to identify graph-theoretic conditions associated to monotonic transient dynamics.

The graph used in \cite[Theorem~1]{AscensaoDattaHanciogluSontagEtAl2016} does \emph{not} have any self-loops. We define the \df{loop-free J-graph} $G^\circ$ to be the subgraph of $G$ after removing all self-loops. 

\begin{thm}[{\cite[Theorem~1]{AscensaoDattaHanciogluSontagEtAl2016}}]
\label{thm:montone-thm1}
    Consider a system $\dot{\vv x} = \vv F(\vv x(t); u(t))$ with $C^1$ right-hand side, defined on an open set $\X \subseteq \rr^n$, and $u$ is a monotonic positive function with at most a discrete set of discontinuities. Suppose each $\deriv{F_j}{i}$, $\deriv{F_j}{u}$ is either strictly non-positive, strictly non-negative, or identically zero on $\X \times \ran(u)$. Let $G^\circ$ be the loop-free J-graph. Assume every walk in $G^\circ$ from $u$ to $x_n$ has the same sign $\epsilon \in \{-1,1\}$. 
    If the system is initially at steady state, the output $x_n$ will be monotonic in $t$. In particular, if $\epsilon = 1$, monotonically increasing (respectively decreasing) $u$ will trigger monotonic increase (respectively decrease) in $x_n(t)$; if $\epsilon = -1$, it will trigger monotonic decrease (respectively increase) in $x_n(t)$. 
\end{thm}

In short, \Cref{thm:montone-thm1} says that a non-monotonic transient response requires two walks of opposite signs. A walk in $G^\circ$ can be translated as FBLs and FFLs in $G$. 

\begin{lem}
\label{lem:walks-equiv_proposition}
    Let $G^\circ$ be a simple signed directed graph with input $i$ and output $j \neq i$. There are two oppositely signed walks in $G^\circ$ from $i$ to $j$ if and only if $G^\circ$ has either an IFFL, or an NFBL that is downstream of $i$ and upstream of $j$.
\end{lem}
\begin{proof}
    Clearly, an IFFL contains two paths of opposite signs. 
    Suppose $C$ is an NFBL downstream of $i$ and upstream of $j$. Then there exist paths from $i$ to a node in $C$, along $C$, and from $C$ to $j$. Comparing this, with the walk obtained by walking around $C$ one more time, we see that there exist two walks of opposite signs. 
    
    Conversely, any walk from $i$ to $j$ contains an I/O path, possibly with detours made of closed walks. If such a walk traverses an NFBL, we are done. If no walk contains an NFBL, then trimming away the detours, all of which must be PFBLs, produces an I/O path whose sign is the same as that of the walk. Thus, an IFFL exists.
\end{proof}

\begin{prop}
\label{prop:main-stable_biphasic} 
    Consider the system \eqref{eq:syst_main} under assumption~\eqref{eq:assumption-A1} with output variable $x_n$, and J-graph $G$. 
    If the system exhibits stable biphasic response, then $G$ has either an IFFL from $u$ to $x_n$, or an NFBL that is downstream of $u$ and upstream of $x_n$. 
\end{prop}
\begin{proof}
    The J-graph $G$ and the loop-free $G^\circ$ share the same edges between distinct nodes. Thus, $G^\circ$ has an IFFL (or an NFBL downstream of input and upstream of output) if and only if $G$ has one. 

    Suppose $G$ has neither an IFFL nor an NFBL downstream of $u$ and upstream of $x_n$. Then by \Cref{lem:walks-equiv_proposition}, all walks from input to output in $G^\circ$ have the same sign. Without loss of generality, assume all I/O paths are positive; we claim that $x_n^*(u)$ is monotonically non-decreasing. (If all I/O paths are negative, the response curve would be monotonically non-increasing.)
    
    Fix any $u_0 \in \uInt_0$, and let $\delta > 0$ be sufficiently small such that $u' = u_0 + \delta \in \uInt_0$ and $\xx^*(u_0) \in B_{\epsilon(u')}(\xx^*(u'))$; for the existence of $\delta$, see \Cref{sec:A1-implies-convergence}. This ensures that with control parameter $u'$ and initial state $\xx^*(u_0)$, the system converges to $\xx^*(u')$. 
    Define the piecewise function $u(t) = u_0 + \delta$ for $t \geq 0$, and $u(t) = u_0$ for $t < 0$. By \Cref{thm:montone-thm1}~\cite{AscensaoDattaHanciogluSontagEtAl2016}, the solution $\xx(t)$ to the initial value problem $\dot{\xx} = \vv f(\xx; u(t))$, $\xx(0) = \xx^*(u_0)$ is monotonic in time. Therefore, $x_n(t) \geq x_n(0)$ for all $t > 0$, and  $\lim_{t\to\infty} x_n(t) = x_n^*(u_0 + \delta) \geq x^*_n(u_0)$.  Since $\delta$ can be arbitrarily small, taking the limit as $\delta \to 0^+$, we conclude that $d_u x_n^*(u_0) \geq 0$, as desired.  
\end{proof}

Combining \Cref{prop:main-alg}\ref{prop:main-alg-biphasic} and \Cref{prop:main-stable_biphasic} gives us tighter necessary conditions for stable biphasic response.

\begin{thm}
\label{cor:main}
    Consider the system \eqref{eq:syst_main} under assumption~\eqref{eq:assumption-A1} with output variable $x_n$, and J-graph $G$. If the system exhibits stable biphasic response, then $G$ has either (i) an IFFL from $u$ to $x_n$, or (ii) both an NFBL that is downstream of $u$ and upstream of $x_n$ and a PFBL that is vertex-disjoint from an I/O path. 
\end{thm}

\begin{rmk}
\label{rmk:Cramer-at-ss-dynamical} 
    \Cref{prop:main-stable_biphasic} and \Cref{cor:main} requires entries of the augmented Jacobian matrix to be signed on all of $\X \times \uInt_0$, not just along a steady state branch $\xx^*$, which is the case for \Cref{prop:main-alg}; see \Cref{rmk:Cramer-at-ss}. 
\end{rmk}

\begin{figure}[th]
\centering 
\begin{subfigure}[b]{0.35\textwidth}
    \centering\figdraftboxwidth
    \raisebox{0.8cm}{\includegraphics{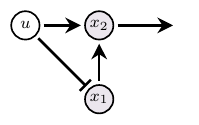}}
    
    \caption{}
    \label{fig:ex-CRN-Jgraph}
\end{subfigure}
\begin{subfigure}[b]{0.35\textwidth}
    \centering\figdraftboxwidth
    \includegraphics[height=1.5in]{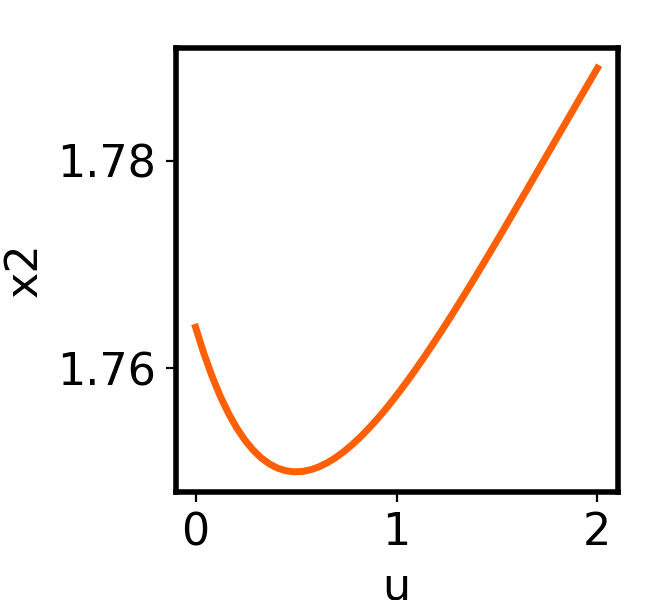} 
    \caption{}
    \label{fig:ex-CRN-response}
\end{subfigure}
\caption{
    (a) The loop-free J-graph from \Cref{ex:CRN-rateconstant}, and   
    (b) the response curve $x_2^*(u)$ when the parameters $\kk_i$ are set to $1$.  
}
\label{fig:ex-CRN}
\end{figure}

\begin{ex}
\label{ex:CRN-rateconstant}
    Consider the reaction network
    \begin{gather*}
        \cf{X}_1 \FRxn{\,\,u\,\,}  \cf{X}_2,  \quad 
        \cf{0} \RRxn{\,\,\kk_1\,\,}{\,\,\kk_2\,\,} \cf{X}_1 , \quad 
        \cf{0} \RRxn{\,\,\kk_3\,\,}{\kk_4} \cf{X}_2 , \quad 
        2\cf{X}_1  \FRxn{\,\,\kk_5\,\,}  \cf{X}_1 + 2\cf{X}_2  . 
    \end{gather*}
    Under mass-action kinetics, the dynamics is governed by
    \eq{
        \dot{x}_1 &= \kk_1 - \kk_2 x_1  - \kk_5 x_1^2 - u x_1\\
        \dot{x}_2 &= \kk_3 - \kk_4 x_2  + 2 \kk_5 x_1^2 + u x_1, 
    }
    where $\kk_i$, $u > 0$ are parameters. The augmented Jacobian matrix is 
    \eq{
        \mm A = \begin{pNiceArray}{r|cc} 
            -1 &  0 & 0  \\
            \hline 
            -x_1 &  - (u +\kk_2 + 2\kk_5 x_1) & 0 \\ 
            x_1 &  u + 4\kk_5 x_1 & -\kk_4 
        \end{pNiceArray}, 
    }
    so its J-graph is a simple IFFL (\Cref{fig:ex-CRN-Jgraph}). Stability can be derived from the fact that $\det(\mm J) > 0$ and $\Tr(\mm J) < 0$.

    Because of the lack of FBLs, by \Cref{cor:main} the IFFL is necessary for biphasic response. 
    Applying \Cref{lem:Cramer}, we find that $d_u x_2^* \propto \mm A[\hat 1, \hat 3] = x_1^* (\kk_2 - 2\kk_5 x_1^* )$. 
    In other words, the system exhibits biphasic response in $x_2$ if and only if $\kk_2 - 2\kk_5 x_1^*(u)$ changes sign as $u$ increases. In \Cref{app:ex-CRN-rateconstant}, we show that biphasic response is possible if and only if $4 \kk_1 \kk_5 > 3 \kk_2^2$; for example, \Cref{fig:ex-CRN-response} shows the response curve when $\kk_r = 1$ for all $r$. 
\end{ex}

\begin{rmk}
\label{rmk:FFL-motif}
    Our definition of feedforward loop permits common internal nodes among the two I/O paths. However, an IFFL necessarily contains within in an IFFL with disjoint interior. See \Cref{sec-app:rmk:FFL-motif} for detail. 
\end{rmk}

\section{Applications}
\label{sec:application}

Signed directed graphs are ubiquitous in biology, particularly in signaling cascades and gene regulation, which make up the first two applications.  We show that in these models, the regulatory networks, as they are typically drawn, are essentially the loop-free J-graphs of the corresponding models. Thus, \Cref{cor:main} provides necessary network structure for stable biphasic response in regulatory network models. 

For our last application (\Cref{sec:application-CRM}), we consider consumer-resource models. These ecological models are typically represented by community networks --- bipartite graphs depicting interactions between the resources and the consumers. From the community network, we derive the J-graph at an arbitrary coexistence equilibrium. In this case, the diagonal has some zero entries. Using \Cref{lem:nonzeroterms-subgraphs}, we give an alternative proof of the Competitive Exclusion Principle for consumer-resource models.

\subsection{Signaling pathway and T-cell activation}
\label{sec:application-Tcell}

Experimentally, the amount of cytokines released by T-cells depends non-monotonically on the concentration of antigen~\cite{LeverLimKrugerNguyenEtAl2016}. A model of T-cell activation was investigated computationally in \cite{LeverLimKrugerNguyenEtAl2016}, and a special case of it was studied analytically in \cite{YuSontag-biphasic1}. 
Here, we analyze the model in general to determine the necessary network architectures for a biphasic response. Moreover, we give a concrete example where the signaling pathway does \emph{not} have an IFFL, yet the steady state response is biphasic.

The model from \cite{LeverLimKrugerNguyenEtAl2016} consists of two parts (\Cref{fig:appl-Tcell-Lever}). The first part, occurring at the cell membrane, involves the binding of an extracellular antigen to a T-cell receptor (TCR). The resulting TCR complex undergoes tyrosine phosphorylation according to the McKeithan kinetic proofreading model~\cite{McKeithan1995}. The second part is an intracellular signaling network $G^\circ$ governed by catalytic regulation (see \Cref{sec:application-Tcell-IC} for detail). The coupling between the two parts takes place when a phosphorylated TCR complex triggers the intracellular pathway.

\begin{figure}[h!t]
    \centering\setlength\tabcolsep{1.5em}
    \begin{tabular}{cc}
    \multirow{2}{*}[1.83cm]{\subcaptionbox{\label{fig:appl-Tcell-Lever}}{
        \includegraphics[]{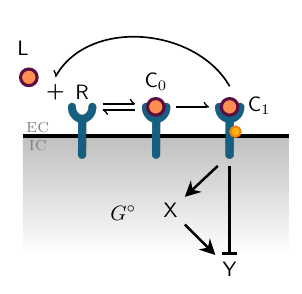}
    }}
    &
    \subcaptionbox{\label{fig:appl-Tcell-KPR}}{
        \includegraphics[]{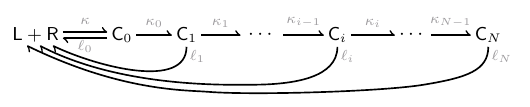}
    }
    \\[1cm]
    &
    \subcaptionbox{\label{fig:appl-Tcell-downstream-rxns}}{
        \includegraphics[]{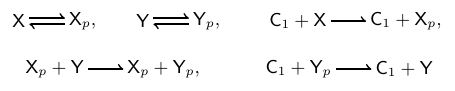}
    }
    \end{tabular}
\caption{
    (a) A model of T-cell activation from \cite{LeverLimKrugerNguyenEtAl2016}, where the kinetic proofreading stage at the cell membrane is coupled to an intracellular signaling network $G^\circ$. 
    (b) The reaction scheme for the kinetic proofreading stage where a T-cell receptor $\cf{R}$ binds to an antigen $\cf{L}$ and undergoes sequential phosphorylation. 
    (c) The reaction scheme for the signaling network $G^\circ$ shown in (a). Here, $\cf{X}_p$ and $\cf{Y}_p$ are active forms of their respective proteins. 
}
\label{fig:appl-Tcell}
\end{figure}

\subsubsection{McKeithan kinetic proofreading model}
\label{sec:application-Tcell-KPR}

At the cell membrane, an extracellular antigen ($\cf{L}$) binds to a TCR ($\cf{R}$) that has $N$ phosphorylation sites. Assuming sequential phosphorylation, the elementary reactions are shown in \Cref{fig:appl-Tcell-KPR}. Under mass-action kinetics, the model has a globally attracting positive equilibrium for any positive parameters~\cite{Sontag2001}. 

In \cite[Appendix~D]{YuSontag-biphasic1}, we derived the equilibrium concentration of each phosphorylated antigen-TCR complex: 
\eqn{\label{eq:Tcell-Ci-ss}
    [\cf{C}_i]^* = A_i \underbrace{\left( L_T + \rho + \alpha - \sqrt{(L_T +  \rho + \alpha)^2 - 4 \rho L_T} \right)}_{U(L_T)} , \qquad 0 \leq i \leq N, 
}
where $A_i, \alpha > 0$ are parameters dependent on rate constants, $\rho > 0$ is the total amount of TCRs, and $L_T$ is the total amount of antigen. 

Since $U(L_T)$, and hence $[\cf{C}_i]^*$, increases monotonically with $L_T$, the biphasic response observed in the cytokine released must stem from the intracellular signaling network $G^\circ$.

\subsubsection{Intracellular signaling network}
\label{sec:application-Tcell-IC}

We now turn our attention to the signaling pathway $G^\circ$. For simplicity, we assume mass-action kinetics; however, the main result of this section applies under more general kinetics (e.g., Michaelis--Menten, Hill-type). These alternative kinetics are covered in \Cref{sec:application-Tcell-appendix}. 

Consider a signaling network $G^\circ$, a simple signed directed graph. Each node represents either a TCR complex $\cf{C}_i$ or a signaling protein in one of two forms. We assume spontaneous conversion between the inactive form ($\cf{X}$) and active form ($\cf{X}_p$). Each edge in $G^\circ$ represents either \emph{catalytic activation} or \emph{catalytic inhibition}. 

The activation regulation $\cf{X} \posedge \cf{Y}$ refers to the elementary reaction $\cf{X}_p + \cf{Y} \FRxn{} \cf{X}_p + \cf{Y}_p$, where $\cf{X}_p$ catalyzes the activation of $\cf{Y}$. Under mass-action kinetics, the reaction rate function of the activation regulation is 
    \eq{
        \lambda = k [\cf{X}_p] [\cf{Y}]  
    }
with rate constant $k > 0$. In contrast, the inhibition regulation $\cf{X} \negedge \cf{Y}$ represents the reaction $\cf{X}_p + \cf{Y}_p \FRxn{} \cf{X}_p + \cf{Y}$, with rate function 
    \eq{
        \mu &= 
                 k [\cf{X}_p] [\cf{Y}_p]  . 
    }
We assume no auto-regulation, that is, any regulation $\cf{X} \neutedge \cf{Y}$ involves two distinct proteins. 
Finally, any TCR complex $\cf{C}_i$ ($i \geq 1$) can act as a regulating protein, provided all complexes regulate any intracellular proteins consistently; for example, if $\cf{C}_i$ activates $\cf{Y}$, then no $\cf{C}_j$ can inhibit $\cf{Y}$.

\begin{ex}
\label{ex:Tcell-3node}
    Consider the three-node IFFL signaling network $G^\circ$ in \Cref{fig:appl-Tcell-Lever}. Its reaction scheme, shown in  \Cref{fig:appl-Tcell-downstream-rxns},  
    gives rise to the following ODE system on $\rrpp^4$ under mass-action kinetics: 
    \eq{
        -\dot{[\cf{X}]} = \dot{[\cf{X}_p]} &= \gamma_{1} [\cf{X}] - \delta_{1} [\cf{X}_p]  
            + \eta_1 [\cf{C}_1] [\cf{X}]  \\ 
        - \dot{[\cf{Y}]} = \dot{[\cf{Y}_p]} &= \gamma_{2} [\cf{Y}] - \delta_{2} [\cf{Y}_p]  
            + k_1  [\cf{X}_p] [\cf{Y}]
            - \eta_2  [\cf{C}_1]  [\cf{Y}_p] , 
    }
    with parameters $\gamma_i, \delta_i, \eta_i, k_1 >0$. 
    Using the conservation laws $[\cf{X}] + [\cf{X}_p] = X_T$, $[\cf{Y}] + [\cf{Y}_p] = Y_T$, and rescaling $x = [\cf{X}_p]/X_T$,  $y = [\cf{Y}_p]/Y_T$, we reduce the system to one on $(0,1)^2 \subseteq \rrpp^2$:  
    \eq{
        \dot{x} &=  \gamma_{1} (1-x) - \delta_{1} x 
            + \eta_1 [\cf{C}_1] (1-x)  \\ 
        \dot{y} &= \gamma_2 (1-y) - \delta_2 y + \hat{k}_1 x(1-y) - \eta_2 [\cf{C}_1] y ,
    }
    where $\hat{k}_1 = k_1X_T > 0$. 
    
    Without loss of generality, we assume $[\cf{C}_1]$ has reached its globally attracting equilibrium, and substitute in $[\cf{C}_1] = A_1 U(L_T)$ from \eqref{eq:Tcell-Ci-ss}. 
    Let $L_T$ be the control parameter. The augmented Jacobian matrix of the reduced system is 
    \eq{
        \mm A = 
        \begin{pNiceArray}{c|cc} 
            -1 &  0 & 0  \\
            \hline 
            \eta_1 A_1 (1-x) U' & -(\gamma_1 + \delta_1 + \eta_1 A_1 U) & 0 \\ 
            - \eta_2 A_1 y U' & \hat{k}_1 (1-y) & -(\gamma_2 + \delta_2 + \hat{k}_1 x + \eta_2  A_1 U) 
        \end{pNiceArray}.
    }
    Observe that the signaling network $G^\circ$ in \Cref{fig:appl-Tcell-Lever} is the loop-free J-graph of $\mm A$. The diagonal of $\mm A$ is negative because of the spontaneous conversions. Since the J-graph does not have any FBLs, \Cref{cor:main} implies that the IFFL is necessary for stable biphasic response. 
\end{ex}

The example above contains the key ideas for the proof of \Cref{thm:T-cell} on necessary condition for stable biphasic response in the T-cell activation model. In the case where multiple TCR complexes regulate the intracellular proteins, there really is only one input node for the control parameter because $[\cf{C}_j]^* = A_j U(L_T)$ for all $j$. Thus, we merge the nodes for $\cf{C}_j$ into a single vertex to obtain the loop-free J-graph $\tilde{G}^\circ$ from the signaling network $G^\circ$. 
Thus, by an \df{IFFL} of the signaling network $G^\circ$, we mean an IFFL of $\tilde{G}^\circ$. Note that the two paths in the IFFL can start at different $\cf{C}_j$'s. Analogous definitions can be made for FBLs that are downstream of input and upstream of output.

We now state and prove the main result of this section. 

\begin{thm}
\label{thm:T-cell}
    Consider a model of T-cell activation, consisting of (i) the kinetic proofreading model for a TCR with $N$ sites; (ii) an intracellular signaling pathway $G^\circ$ with $n$ intracellular signaling proteins, and (iii) regulation of some intracellular proteins by the TCR complexes. Assume regulation by the TCR complexes happen in a consistent manner. 
    Suppose the amount of cytokines released by the T-cell is proportional to $\cf{X}_{np}$, and the Jacobian matrix of the reduced system is Hurwitz at any positive equilibrium. Let $\tilde{G}^\circ$ the graph obtained from $G^\circ$ after merging all TCR complexes. 

    If $[\cf{X}_{np}]$ exhibits stable biphasic response with respect to the total amount of antigen $L_T$, then $\tilde{G}^\circ$ must contain either an IFFL, or a combination of an NFBL and a PFBL, where the PFBL is vertex-disjoint from an I/O path, and the NFBL is downstream of $\cf{C}_j$ for some $j$ and upstream of $\cf{X}_n$. 
\end{thm}
\begin{proof}
    Without loss of generality, we assume that the kinetic proofreading system has reached its globally stable equilibrium~\cite{Sontag2001}, and substitute in $[\cf{C}_j] = A_j U(L_T)$ from \eqref{eq:Tcell-Ci-ss} into the ODE system for the intracellular signaling pathway. 

    Let $\mm \Gamma \in \rr^{n\times N}$ be the adjacency matrix for how $\cf{C}_j$ regulates $\cf{X}_i$, and let $\mm \Xi \in \rr^{n\times n}$ be the matrix encoding regulations among the species $\cf{X}_i$.
    For example, for the network in \Cref{fig:appl-Tcell-Lever},\\[-2em] 
    \eq{
        \mm \Gamma = \begin{blockarray}{cr}
            & {\color{gray}\scriptstyle C_1} \\ 
            \begin{block}{c(r)}
            {\color{gray}\scriptstyle X} & 1 \\ 
            {\color{gray}\scriptstyle Y} & -1 \\
            \end{block} 
        \end{blockarray} \, , 
        \quad \text{and} \quad 
        \mm \Xi = \begin{blockarray}{ccc}
            & {\color{gray}\scriptstyle X} & {\color{gray}\scriptstyle Y} \\ 
            \begin{block}{c(cc)}
            {\color{gray}\scriptstyle X} & 0 & 0 \\ 
            {\color{gray}\scriptstyle Y} & 1 & 0 \\
            \end{block} 
        \end{blockarray} \, . \\[-3.25em]
    }
    By our assumption of consistent regulation by TCR complexes, each row of $\mm \Gamma$ cannot have mixed signs. Because no auto-regulation is allowed, the diagonal of $\mm \Xi$ is zero.

    The full ODE system on $\rrpp^{2n}$ modeling the intracellular species is given by $\dot{[\cf{X}_i]} = - \dot{[\cf{X}_{ip}]}$ for $i=1,\ldots, n$, where 
    \eq{
        \dot{[\cf{X}_{ip}]} &=  
            \gamma_i [\cf{X}_i] - \delta_i [\cf{X}_{ip}] 
            + \underbrace{\sum_{\substack{j=1, \\ [\mm \Gamma]_{i,j} = 1}}^N \eta_{ji}[\cf{C}_j] [\cf{X}_i] 
                - \sum_{\substack{j=1, \\ [\mm \Gamma]_{i,j} = -1}}^N \eta_{ji}[\cf{C}_j] [\cf{X}_{ip}] 
                }_{\text{Regulation of $\cf{X}_i$ by $\cf{C}_j$}}
            \\&\qquad + \underbrace{\sum_{\substack{j=1, \\ [\mm \Xi]_{i,j} = 1}}^n k_{ji}[\cf{X}_{jp}] [\cf{X}_i] 
                - \sum_{\substack{j=1, \\ [\mm \Xi]_{i,j} = -1}}^n k_{ji}[\cf{X}_{jp}] [\cf{X}_{ip}] 
                }_{\text{Regulation of $\cf{X}_i$ by $\cf{X}_j$}},
    }
    where $[\cf{C}_j] = A_jU(L_T)$ and $\gamma_i$, $\delta_i$, $\eta_{ji}$, $k_{ji} > 0$.  
    Using the conservation laws $[\cf{X}_i] + [\cf{X}_{ip}] = T_i$ and rescaling the state variables $x_i = [\cf{X}_{ip}]/T_i$, we reduce the system to one on $(0,1)^n \subseteq \rrpp^n$: 
    \eqn{\label{eq:Tcell-nondim}
    \begin{split}
        \dot{x}_i &= \gamma_i (1-x_i) - \delta_i x_i
            + \sum_{j \colon [\mm \Gamma]_{i,j} = 1} \hat{\eta}_{ji}  (1-x_i) U(L_T)
                - \sum_{j \colon [\mm \Gamma]_{i,j} = -1} \hat{\eta}_{ji} x_i U(L_T) 
           \\&\qquad + \sum_{j \colon [\mm \Xi]_{i,j} = 1} \hat{k}_{ji} x_j (1-x_i)
                - \sum_{j \colon [\mm \Xi]_{i,j} = -1} \hat{k}_{ji} x_j x_i,  
    \end{split}
    }
    where $\hat{\eta}_{ji} = \eta_{ji} A_j$ and $\hat{k}_{ji} = k_{ji}T_j$. 
    
    Next, treating $L_T$ as the control parameter, we compute the augmented Jacobian matrix $\mm A \in \rr^{(n+1)\times(n+1)}$. To keep notation simple, we start indexing at $0$ (for the control $L_T$). The diagonal entries are strictly negative, since $[\mm A]_{0,0} = -1$ and for $i \in [n]$, 
    \eq{
        [\mm A]_{i,i} &= -\gamma_i - \delta_i 
            - \sum_{j \colon [\mm \Gamma]_{i,j} \neq 0} \hat{\eta}_{ji} U(L_T)
            - \sum_{j \colon [\mm \Xi]_{i,j} \neq 0} \hat{k}_{ji} x_j < 0. 
    }
    For the off-diagonal entry, we have $[\mm A]_{0,j} = 0$ for all $j \in [n]$. Let $i \neq j \in [n]$; differentiating the right-hand side of \eqref{eq:Tcell-nondim} with respect to $x_j$, we obtain 
    \eq{
        [\mm A]_{i,j} &= \left\{
            \begin{array}{cl}
               \hat{k}_{ji} (1-x_i)  & \text{ if } [\mm\Xi]_{i,j} = 1, \\
                -\hat{k}_{ji} x_i & \text{ if } [\mm\Xi]_{i,j} = -1, \\
               0  & \text{ if } [\mm\Xi]_{i,j} = 0 .
            \end{array}
            \right. 
    } 
    Finally, the first column of $\mm A$, which captures the direct effect $L_T$ has on $x_i$, is 
    \eq{
        [\mm A]_{i,0} &= \left\{ 
        \begin{array}{cl}
            0 & \text{ if } [\mm \Gamma]_{i,\cdot} = \vv 0, \\ 
             \sum_{j \colon [\mm \Gamma]_{i,j} \neq 0 } \hat{\eta}_{ji} (1-x_i) U'(L_T) & \text{ if } [\mm \Gamma]_{i,\cdot} \neq \vv 0 ,  \, [\mm \Gamma]_{i,\cdot} \geq \vv 0,   \\
             -\sum_{j \colon [\mm \Gamma]_{i,j} \neq 0 } \hat{\eta}_{ji} x_i U'(L_T)  & \text{ if }  [\mm \Gamma]_{i,\cdot} \neq \vv 0 , \,  [\mm \Gamma]_{i,\cdot} \leq \vv 0  ,
        \end{array}
        \right. 
    }
    for any $i \in [n]$. In other words, the loop-free J-graph of $\mm A$ is precisely $\tilde{G}^\circ$.  
    If $x_n$ exhibits biphasic response with respect to $L_T$, then by \Cref{cor:main}, $\tilde{G}^\circ$ must contain either an IFFL, or a combination of a PFBL and an NFBL with the stated properties. 
\end{proof}

\begin{figure}[h!t]
\centering 
\begin{subfigure}[b]{0.45\textwidth}
    \centering\figdraftboxwidth
    \raisebox{0.3cm}{\includegraphics{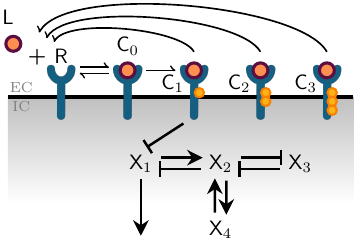}}
    \caption{}
    \label{fig:ex-noIFFL-scheme}
\end{subfigure}
\begin{subfigure}[b]{0.4\textwidth}
    \centering\figdraftboxwidth
    \includegraphics[height=1.75in]{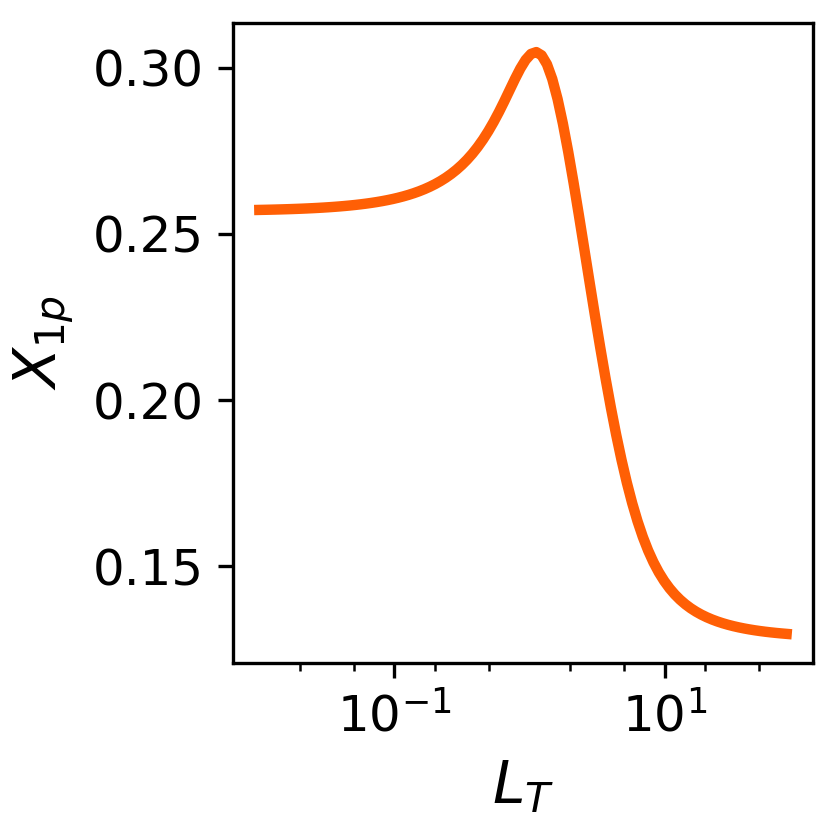} 
    \caption{}
    \label{fig:ex-noIFFL-response}
\end{subfigure}
\caption{(a) A model without an IFFL, (b) that is capable of exhibiting biphasic response.}
\label{fig:appl-Tcell-noIFFL}
\end{figure}

\begin{ex}
\label{ex:appl-Tcell-noIFFL}
    The model depicted in \Cref{fig:ex-noIFFL-scheme} is capable of exhibiting biphasic response without an IFFL (\Cref{fig:ex-noIFFL-response}). 
    See \Cref{app:ex:appl-Tcell-noIFFL} for details about this model. 
\end{ex}

\subsection{Gene regulation}
\label{sec:application-GRN}

Consider a gene regulatory network (GRN) model that tracks only transcript concentrations. 
We assume an arbitrary GRN $G^\circ$ consisting of $n$ genes, where the first gene is subject to external stimulation, and the last gene produces the output. Assume \emph{no} auto-activation. Note that $G^\circ$ is a signed directed graph, possibly with some negative self-loops (if there are auto-inhibition), but no positive self-loops.

For $i\in [n]$, let $x_i$ be the transcript concentration of gene $i$. Each gene has a basal transcription rate $\gamma_i > 0$, and each transcript degrades at rate $\delta_i > 0$. We assume that $\gamma_1 = \gamma_1(u)$ is a monotonically increasing function of the strength of an external signal $u$. 
Regarding the regulatory edges in $G^\circ$, a positive edge $x_i \posedge x_j$ ($j \neq i$) implies we add to the ODE for $x_j$, the Hill activation function 
\eqn{\label{eq:activation_fcn-GRN}
    \lambda(x_i) = \frac{k x_i^{m}}{V + x_i^{m}}, 
}
where $k$, $V$, $m > 0$. Similarly, for a negative edge $x_i \negedge x_j$, we add the repressive term 
\eqn{\label{eq:inhibition_fcn-GRN}
    \mu(x_i) = \frac{k}{V + x_i^{m}}.
}
Note that $\lambda(x_i)$ is a positive, increasing function of $x_i$, whereas $\mu(x_i)$ is a positive, decreasing function.

\begin{figure}[h!t]
\centering 
\begin{subfigure}[b]{0.35\textwidth}
    \centering\figdraftboxwidth
    \raisebox{1.2cm}{\includegraphics{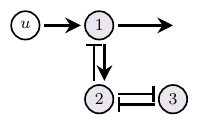}}
    
    \caption{}
    \label{fig:ex-GRN-4node-Jgraph}
\end{subfigure}
\begin{subfigure}[b]{0.35\textwidth}
    \centering\figdraftboxwidth
    \includegraphics[height=1.75in]{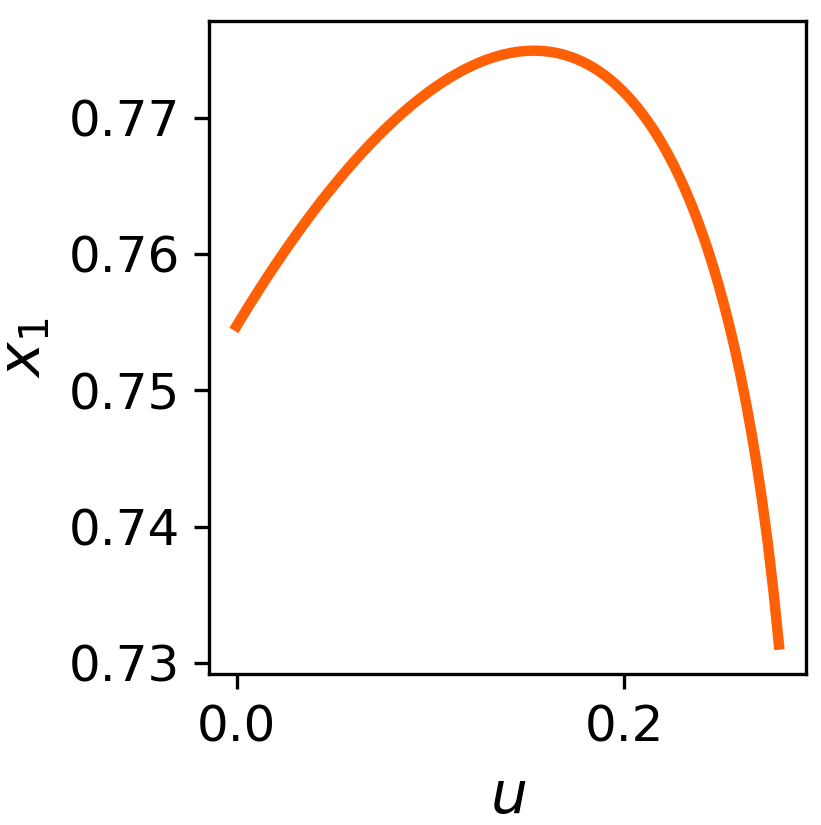} 
    \caption{}
    \label{fig:ex-GRN-4node-ss}
\end{subfigure}
\caption{
    (a) The GRN from \Cref{ex:GRN-4node}.    
    (b) The response curve $x_1^*(u)$. 
}
\label{fig:ex-GRN-4node}
\end{figure}

\begin{ex}
\label{ex:GRN-4node}
    Consider the four-node GRN $G^\circ$ with output $x_1$ in \Cref{fig:ex-GRN-4node-Jgraph}. With $\lambda_r(x_i) = k_r x_i^{m_r}/(V_r + x_i^{m_r})$ and $\mu_r(x_i) = k_r / (V_r + x_i^{m_r})$, the system modeling this network is 
    \eq{
        \dot{x}_1 &= u - \delta_1 x_1 + \mu_1(x_2) \\ 
        \dot{x}_2 &= \gamma_2 - \delta_2 x_2 + \lambda_2(x_1) + \mu_3(x_3)
            \\ 
        \dot{x}_3 &= \gamma_3 - \delta_3 x_3 +  \mu_4(x_2) 
            . 
    }
    Assuming all parameters are positive, and $u$ is the control parameter, the augmented Jacobian matrix is 
    \eq{
        \mm A = \begin{pNiceArray}{r|rrr} 
        -1 & 0 & 0 & 0 \\ 
        \hline 
        1&    -\delta_1 & \mu_1' & 0 \\ 
        0&    \lambda_2' & -\delta_2 & \mu_3' \\
        0&    0 & \mu_4' & -\delta_3 
        \end{pNiceArray},
    }
    where 
    \eqn{\label{eq:derivative-fcn-GRN}
        \lambda_r' = \frac{V_r m_r}{k_r x_i^{m_r+1}} \lambda_r^2
            > 0 
        \quad\text{and} \quad 
        \mu_r' = \frac{-m_rx_i^{m_r-1} }{k_r } \mu_r^2 
            < 0. 
    }
    The GRN $G^\circ$ is also the loop-free J-graph of this system.  
    By \Cref{lem:Cramer}, the sensitivity is determined by the minor 
    \eq{
        d_ux_1^* \propto \mm A[\hat 1, \hat 2] 
            = \delta_2 \delta_3 - |\mu_3'| |\mu_4'|. 
    }
    This indicates that biphasic response cannot be ruled out (\Cref{fig:ex-GRN-4node-ss}). See \Cref{app:ex:GRN-4node} for parameter values and an argument for stability.  
\end{ex}

This example above contains the key ideas for the proof of \Cref{thm:GRN} on necessary condition for stable biphasic response in GRNs. 
Crucial to the proof is that activating involves an increasing term ($\lambda_r' > 0$), and inhibition involves a decreasing term ($\mu_r' < 0$). Note that auto-activation would have introduced mixed signs along the diagonal of the Jacobian matrix, and thus fall outside the scope of our framework. 

We now state and prove the main result of this section.

\begin{thm}
\label{thm:GRN} 
    Consider a GRN $G^\circ$ consisting of $n$ genes with no auto-activation. Suppose the basal transcription rate of $x_1$ is a monotonically increasing function of an external signal $u$, and let $x_n$ be the output. Assume the Jacobian matrix is Hurwitz at any positive equilibrium. 

    If $x_n$ exhibits stable biphasic response with respect to $u$, then the GRN $G^\circ$ must have either an IFFL, or a combination of an NFBL and a PFBL, where the PFBL is vertex-disjoint from an I/O path, and the NFBL is downstream of $x_1$ and upstream of $x_n$. 
\end{thm}
\begin{proof}
    The proof mirrors that of \Cref{thm:T-cell}. 
    Let $\mm \Xi \in \rr^{n\times n}$ be the adjacency matrix of $G^\circ$; it has non-positive diagonal because there is no auto-activation. The model for the GRN is 
    \eq{
        \dot{x}_i & = \gamma_i  - \delta_i x_i + \sum_{j \colon [\mm \Xi]_{i,j} =1 } \lambda_{ji}(x_j) +  \sum_{j \colon [\mm \Xi]_{i,j} = -1 } \mu_{ji}(x_j), \qquad i=1, 2,\ldots, n, 
    }
    where $\gamma_1 = \gamma_1(u)$, each $\lambda_{ji}$ is of the form \eqref{eq:activation_fcn-GRN}, and each $\mu_{ji}$ is of the form \eqref{eq:inhibition_fcn-GRN}. The off-diagonal entries of the Jacobian matrix $\mm J$ are
    \eq{
        [\mm J]_{j,i} = \left\{ 
            \begin{array}{cl}
                \lambda_{ji}' > 0 & \text{ if }  [\mm \Xi]_{j,i} =1, \\
                \mu_{ji}' < 0 & \text{ if }  [\mm \Xi]_{j,i} = -1,  \\
                0 & \text{ otherwise,} 
            \end{array}
            \right. 
    }
    with $\lambda_{ji}'$ and $\mu_{ji}'$ as in \eqref{eq:derivative-fcn-GRN}.   
    Each diagonal entry is either $[\mm J]_{i,i} = -\delta_i$ or $[\mm J]_{i,i} = -\delta_i + \mu'_{ii} < 0$, depending on whether gene $i$ auto-inhibits itself. Thus, after adding negative self-loops to each node of $G^\circ$, a node for $u$, and an edge $u \posedge x_1$, we obtain the J-graph of the system. Note any I/O path must go through $x_1$, so the claim then follows from \Cref{cor:main}. 
\end{proof}

\subsection{A consumer-resource model}
\label{sec:application-CRM}

Finally, we consider an ecological consumer-resource model (CRM) of microbes in a chemostat with cross-feeding. At a positive equilibrium, some entries on the diagonal of the Jacobian matrix are negative, while others are zero. As such, \Cref{cor:main} does not apply; nonetheless, the matrix is signed, so invoking \Cref{lem:nonzeroterms-subgraphs}, we give an alternative proof of the Competitive Exclusion Principle.

A CRM consists of resources and consumers, where a consumer uses a subset of resources to grow and reproduce, and it may produce a resource as waste. A common representation of a CRM is the \emph{community network} (\Cref{fig:ex-CRM-CN}), which describes the flow of energy or biomass through the community. Formally, a community network is a bipartite graph $\mc G = (V_x, V_y, E, E_\mrm{x'fd})$, where nodes in $V_x$ and $V_y$ represent resource and consumer species respectively. Consumer $y_j$ uses resource $x_i$ to grow if and only if $(x_i, y_j) \in E$; we assume every consumer must be able to consume at least one resource. Consumer $y_j$ produces $x_i$ as waste if and only if $(y_j, x_i) \in E_\mrm{x'fd}$. We assume that $E \cap E_\mrm{x'fd} = \emptyset$, i.e., a consumer cannot use its own waste as a resource. To distinguish between the two types of edges, we draw edges in $E_\mrm{x'fd}$ as a directed edge from consumer to resource, and those in $E$ as undirected.

\begin{figure}[h!t]
\centering 
\begin{subfigure}[b]{0.45\textwidth}
    \centering\figdraftboxwidth
    \raisebox{0cm}{
    \includegraphics{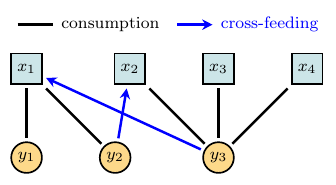}}
    \caption{}
    \label{fig:ex-CRM-CN}
\end{subfigure}
\begin{subfigure}[b]{0.45\textwidth}
    \centering\figdraftboxwidth
    \raisebox{0cm}{
    \includegraphics{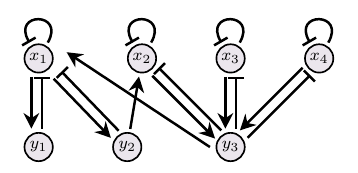}}
    \caption{}
    \label{fig:ex-CRM-Jgraph}
\end{subfigure}
\caption{
    (a) The community network with cross-feeding from \Cref{ex:CRM}. (b) The corresponding J-graph.
}
\label{fig:ex-CRM}
\end{figure}

We assume a community of $N$ microbial consumers and $M$ abiotic resources in a well-mixed chemostat. In other words, resources are supplied at constant rates, and microbes grow using a Holling-type functional response. Furthermore, we assume $M \geq N$, which is necessary for the existence of a non-degenerate coexistence equilibrium (\Cref{rmk:application-CRM-Jac}).  

In general, the ODE system takes the form 
\begin{alignat}{2}
    \dot{x}_i &= s_i - \mu_i x_i - \sum_{j=1}^N \alpha_{ij}   f_{ij}(x_i) y_j + \sum_{j=1}^N \gamma_{ij} y_j, \qquad &i =1,\ldots, M, \label{eq:application-CRM-ODE-R}\\ 
    \dot{y}_j &= y_j \left( \sum_{i=1}^M \beta_{ij} f_{ij}(x_i) - \delta_j \right) , \qquad &j = 1,\ldots, N, \label{eq:application-CRM-ODE-C}
\end{alignat}
with growth and decay/death rates  
\begin{gather*}
    \vv s = (s_i)_i \in \rrpp^M, \quad 
    \vv\mu = (\mu_i)_i \in \rrpp^M, \quad 
    \vv\delta = (\delta_j)_j \in \rrpp^N, 
\end{gather*}
consumption rates and response coefficients 
\begin{gather*}
    \tilde{\mm A} = [\alpha_{ij}]_{i,j}, \quad 
    \tilde{\mm B} = [\beta_{ij}]_{i,j} \in \rrp^{M\times N}
\end{gather*}
respectively. The parameter $\alpha_{ij}$, when positive, is the consumption rate of $x_i$ by $y_j$; the response coefficient $\beta_{ij} > 0$ is $\alpha_{ij}$ scaled by the conversion efficiency. We assume that $\alpha_{ij} > 0$ if and only if $\beta_{ij} > 0$, so the interaction-encoding matrices $\tilde{\mm A}$, $\tilde{\mm B}$ share the same sign pattern, which is the adjacency matrix of $(V_x, V_y, E)$.  
The cross-feeding terms are encoded in 
\begin{gather*}
    \tilde{\mm \Gamma} =[\gamma_{ij}]_{i,j} \in \rrp^{M\times N}, 
\end{gather*}
where we assume  $\gamma_{ij} > 0$ implies $\alpha_{ij} = \beta_{ij} = 0$. 

Finally, let $\tilde{\mm F} = [f_{ij}]_{i,j}$, where $f_{ij} = f_{ij}(x_i)$ is any positive and strictly increasing function whenever $\alpha_{ij} > 0$ (and $f_{ij}(x_i) = 0$ otherwise). A common choice is a Holling-type functional response, i.e., for each $\alpha_{ij} > 0$, we can choose between 
\eq{
    f_{ij}(x) = \left\{ 
        \begin{array}{cl}
            x & \text{Type-I}, \\[8pt]
            \ds \frac{x}{(V_{ij}+r)} & \text{Type-II}, \\[10pt]  
            \ds \frac{x^{m_{ij}}}{(V_{ij}+x^{m_{ij}})} & \text{Type-III} ,
        \end{array}
        \right. 
}
where $V_{ij}$, $m_{ij} > 0$, and $x = x_i$.  
If only Type-I responses are used, we have a generalized Lotka--Volterra system with first-order interactions, also known the MacArthur model in the context of CRMs~\cite{MacArthur1970}. Type-II and Type-III responses account for a limit on how fast a consumer can process a resource. We denote the matrix of derivatives $[f_{ij}'(x_i)]_{i,j}$ by $\tilde{\mm F}'$. 

\Cref{eq:application-CRM-ODE-R,eq:application-CRM-ODE-C} can be written in vector form as 
    \eqn{\label{eq:application-CRM-ODE-matrix}
    \begin{split}
        \frac{d\vv x}{dt}  &= \vv s - \diag(\vv\mu) \vv x - \left( \tilde{\mm A} \odot \tilde{\mm F} - \tilde{\mm \Gamma}\right)\vv y
        \\
        \frac{d\vv y}{dt} &= \left( (\tilde{\mm B}\odot \tilde{\mm F})^\top - \diag(\vv \delta) \right)\vv y . 
    \end{split}
    }

\begin{ex}
\label{ex:CRM}
    Consider the community network in \Cref{fig:ex-CRM-CN}. Assuming a mix of Type-I and Type-II functional responses, a CRM is 
    \eq{
        \begin{array}{ll}
        \ds \dot x_1 = s_1  - \mu_1 x_1 -  \frac{\alpha_{11}x_1y_1}{V_{11} + x_1}  -   \frac{\alpha_{12}x_1y_2}{V_{12} + x_1} + \gamma_1 y_3  \\[8pt] 
        \ds \dot x_2 = s_2 - \mu_2 x_2  - \alpha_{23} x_2 y_3 + \gamma_2 y_2 \\ 
        \ds \dot x_i = s_i - \mu_i x_i  - \alpha_{i3} x_i y_3 ,
            &\qquad i = 3,4\\[4pt]  
        \ds \dot y_j = y_j \left( \frac{\beta_{1j}x_1}{V_{1j} + x_1}  - \delta_j \right),
            & \qquad j = 1,2 \\[8pt]  
        \ds \dot y_3 = y_3 \left(\beta_{23} x_2 + \beta_{33} x_3 + \beta_{43} x_4  - \delta_3 \right) ,
        \end{array} 
    }
    where all parameters are positive. 
\end{ex}

For the system \eqref{eq:application-CRM-ODE-matrix}, we assume the existence of a non-degenerate coexistence equilibrium $(\vv x^*, \vv y^*) \in \rrpp^{M+N}$. Let $\odot$ denote the component-wise multiplication of two matrices of the same size. Then a coexistence equilibrium is a positive solution to 
\eqn{\label{eq:application-CRM-ss}  
    \left[ \tilde{\mm B} \odot \tilde{\mm F}(\vv x^*) \right]^\top \ones  = \vv \delta 
    \quad \text{and} \quad 
    \vv s - \diag(\vv \mu) \vv x^*  =   \left[  \tilde{\mm A}\odot \tilde{\mm F}(\vv x^*) - \tilde{\mm \Gamma} \right] \vv y^* , 
} 
where  $\ones$ is a constant vector of ones.

\subsubsection{The Jacobian matrix and J-graph at a coexistence equilibrium}
\label{sec:application-CRM_Jac}

In \Cref{sec:application-CRM-appendix}, we show that the Jacobian matrix $\mm J^*$ of \eqref{eq:application-CRM-ODE-matrix} at a coexistence equilibrium $(\vv x^*, \vv y^*) \in \rrpp^{M+N}$ is 
\eqn{\label{eq:application-CRM-Jac}
    \mm J^*
    = \begin{pNiceArray}{r|cr}
            - \mm D & \tilde{\mm \Gamma} - \mm I_1 \\ 
            \hline 
            \mm I_2 & \mm 0
        \end{pNiceArray}, 
}
where $\mm D \in \rr^{M\times M}$ is a positive diagonal matrix, and both $\mm I_1$, $\mm I_2^\top \in \rr^{M\times N}$ have the same sign pattern as $\tilde{\mm A}$. Furthermore, $\mm J^*$ is a signed matrix because $[\tilde{\mm \Gamma}]_{i,j} \neq  0$ implies $[\mm{I}_1]_{i,j} = 0$. See \cite{ButlerODwyer2018} for similar derivations. 

\begin{rmk}
\label{rmk:application-CRM-Jac}
    By the Schur's complement, $\det(\mm J^*) = (-1)^{M+N}\det(\mm D) \det(\mm I_2 \mm D^{-1}(\mm I_1-\tilde{\mm \Gamma}))$, so $\mm J^*$ is non-singular if and only if $\mm I_2 \mm D^{-1}(\mm I_1 - \tilde{\mm \Gamma})$ is non-singular, which is impossible if $N > M$. 
    Hence the $N \leq M$ assumption. While this inequality is often referred to as the Competitive Exclusion Principle (CEP), we prove below the \lq\lq extended\rq\rq\ version, which is an exclusion principle on \emph{subsets} of consumers. 
\end{rmk}

\begin{rmk}
\label{rmk:application-CRM-Jgraph}
    The J-graph of $\mm J^*$ can be derived directly from the community network $\mc G$.
    From \eqref{eq:application-CRM-Jac}, we know each resource node $x_i$ has a negative self-loop, whereas the consumer node $y_j$ does not. From the structure of $\mm I_1$, $\mm I_2$, we add an NFBL $x_i \pnFBL y_j$ in the J-graph for every consumption edge $(x_i, y_j) \in E$. From $\tilde{\mm \Gamma}$, we add a positive edge $y_j \posedge x_i$ for every cross-feeding edge $(y_j, x_i) \in E_\mrm{x'fd}$. See \Cref{fig:ex-CRM} for an example. Since $E \cap E_\mrm{x'fd} = \emptyset$, the resulting J-graph is still bipartite. 
\end{rmk}

The J-graph is defined at a coexistence equilibrium, and that some diagonal entries are zero. The hypotheses in \Cref{cor:main} are not met; therefore, we turn to \Cref{lem:nonzeroterms-subgraphs,lem:Cramer}.

\subsubsection{Competitive Exclusion Principle}
\label{sec:application-CRM-levin1970}

We give an alternative proof of the Competitive Exclusion Principle (CEP)~\cite{Levin1970}: that there cannot be a linearly stable coexistence equilibrium if there exists a subset of $r$ consumers that rely on strictly fewer than $r$ resources. Our proof is a straightforward application of \Cref{lem:nonzeroterms-subgraphs}\ref{lem:nonzeroterms-subgraphs-a} and Hall's Marriage Theorem (see, e.g., \cite[Theorem~5.4]{Saoub2021}). 

To state Hall's Theorem, we recall some terms from graph theory. 
Let $G = (V_x, V_y, E)$ be a bipartite graph with $|V_y| \leq |V_x|$.  
Given any subset $Y\subseteq V_y$, its set of neighbors is 
\eq{
    \mc N_{G}(Y) = \{ x_i \in V_x \colon  \exists y_j \in Y\text{ such that } (x_i, y_j) \in E \}. 
}
A \df{matching} on $G$ is a vertex-disjoint collection of edges; it is a \df{perfect matching on $V_y$} if it covers every node in $V_y$. A perfect matching on $G$ is one that covers the smaller of the vertex sets. Hall's Marriage Theorem states that a perfect matching on $V_y$ exists if and only if for every subset $Y \subseteq V_y$, we have $|Y| \leq |\mc N_{G}(Y)|$.

\begin{thm}
\label{thm:CRM-nondeg}
    Consider the CRM \eqref{eq:application-CRM-ODE-matrix} with $M$ resources and $N$ consumers where $M \geq N$. Suppose there exist $r$ consumers that rely on strictly fewer than $r$ resources.
    Then no non-degenerate coexistence equilibrium exists. In particular, no linearly stable coexistence equilibrium exists. 
\end{thm} 
\begin{proof} 
    Suppose the system has a non-degenerate positive equilibrium $(\vv x^*, \vv y^*)$. Let $\mbc J^*$ be the signed symbolic matrix consistent with $\mm J^*$. Since $\det(\mm J^*) \neq 0$, the polynomial $\det(\mbc J^*)$ is not identically zero; by \Cref{lem:nonzeroterms-subgraphs}\ref{lem:nonzeroterms-subgraphs-a}, the J-graph of $\mbc J^*$ has a disjoint cycle cover $H$. 
    The consumer nodes $y_j$ must be covered by FBLs in $H$ since they have no self-loops (\Cref{rmk:application-CRM-Jgraph}). 

    We claim that the disjoint cycle cover $H$ induces a perfect matching on the set of consumers. To be more precise, 
    let $G = (V_x, V_y, E)$ be the consumption part of the community network, without $E_\mrm{x'fd}$. The bipartite structure of the J-graph implies each FBL in $H$ induces a matching on $G$. For example, the FBL $(x_{i_1} \posedge y_{j_1} \neutedge \cdots \posedge y_{j_{\ell}} \neutedge x_{i_1})$ induces the matching $\{(x_{i_k}, y_{j_k}) \colon k=1,\ldots, \ell\}$. Since $H$ is vertex-disjoint cycle cover, it induces a \emph{perfect} matching on $G$.  
    By Hall's Marriage Theorem, $|Y| \leq |\mc N_{G}(Y)|$ for every subset of consumers $Y$. Therefore, the existence of a non-degenerate positive equilibrium implies there cannot be a subset of $r$ consumers that rely on strictly fewer than $r$ resources. 
\end{proof}

It can be argued that \Cref{thm:CRM-nondeg} is really characterizing structural degeneracy of a positive equilibrium. By Hall's Theorem, that $|Y| \leq |\mc N_{G}(Y)|$ for any $Y \subseteq V_y$ is equivalent to the existence of a perfect matching on $G$. This in turn is equivalent to the existence of a disjoint cycle cover for the J-graph%
    \footnote{Since the perfect matching on $G$ only uses the interaction edges $E$, each edge corresponds to the NFBL in the J-graph. This, along with self-loops on the resource nodes not covered, defines a disjoint cycle cover.
    }%
, i.e., $\det(\mbc J^*) \not\equiv 0$ as a polynomial. Nonetheless, this does not imply the Jacobian matrix of the system is non-singular at an equilibrium, when accounting for the entries of the matrix. 

Note that \Cref{thm:CRM-nondeg} does not depend on the cross-feeding edges. In other words, for the system \eqref{eq:application-CRM-ODE-matrix}, we cannot support more species by adding more cross-feeding. However, it has been observed that cross-feeding interactions may affect stability~\cite{ButlerODwyer2018}.   

\begin{ex}
\label{ex:CRM-2}
    Consider the CRM from \Cref{fig:ex-CRM}. Both $y_1$ and $y_2$ use $x_1$ as their only resource, so no perfect matching exists on the set of consumers.  By \Cref{thm:CRM-nondeg}, this system cannot have a linearly stable coexistence equilibrium. This can also be seen from the J-graph in \Cref{fig:ex-CRM-Jgraph}. A disjoint cycle cover must include $y_1 \posedge x_1$, but then there cannot be a vertex-disjoint cycle involving $y_2$. Therefore, $\det(\mm J^*) = 0$ at any coexistence equilibrium. 

    The mapping between perfect matchings and disjoint cycle covers breaks down if we include $E_\mrm{x'fd}$ in the bipartite graph to which we applied Hall's Theorem. Essentially we are limited by the edges from resources to consumers.  
    For example, there is no way to complete the perfect matching $\{ (y_1,x_1), (y_2,x_2), (y_3,x_3)\}$ to a disjoint cycle cover of the J-graph without including either $x_2 \posedge y_1$ or $x_3 \posedge y_1$, neither of which is an edge. 
\end{ex}

\section{Discussion}
\label{sec:conclusion}

In this paper, we investigated necessary motifs underlying non-monotonic steady state responses. We proved that a stable non-monotonic response requires either an IFFL or a combination of a PFBL and an NFBL, where the PFBL is vertex-disjoint from an input--output path and the NFBL lies on a walk from input to output.  

The proof relies on two mathematical theories. The first is an application of monotone systems theory~\cite[Theorem~1]{AscensaoDattaHanciogluSontagEtAl2016}, which gives a graph-theoretic condition for when the transient dynamics is not monotonic. The second is the algebraic and combinatorial theory of signed matrices~\cite{BrualdiShader1995}. 
The idea of analyzing a steady state response based on the graph of its Jacobian matrix is certainly not new. 
\Cref{prop:main-alg} echoes the theory by Golubitsky and his colleagues on quasi-adaptation (infinitesimal homeostasis)~\cite{WangHuangAntoneliGolubitsky2021, MadeiraAntoneli2022, GolubitskyStewart2018, GolubitskyStewartAntoneliHuangWang2020}. Their framework starts with the matrix, while ours is focused on the graph. In \Cref{sec:app-trimgraph}, we discuss the connection in more depth, and provide a graph-theoretic interpretation to the objects introduced in their works.

We applied our results to three biological contexts, each of which typically admits a representation of the system as a graph.
Regulatory networks used in signaling pathways and gene expressions are essentially the graph of the Jacobian matrix $\mm J$. Therefore, motifs like feedback and feedforward loops have precise mathematical interpretation in terms of $\mm J$. By way of an example, we found a model of T-cell activation from \cite{LeverLimKrugerNguyenEtAl2016} that exhibits biphasic response without an IFFL. Finally, for consumer-resource models, the standard representation as a community network is closely related to the graph of the Jacobian matrix. We combined our framework and a combinatorial theorem to give an alternative proof of the Competitive Exclusion Principle.

The main strength of the framework presented here is that the conditions we obtained are structural; we do not need to solve for the steady state explicitly, nor do we need parameters. However, this strength is also a major limitation. The necessary conditions we obtained are far from tight, especially when the Jacobian matrix $\mm J$ has additional structure; see e.g., \Cref{ex:PA-suff_notNecessary}. 
Different classes of parametric ODE systems can be approached differently.  

First, as we saw in \Cref{sec:application}, when signed directed graphs naturally appear in biological applications, the entries in $\mm J$ tend to be algebraically independent, in the sense that every entry has its own private parameter. In such cases, the framework developed in this paper can be used to design systems that exhibit biphasic response. The term(s) with the wrong sign in the non-principal minor can guide choices of parameter values that enable biphasic response. 

Second, when system comes from a reaction network, the system can be factored $\dot{\xx} = \mm N \vv \nu(\xx; u)$, where the stoichiometric matrix $\mm N$ is a fixed numeric matrix. Thus, $\mm J = \mm N D_{\xx} \vv\nu$. In this case, one can leverage the structure given by this factorization. Perfect adaptation and sign-sensitivity have been explored in this context~\cite{OkadaMochizuki2016, OkadaMochizuki2017, GuptaKhammash2022, Feliu2019}. A more refined version of the machinery developed here may lead to a definition of \lq\lq feedforward loops\rq\rq\ for reaction network models.  
In conclusion, sign- and zero-sensitivity have been investigated by multiple groups, whether for general ODEs or those that arise from reaction networks. As the importance of non-monotonic responses in biology has been recognized, a mathematical theory of non-monotonic responses can prove to be useful for clinicians and biologists.

\section*{Acknowledgments} 

EDS was partially supported by grants AFOSR FA9550-22-1-0316 and NSF/DMS-2052455.

\section*{Disclosure}

Parameters for \Cref{ex:appl-Tcell-noIFFL,ex:GRN-4node} were found by Claude Opus 5.5. The models were subsequently verified by numerical simulation, which generated the figures.

\putbib[cit.bib]
\end{bibunit}


\clearpage 
\setcounter{page}{0}
\appendix
\crefalias{section}{sisection}
\crefalias{subsection}{sisubsection}
\crefalias{subsubsection}{sisubsubsection}
\crefalias{figure}{sifig}   
\crefalias{table}{sitab}  
\crefalias{equation}{sieq}  

\renewcommand{\thesection}{\Alph{section}} 
\renewcommand{\thesubsection}{\thesection.\arabic{subsection}}
\renewcommand{\thesubsubsection}{\thesubsection.\arabic{subsubsection}}

\renewcommand\thefigure{\thesection-\arabic{figure}}
\setcounter{figure}{0}   
\renewcommand\thetable{\thesection-\arabic{table}}
\setcounter{table}{0} 
\renewcommand\theequation{\thesection-\arabic{equation}}
\setcounter{equation}{0} 
	\setcounter{page}{1}
    
\section*{\Huge Appendices}
{\section*{Contents}
    \setlength{\parskip}{-2pt}
    
    \startcontents[sections]
    \printcontents[sections]{l}{1}{\setcounter{tocdepth}{2}}
}
\vspace{-5pt}
\noindent\rule{\textwidth}{0.5pt}


\begin{bibunit}[alpha]

\section{Proof that \eqref{eq:assumption-A2}  implies \eqref{eq:assumption-A1}}
\label{sec:A2-implies-A1}

\setcounter{equation}{0} 

We prove that assumption \eqref{eq:assumption-A2} implies the more technical \eqref{eq:assumption-A1}. 
Under \eqref{eq:assumption-A2}, the Jacobian matrix $\mm J_{u_0}$ evaluated at a steady state $(\xx_0; u_0)$ is Hurwitz.   
By \cite[Theorem~18]{Sontag1998}, there exists a symmetric, positive-definite matrix $\mm P = \mm P(u_0)$ such that 
\eqn{ 
\label{eq:P_mtx-Lyap_eqn}
    \mm J_{u_0}^\top \mm P + \mm P \mm J_{u_0} = -\mm{Id},  
}
where explicitly, $\mm P = \int_0^\infty e^{t\mm J_{u_0}^\top} e^{t\mm J_{u_0}}  \, dt$.
Furthermore, the quadratic function $V_{u_0}(\xx) = \xx^\top \mm P \xx$ is a Lyapunov function for the linearized system $\dot{\xx} = \mm J_{u_0} \xx$. 

Since \eqref{eq:assumption-A2} implies \eqref{eq:assumption-A0}, let $\xx^*$ be a branch of steady state, as in \eqref{eq:assumption-A0}, and let $\tilde{\xx} = \xx - \xx^*(u)$. 
To prove that $V_{u}(\xx) = \tilde{\xx}^\top \mm P \tilde{\xx}$ is a Lyapunov function for $\dot{\xx} = \vv f(\xx; u)$, we first derive a matrix-valued function based on \eqref{eq:P_mtx-Lyap_eqn}.

\begin{lem}
\label{app-lem:inequality-M}
    Let $\mm P$ be the symmetric and positive-definite solution to \eqref{eq:P_mtx-Lyap_eqn}. 
    Consider the function $\mm M \colon \uInt \to M_n(\rr)$, $u \mapsto \mm J_u^\top \mm P + \mm P \mm J_{u}$. Then  
    \begin{enumerate}[label={(\alph*)}]
        \item $\mm M$ is continuous and $\mm M(u_0) = -\mm{Id}$,  and 
        \item for any $\beta \in (0,1)$, there exists $\delta > 0$ such that for all $|u-u_0| < \delta$, the matrix $\mm M(u) + \beta \mm{Id}$ is symmetric and negative-definite, with $\xx^\top \mm M(u) \xx \leq -\beta \norm{\xx}^2$ for all $\xx \in \rr^n$. 
    \end{enumerate}
\end{lem}
\begin{proof}
    The map $u \mapsto \mm M(u)$ is continuous because $\vv f(\xx; u)$ is $C^1$ and matrix operations are continuous. By \eqref{eq:P_mtx-Lyap_eqn}, $\mm M(u_0) = -\mm{Id}$. 
    Since $\mm P^\top = \mm P$, both $\mm M(u)$ and $\mm M(u) + \beta \mm{Id}$ are symmetric for all $u$ and $\beta$. 
    Thus, the eigenvalues of $\mm M(u)$ are all real. Fix any $\beta \in (0,1)$. 
    Because $\sigma(\mm M(u_0)) = \{ -1\}$, by the continuity of eigenvalues, there exists $\delta > 0$ such that for all $|u-u_0|<\delta$, every eigenvalue $\lambda$ of $\mm M(u)$ satisfies $\lambda < -\beta$. Consequently, $\xx^\top \mm M(u) \xx \leq -\beta \norm{\xx}^2$ for all $\xx \in \rr^n$, and so $\mm M(u) + \beta \mm{Id}$ is negative-definite. 
\end{proof}

Let $\uInt_0 = (u_0-\delta, u_0+\delta)$ with $\delta$ from \Cref{app-lem:inequality-M}, which depends on $u_0$ and $\beta$. By moving to a smaller open interval if necessary, we may assume  that $(u_0-\delta, u_0+\delta)$ is contained within the interval $\uInt_0$ from \eqref{eq:assumption-A0}. Let $\uInt_1 = [u_0 - \delta/2, u_0 + \delta/2]$, and let $r > 0$ be such that 
$K = \{ (\xx; u) \colon u \in \uInt_1 , \norm{\xx - \xx^*(u)} \leq r \}$ is a compact subset of $\X \times \uInt$. 

Next, we linearize $\vv f(\xx; u)$ about $(\xx^*(u); u)$. Let 
\eq{
    \omega(s) = \sup\left\{ 
        \norm{\mm J(\xx; u) - \mm J(\xx'; u') }
        \colon (\xx, u), (\xx', u') \in K \text{ and } \norm{\xx-\xx'} + |u-u'| \leq s
        \right\}     
}
be the modulus of continuity of $\mm J$. Note that $\omega$ is non-decreasing and $\lim_{s\to 0^+} \omega(s) = 0$.
Then $\vv f(\xx; u) = \mm J_u \tilde{\xx} + R(\xx; u)$, where the remainder term satisfies  
\eqn{\label{eq:linearize-error}
    \norm{R(\xx; u)} \leq \omega(\norm{\tilde{\xx}})  \norm{\tilde{\xx}} . 
} 

Finally, for any $(\xx, u) \in K$, define $V_u(\xx) = \tilde{\xx}^\top \mm P \tilde{\xx}$ where $\tilde{\xx} = \xx - \xx^*(u)$.

\begin{prop}
\label{app-prop:Lyapunov}
    There exists $\epsilon > 0$, independent of $u \in \uInt_1$, such that $V_u(\xx)$ is a Lyapunov function for $\dot{\xx} = \vv f(\xx; u)$ on $B_\epsilon(\xx^*(u))$. 
    That is, for any $u \in \uInt_1$ and all $\vv a \in B_\epsilon(\xx^*(u))$, the solution $\phi(t, \vv a, u)$ to the initial value problem $\dot{\xx} = \vv f(\xx; u)$, $\xx(0) = \vv a$ satisfies $\lim_{t\to\infty} \phi(t, \vv a, u) = \xx^*(u)$. 
\end{prop}
\begin{proof} 
    Since $D_{\xx} V_u(\xx) 
    = 2 \, \mm P \tilde{\xx}$, the time derivative of $V_u$ along a trajectory $\xx(t)$ of the system \eqref{eq:syst_main} is  
    \eq{
        \frac{d}{dt} V_u (\xx(t))
        &= \left( 2\mm P \tilde{\xx}  \right)^\top \left( 
            \mm J_u \tilde{\xx} + R(\xx; u)
            \right) 
        \\&= \tilde{\xx}^\top \mm J_u^\top \mm P \tilde{\xx} + \tilde{\xx}^\top \mm P \mm J_u \tilde{\xx} + 2 \tilde{\xx}^\top \mm P R(\xx; u)
        \\&= \tilde{\xx}^\top \mm M(u) \tilde{\xx} + 2 \tilde{\xx}^\top \mm P R(\xx;u) 
        \\&\leq - \beta \norm{\tilde{\xx}}^2 + 2 \norm{\mm P} \omega(\norm{\tilde{\xx}}) \norm{\tilde{\xx}}^2 , 
    }
    where the inequality holds for all $(\xx, u) \in K$ by \Cref{app-lem:inequality-M} and \eqref{eq:linearize-error}. For the fixed $\beta \in (0,1)$ from \Cref{app-lem:inequality-M}, choose $\epsilon' > 0$ such that $2 \norm{\mm P} \omega(\epsilon')  < \beta $. 
    Then for any $0 < \norm{\tilde{\xx}} < \epsilon'$, we have  $\dot V_u (\xx(t)) \leq \norm{\tilde{\xx}}^2\left(  2 \norm{\mm P} \omega(\epsilon') - \beta \right) < 0$, with equality if and only if $\tilde{\xx} = \vv 0$, i.e.,  $\xx(t) = \xx^*(u)$. 
    Finally, let $\epsilon = \epsilon' \sqrt{\lambda_{\min}/\lambda_{\max}}$, where $\lambda_{\min}$ and $\lambda_{\max}$ are the smallest and largest eigenvalues of $\mm P$ respectively. Then for any $\vv a \in B_\epsilon(\xx^*(u))$, we have $\phi(t, \vv a, u) \to \xx^*(u)$. 
\end{proof}

Note that $\epsilon$ depends on $u_0$ and $\uInt_1$, via $\mm P$, $\beta$, $\omega$. In particular, $\epsilon$ is constant on $\uInt_1$, so $u \mapsto \epsilon(u)$ is constant (thus continuous) on $\uInt_1$. Then assumption~\eqref{eq:assumption-A1} is true where $\uInt_0$ is replaced with, say, the interior of $\uInt_1$.

\section{Proof that \eqref{eq:assumption-A1} implies convergence from nearby steady states}
\label{sec:A1-implies-convergence}

\setcounter{equation}{0} 

We prove that \eqref{eq:assumption-A1} is sufficient to guarantee that the system will still converge to the steady state curve following a small perturbation in $u$. 

\begin{prop}
\label{prop:A1-implies-convergence}
    Consider the system \eqref{eq:syst_main} under assumption \eqref{eq:assumption-A1}. Let $B_{\epsilon(u)}(\xx^*(u)) \subseteq \X$ be the open ball around $\xx^*(u)$ from \eqref{eq:assumption-A1} that lies within the basin of attraction of $\xx^*(u)$. 
    For any $u_0 \in \uInt_0$, there exists $\delta > 0$ such that for all $\abs{u - u_0} < \delta$, we have $\xx^*(u_0) \in B_{\epsilon(u)}(\xx^*(u))$. 
\end{prop}
\begin{proof}
    By the continuity of $\epsilon(u)$, there exists $\delta_1 > 0$ such that whenever $\abs{u-u_0} < \delta_1$, we have $\abs{\epsilon(u) - \epsilon(u_0)} < \frac{\epsilon(u_0)}{2}$, which implies $\epsilon(u) > \frac{\epsilon(u_0)}{2}$. By \eqref{eq:assumption-A0}, $\xx^*(u)$ is continuous, so there exists $0 < \delta \leq \delta_1$ such that whenever $|u-u_0| < \delta$, we have $\norm{\xx^*(u)-\xx^*(u_0)} < \frac{\epsilon(u_0)}{2} < \epsilon(u)$. Hence, $\xx^*(u_0) \in B_{\epsilon(u)}(\xx^*(u))$. 
\end{proof}

\section{Detailed analysis of \Cref{ex:PFBL-not_suff}} 
\label{app:ex-PFBL-not_suff}
\newcommand{\parsgnA}{{\alpha}}
\newcommand{\parsgnB}{{\beta}}

\setcounter{equation}{0} 

The following calculations were done in Mathematica. We first solve for the steady states. From $\dot x_1 = 0$, we have $x_1^* = u$. From $\dot x_4 = 0$, we get $x_4^* = u^2 + x_3^*$. Similarly, from $\dot x_3 = 0$, we have $(x_2^*)^2 = x_3^*/a$, where $a > 0$. Substitute this relation into $\dot x_2 = 0$ yields $x_3^2 + \left( 1 - 2/a \right) x_3 + u = 0$, which gives  
\eq{
    x_3^* = \left( \frac{1}{a} - \frac{1}{2} \right) \pm \sqrt{\left( \frac{1}{a} - \frac{1}{2} \right)^2 - u}. 
}
A real solution for $x_3^*$ exists if and only if $u \leq u_c = \left(\frac{1}{a} - \frac{1}{2} \right)^2 $. Since  $0 <  a  < 2$, $x_3^*$ yields two possible branches of equilibria. Because $x_2^* = \pm \sqrt{ x_3^*/a}$, there are four branches of equilibria in total, two of which are negative. 

In summary, the real steady states are given by 
\eq{
    x_1^*( u)  &= u \\
    x_2^*(u) &= \parsgnA \sqrt{\frac{2-a}{2a^2} + \frac{\parsgnB}{a} \sqrt{ \left( \frac{2-a}{2a} \right)^2 - u }}  \\
    x_3^*(u) &= \frac{2-a}{2a} + \parsgnB \sqrt{ \left( \frac{2-a}{2a} \right)^2 - u } \\
    x_4^*(u) &= u^2 + \frac{2-a}{2a} + \parsgnB \sqrt{ \left( \frac{2-a}{2a} \right)^2 - u }, 
}
where $\parsgnA$, $\parsgnB \in \{-1, 1\}$. The positive branches correspond to $\parsgnA = 1$. See \Cref{fig:app-ex6-ss} for plots of steady states for $a \in\{ 0.6, 0.8, 1.0\}$.

\begin{figure}[h]
\centering 
    \includegraphics[height=2in]{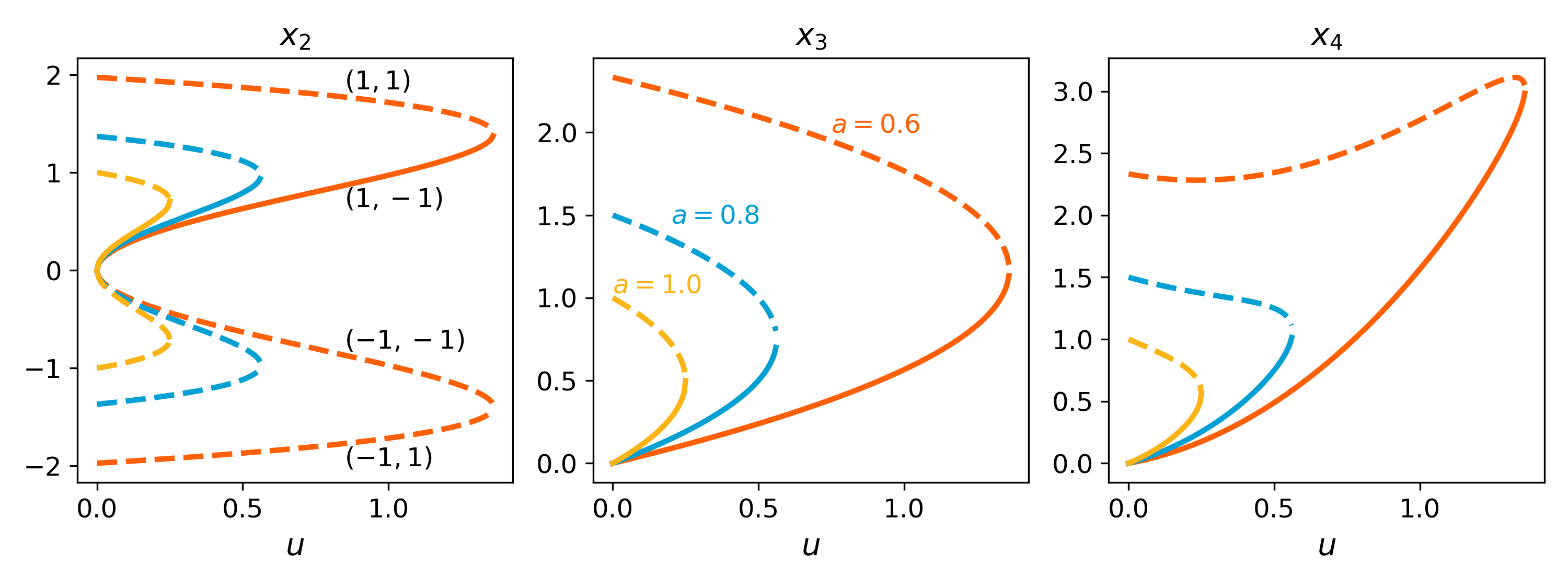} 
\caption{
    The four steady states branches for the system in \Cref{ex:PFBL-not_suff}. In the plot for $x_2^*$, from top to bottom, the branches correspond to $(\alpha,\beta) = (1,1)$, $(1,-1)$, $(-1,-1)$, and $(-1,1)$. The positive branches correspond to $\alpha = 1$, while the stable branch corresponds to $(\alpha,\beta) = (1,-1)$. 
}
\label{fig:app-ex6-ss}
\end{figure}

The Jacobian matrix of the system is 
\eq{
\mm J = 
    \begin{pmatrix}
     -1 & 0 & 0 & 0 \\
     1 & -4 x_2 & 2 x_3+1 & 0 \\
     0 & 2 a x_2 & -1 & 0 \\
     2 x_1 & 0 & 1 & -1 \\
    \end{pmatrix}, 
}
which has fixed sign structure on the domain $\rrpp^4$, and for $0 < a < 2$. Direct computation yields the eigenvalues:
\eq{
    \lambda_1 &= \lambda_2 = -1 \\ 
    \lambda_3 &= \frac{1}{2}\left(-(1+4 x_{2})  -\sqrt{(1+4x_2)^2 - 8x_2(2-2ax_3-a)}\right)  \\
    \lambda_4 &= \frac{1}{2}\left(-(1+4 x_{2}) + \sqrt{ (1+4x_2)^2 - 8x_2(2-2ax_3-a) }\right). 
}
At steady state, we have 
\eq{
    - 8x_2^*(2-2ax_3^*-a)
    &= 16 a x_2^* \left( x_3^* - \frac{2-a}{2a} \right) 
    \\&= 16 a x_2^* \left( \beta \sqrt{\left( \frac{2-a}{2a} \right)^2 - u} \right) 
    \\&= 16  \alpha \beta \cdot a |x_2^*|  \sqrt{u_c - u }, 
}
where $\alpha, \beta \in \{ -1, 1\}$ specify the branch. 
Thus, 
\eq{
    2 \lambda_{4} &= -(1+4x_2^*) + \sqrt{(1+4x_2^*)^2 + 16 \alpha \beta \cdot a |x_2^*| \sqrt{u_c - u} } .
}

Focusing on the positive equilibria ($\alpha = 1$), the branch corresponding to $(\alpha,\beta) = (1,1)$ is always unstable, since 
\eq{
    2 \lambda_4 = -(1+4x_2^*) + \underbrace{\sqrt{(1+4x_2^*)^2 + 16 a |x_2^*| \sqrt{u_c - u} }}_{>-(1+4x_2^*)} > 0. 
}
For the remaining positive branch, $(\alpha,\beta) = (1,-1)$, we have 
\eq{
    2 \lambda_{3,4} = \underbrace{-(1+4x_2^*)}_{<-1} \pm \sqrt{(1+4x_2^*)^2 - 16a  |x_2^*| \sqrt{u_c - u} } . 
}
If $(1+4x_2^*)^2 - 16 a |x_2^*| \sqrt{u_c - u} \geq 0$, then $\lambda_{3,4} < 0$. If $(1+4x_2^*)^2 - 16 a |x_2^*| \sqrt{u_c - u} < 0$, then $\mrm{Re}(\lambda_{3,4}) = -(1+4x_2^*) < 0$. Therefore,  only one of the two positive branches is stable. 

Finally, to see that biphasic response is possible only when $a < \sqrt{3}-1$, we use the steady state parametrization of $x_4^*$ with $\alpha = \beta = 1$. The derivative $d_u x_4^* = 2u - (2\sqrt{u_c - u})^{-1}$, which has a root if and only if $u_c -u = 1/(16u^2)$. For $u_c$ sufficiently small, there is no positive solution, but for larger $u_c$, there are two positive solutions. 

Finding the smallest $u_c$ for which $u_c - u = 1/(16 u^2)$ has a real positive solution is a calculus problem. Let $u_m$ be the unique $u_c$ for which $u_c - u$ meets $1/(16 u^2)$ tangentially; we claim $u_m = 3/4$. From intersection of the two curves and their tangents, we have two equations:  $u_m -u = 1/(16 u^2)$ and $-1 = -2/(16u^3)$. This simplifies to $u_m - u = \frac{u}{2} \frac{1}{16 u^2} = \frac{u}{2}$. In other words, the intersection point between the two curves occurs at $u = 2u_m/3$. Back substituting, we have 
\eq{
    &\qquad u_m -u = \frac{1}{16 u^2}
    \\ \iff & \qquad 
    \frac{u_m}{3} = \frac{1}{16} \left( \frac{9}{16u_m^2} \right)
    \\ \iff & \qquad 
    u_m^3 = \left( \frac{3}{4} \right)^3. 
}
Therefore, biphasic response is possible when $u_c > 3/4$. Finally, from $u_m = \left( \frac{2-a}{2a}\right)^2$, we can solve for the smallest positive $a = \frac{2}{\sqrt 3 + 1} = \sqrt{3} - 1$.

\section{Condition for biphasic response in \Cref{ex:CRN-rateconstant}} 
\label{app:ex-CRN-rateconstant}

\setcounter{equation}{0} 

We prove that $4 \kk_1 \kk_5 > 3 \kk_2^2$ is a necessary and sufficient condition for the system in \Cref{ex:CRN-rateconstant} to exhibit stable biphasic response. 
Stability follows from the positive-determinant and negative-trace conditions of the Jacobian matrix. We focus on the biphasic response. 

In the main text, we showed that the sensitivity changes sign at some $u > 0$ if and only if $\kk_2 - 2\kk_5 x_1^*(u)$ changes sign as $u$ increases. 
Setting $\dot x_1 =  \kk_1 - \kk_2 x_1 - \kk_5 x_1^2 - u x_1$ to zero, the unique positive steady state is given by 
\eq{
    2\kk_5 x_1^* = -(\kk_2 + u) + \sqrt{(\kk_2 + u)^2 + 4\kk_5\kk_1} . 
}
Then $\kk_2 - 2\kk_5 x_1^*(u)$ has a root if and only if 
\eq{\begin{array}{lrl}
    & \kk_2 &=  -(u+\kk_2) + \sqrt{(u+\kk_2)^2 + 4 \kk_1 \kk_5} \\
    \iff & \qquad u + 2 \kk_2 &= \sqrt{(u+\kk_2)^2 + 4 \kk_1 \kk_5} \\ 
    \iff &\qquad (u + 2 \kk_2)^2 &= (u+\kk_2)^2 + 4 \kk_1 \kk_5 \\
    \iff & \qquad u^2 + 4 u\kk_2 + 4 \kk_2^2 &= u^2 + 2 u\kk_2 + \kk_2^2 + 4 \kk_1 \kk_5 \\
    \iff & \qquad 2u \kk_2 + 3 \kk_2^2 &= 4 \kk_1 \kk_5 \\
    \iff & \qquad u &= \frac{4 \kk_1 \kk_5 - 3 \kk_2^2}{2\kk_2 }, 
\end{array}
}
which yields a positive solution $u$ if and only if $4 \kk_1 \kk_5 > 3 \kk_2^2$. Moreover, as $u$ increases past this threshold, the sensitivity changes sign.

\section{General kinetics for regulations in a signaling pathway}
\label{sec-app:rmk:FFL-motif}

    In \Cref{rmk:FFL-motif}, we claimed that an IFFL according to our definition must contain within it an IFFL with disjoint interior. We provide an argument for this here. 
    
    Let $(P, Q)$ be an IFFL, and let $S$ be the set of common nodes between them. Order $S = \{s_0 \coloneqq u, s_1, s_2,\ldots, s_k, s_f \coloneqq x_n\}$ according to $P$, and similarly let $\{s_0, s_{\sigma(1)},\ldots, s_{\sigma(k)}, s_f\}$ be the ordering with respect to $Q$ for some fixed $\sigma \in S_k$. 
    Denote by $<_P$ and $<_Q$ the orderings of $S$ according to $P$ and $Q$ respectively. 
    There are two cases to consider.
    
    In the case $\sigma = \mrm{id}$, i.e., the same ordering, the IFFL $(P,Q)$ can be decomposed into segments, each of which is either a path, or a feedforward loop with disjoint interior from $s_i$ to $s_{i+1}$. Since $P \neq Q$, there must be at least one feedforward loop, and because $\sgn(P) \neq \sgn(Q)$, an odd number of these must be an IFFL. 

    Now suppose the orders on $S$ are not the same. If for some $\ell_0 \in [k]$, we have $\sigma(i) = i$ for all $i \leq \ell_0$, then we may treat the initial segment from $s_1$ to $s_{\ell_0}$ similar to the previous case. If there is an IFFL, we are done; otherwise, proceed to the next step. Without loss of generality, let $\ell_0$ be the largest such index. 

    To simplify notation, treat $\ell_0$ as the input node, i.e., we relabel $s_{\ell_0+i}$ as $s_i$ for every $i \geq 0$. Then $\sigma(1) \neq 1$. Decompose each I/O path into three parts: $P \coloneqq P_1 \cup P_2 \cup P_3$, where $P_1 = (s_0 \neutedge s_1)$, $P_2$ is the path from $s_1$ to $s_{\sigma(1)}$ along $P$, and $P_3$ is the path from $s_{\sigma(1)}$ to $s_f$ along $P$. Similarly decompose $Q \coloneqq Q_1 \cup Q_2 \cup Q_3$. See \Cref{fig:rmk:FFL-motif}\hyperref[fig:rmk:FFL-motif]{(a-b)} for a cartoon. 

    Then there are at least two other I/O paths defined by $P' \coloneqq Q_1 \cup P_3$ and $Q' \coloneqq P_1 \cup Q_3$ (\Cref{fig:rmk:FFL-motif}\hyperref[fig:rmk:FFL-motif]{(c-d)}). Note that these are paths, since the nodes along $Q_3$ must appear along $P_2 \cup P_3$; similarly the nodes of $P_3$ appear on the complement of $Q_1$. Therefore, besides $(P,Q)$, we have four feedforward loops. 

    Consider the feedforward loop defined by $P$ and $P'$ (\Cref{fig:rmk:FFL-motif}\hyperref[fig:rmk:FFL-motif]{(a,c)}). If $P_1 \cup P_2$ and $Q_1$ form an IFFL, then we have an IFFL with disjoint interior, and we are done. If this forms a CFFL, then $\sgn(P_1) \sgn(P_2) = \sgn(Q_1)$. Substituting this into $\sgn(P) = - \sgn(Q)$, we have $\sgn(P_3) = -\sgn(Q_2) \sgn(Q_3)$, and so $P_3$ and $Q_2 \cup Q_3$ form an IFFL from $s_{\sigma(1)}$ and $s_f$ (\Cref{fig:rmk:FFL-motif}\hyperref[fig:rmk:FFL-motif]{(b-c)}. By induction on the number of nodes on the IFFL, we may repeat this argument until we find an IFFL with disjoint interior. 

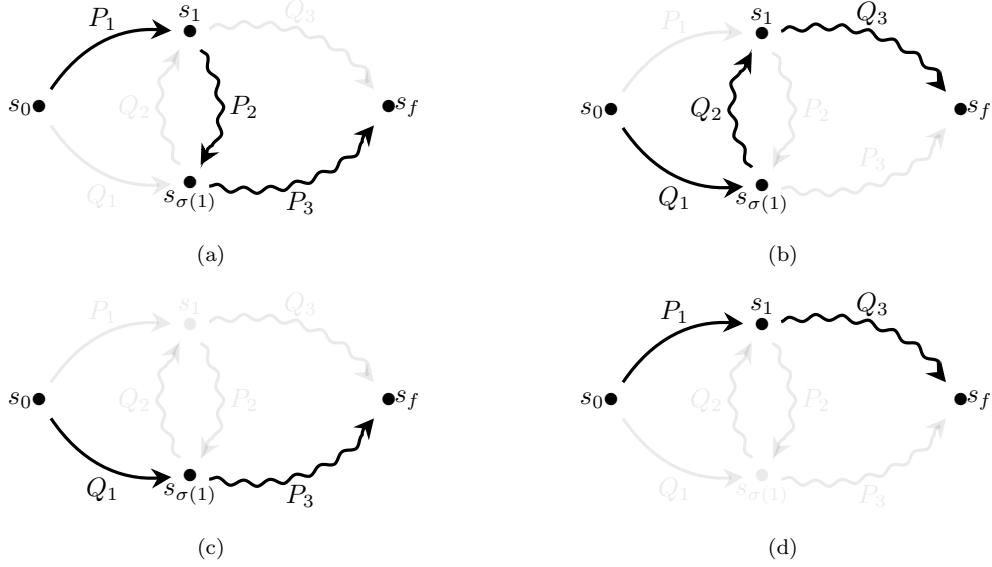
\begin{figure}[th]
\centering 
\begin{subfigure}[b]{0.45\textwidth}
    \centering\figdraftboxwidth
    \begin{tikzpicture}
    \tikzset{decoration={snake,amplitude=.4mm}}
    \begin{scope}[shift={(0,0)}] 
        \node at (0,1.755) {};
        \node (0) at (0.5,0) {$\bullet$}; 
            \node at (0.west) {\footnotesize $s_0$}; 
        \node (1) at (2.5,1) {$\bullet$}; 
            \node at (1.north) {\footnotesize $s_1$}; 
        \node (1b) at (2.5,-1) {$\bullet$}; 
            \node at (1b.south) {\footnotesize $s_{\sigma(1)}$}; 
        \node (f) at (5.125,0) {$\bullet$}; 
            \node at (f.east) {\footnotesize $s_f$}; 
        \draw[posEdge] (0) to[bend left]  node [midway, above] {\footnotesize $P_1$} (1);
        \draw[posEdge, opacity=0.085] (0) to[bend right]  node [midway, below] {\footnotesize $Q_1$} (1b); 
        \draw[posEdge, decorate, opacity=0.085] (1) to[bend left]  node [midway, above] {\footnotesize $Q_3$} (f);
        \draw[posEdge, decorate] (1b) to[bend right]  node [midway, below] {\footnotesize $P_3$} (f); 
        \draw[posEdge, decorate] (1) to[bend left]  node [midway, right] {\footnotesize $P_2$} (1b);
        \draw[posEdge, decorate, opacity=0.085] (1b) to[bend left]  node [midway, left] {\footnotesize $Q_2$} (1);
    \end{scope}
    \end{tikzpicture}
        \caption{}
    \label{fig:rmk:FFL-motif-a}
\end{subfigure}
\begin{subfigure}[b]{0.45\textwidth}
    \centering\figdraftboxwidth
    \begin{tikzpicture}
    \tikzset{decoration={snake,amplitude=.4mm}}
    \begin{scope}[shift={(0,0)}] 
        \node at (0,1.755) {};
        \node (0) at (0.5,0) {$\bullet$}; 
            \node at (0.west) {\footnotesize $s_0$}; 
        \node (1) at (2.5,1) {$\bullet$}; 
            \node at (1.north) {\footnotesize $s_1$}; 
        \node (1b) at (2.5,-1) {$\bullet$}; 
            \node at (1b.south) {\footnotesize $s_{\sigma(1)}$}; 
        \node (f) at (5.125,0) {$\bullet$}; 
            \node at (f.east) {\footnotesize $s_f$}; 
        \draw[posEdge, opacity=0.085] (0) to[bend left]  node [midway, above] {\footnotesize $P_1$} (1);
        \draw[posEdge] (0) to[bend right]  node [midway, below] {\footnotesize $Q_1$} (1b); 
        \draw[posEdge, decorate] (1) to[bend left]  node [midway, above] {\footnotesize $Q_3$} (f);
        \draw[posEdge, decorate, opacity=0.085] (1b) to[bend right]  node [midway, above] {\footnotesize $P_3$} (f); 
        \draw[posEdge, decorate, opacity=0.085] (1) to[bend left]  node [midway, right] {\footnotesize $P_2$} (1b);
        \draw[posEdge, decorate] (1b) to[bend left]  node [midway, left] {\footnotesize $Q_2$} (1);
    \end{scope}
    \end{tikzpicture}
        \caption{}
    \label{fig:rmk:FFL-motif-b}
\end{subfigure}

\begin{subfigure}[b]{0.45\textwidth}
    \centering\figdraftboxwidth
    \begin{tikzpicture}
    \tikzset{decoration={snake,amplitude=.4mm}}
    \begin{scope}[shift={(0,0)}] 
        \node at (0,1.755) {};
        \node (0) at (0.5,0) {$\bullet$}; 
            \node at (0.west) {\footnotesize $s_0$}; 
        \node[opacity=0.085] (1) at (2.5,1) {$\bullet$}; 
            \node[opacity=0.085] at (1.north) {\footnotesize $s_1$}; 
        \node (1b) at (2.5,-1) {$\bullet$}; 
            \node at (1b.south) {\footnotesize $s_{\sigma(1)}$}; 
        \node (f) at (5.125,0) {$\bullet$}; 
            \node at (f.east) {\footnotesize $s_f$}; 
        \draw[posEdge, opacity=0.085] (0) to[bend left]  node [midway, above] {\footnotesize $P_1$} (1);
        \draw[posEdge] (0) to[bend right]  node [midway, below] {\footnotesize $Q_1$} (1b); 
        \draw[posEdge, decorate, opacity=0.085] (1) to[bend left]  node [midway, above] {\footnotesize $Q_3$} (f);
        \draw[posEdge, decorate] (1b) to[bend right]  node [midway, below] {\footnotesize $P_3$} (f); 
        \draw[posEdge, decorate, opacity=0.085] (1) to[bend left]  node [midway, right] {\footnotesize $P_2$} (1b);
        \draw[posEdge, decorate, opacity=0.085] (1b) to[bend left]  node [midway, left] {\footnotesize $Q_2$} (1);
    \end{scope}
    \end{tikzpicture}
        \caption{}
    \label{fig:rmk:FFL-motif-c}
\end{subfigure}
\begin{subfigure}[b]{0.45\textwidth}
    \centering\figdraftboxwidth
    \begin{tikzpicture}
    \tikzset{decoration={snake,amplitude=.4mm}}
    \begin{scope}[shift={(0,0)}] 
        \node at (0,1.755) {};
        \node (0) at (0.5,0) {$\bullet$}; 
            \node at (0.west) {\footnotesize $s_0$}; 
        \node (1) at (2.5,1) {$\bullet$}; 
            \node at (1.north) {\footnotesize $s_1$}; 
        \node[opacity=0.085] (1b) at (2.5,-1) {$\bullet$}; 
            \node[opacity=0.085] at (1b.south) {\footnotesize $s_{\sigma(1)}$}; 
        \node (f) at (5.125,0) {$\bullet$}; 
            \node at (f.east) {\footnotesize $s_f$}; 
        \draw[posEdge] (0) to[bend left]  node [midway, above] {\footnotesize $P_1$} (1);
        \draw[posEdge, opacity=0.085] (0) to[bend right]  node [midway, below] {\footnotesize $Q_1$} (1b); 
        \draw[posEdge, decorate] (1) to[bend left]  node [midway, above] {\footnotesize $Q_3$} (f);
        \draw[posEdge, decorate, opacity=0.085] (1b) to[bend right]  node [midway, below] {\footnotesize $P_3$} (f); 
        \draw[posEdge, decorate, opacity=0.085] (1) to[bend left]  node [midway, right] {\footnotesize $P_2$} (1b);
        \draw[posEdge, decorate, opacity=0.085] (1b) to[bend left]  node [midway, left] {\footnotesize $Q_2$} (1);
    \end{scope}
    \end{tikzpicture}
        \caption{}
    \label{fig:rmk:FFL-motif-d}
\end{subfigure}

\caption{
    Given an IFFL $P \coloneqq \bigcup_i P_i$ and $Q \coloneqq \bigcup_i Q_i$, there are at least four I/O paths. Wavy edges represent a path of unspecified length.
}
\label{fig:rmk:FFL-motif}
\end{figure}

\section{General kinetics for regulations in a signaling pathway}
\label{sec:application-Tcell-appendix}

\setcounter{equation}{0} 

In \Cref{sec:application-Tcell}, we remarked that \Cref{thm:T-cell} holds not only under mass-action kinetics, but also under more general kinetics assumption, including the Hill-type (\lq\lq saturating\rq\rq) kinetics briefly considered in the appendix of \cite{LeverLimKrugerNguyenEtAl2016}. We give a version of \Cref{thm:T-cell} below for any kinetic rate function satisfying \eqref{eq:activation_fcn-appendix}. \Cref{rmk:application-T-cell-minimal} discusses the core mathematical structure underlying \hyperref[thm:T-cell]{Theorems~\ref{thm:T-cell}} and \ref{thm:T-cell-appendix}.

Concerning the kinetic proofreading model of the T-cell receptor (TCR), we keep it as described in \Cref{sec:application-Tcell-KPR}, including the assumption of mass-action kinetics. 
Recall that at the globally stable equilibrium, we have 
\eq{
    [\cf{C}_i]^* = A_i \underbrace{\left( L_T + \rho + \alpha - \sqrt{(L_T +  \rho + \alpha)^2 - 4 \rho L_T} \right)}_{U(L_T)} ,
}
which is a monotonically increasing function of $L_T$.

For the intracellular regulatory network $G^\circ$, the only change we make concerns the rate functions associated to each regulatory edge. Consider any function $\lambda \colon \rrp^2 \to \rrp$ satisfying 
    \eqn{\label{eq:activation_fcn-appendix}
        \lambda(0, y) = \lambda(x, 0) = 0, \quad 
        \deriv{\lambda}{1} \coloneqq \frac{\partial \lambda}{\partial x} > 0, 
        \quad 
        \deriv{\lambda}{2} \coloneqq  \frac{\partial \lambda}{\partial y} > 0. 
    }
Under mass-action kinetics, $\lambda(x, y) = k xy$ for some constant $k>0$ satisfies the above property. However, we can also assume Michaelis--Menten or Hill-type kinetics: 
    \eq{
        \lambda(x, y) = \frac{k x^m y}{V + x^m}  
    }
with parameters $k, V, m > 0$. 

The reactions to be included with each regulatory edge remain unchanged. For example, an activation edge $\cf{X} \posedge \cf{Y}$ refers to the reaction $\cf{X}_p + \cf{Y} \FRxn{} \cf{X}_p + \cf{Y}_p$, with rate function $\lambda([\cf{X}_p], [\cf{Y}])$, whereas an inhibitory edge $\cf{X} \negedge \cf{Y}$ represents the reaction $\cf{X}_p + \cf{Y}_p \FRxn{} \cf{X}_p + \cf{Y}$, with rate function $\lambda([\cf{X}_p], [\cf{Y}_p])$. 

Note the only difference between activation and inhibition lies in the second variable of $\lambda$: $[\cf{Y}]$ for activation and $[\cf{Y}_p]$ for inhibition. 
The function $\lambda$ is activatory in the Hill-type sense. However, the \lq\lq inhibition\rq\rq\ for $\cf{X} \negedge \cf{Y}$ comes from the reaction\footnote{For models that use inhibitory Hill-type functions, see \Cref{sec:application-GRN} instead.} ($\cf{X}_p$ catalyzing the conversion away from $\cf{Y}_p$). This is reflected in its rate function after substituting in the conservation law: $\lambda([\cf{X}_p], [\cf{Y}]) = \lambda([\cf{X}_p], Y_T - [\cf{Y}_p])$, where $Y_T$ is the total amount of $\cf{Y}$.

Just as it was done in \cite{LeverLimKrugerNguyenEtAl2016}, we allow a mix of non-saturating mass-action kinetics and saturating regulatory functions.

\subsection{Analogue of \Cref{thm:T-cell} under general kinetics assumption}
\label{sec:application-Tcell-appendix-gen-kinetics}

\begin{thm}
\label{thm:T-cell-appendix}
    Consider a model of T-cell activation, consisting of (i) the kinetic proofreading model for a TCR with $N$ sites; (ii) an intracellular signaling pathway with $n$ intracellular signaling proteins, and (iii) regulation of some intracellular proteins by TCR complexes. Assume regulation by the TCR complexes happen in a consistent manner. For $G^\circ$, assume any general kinetics satisfying \eqref{eq:activation_fcn-appendix}. 
    Suppose the amount of cytokines released by the T-cell is proportional to $\cf{X}_{np}$, and the Jacobian matrix of the reduced system is Hurwitz at any positive equilibrium.

    If $[\cf{X}_{np}]$ exhibits stable biphasic response with respect to the total amount of antigen $L_T$, then the signaling network $G^\circ$ must contain either an IFFL, or a combination of an NFBL and a PFBL, where the PFBL is vertex-disjoint from an I/O path, and the NFBL is downstream of $\cf{C}_j$ for some $j$ and upstream of $\cf{X}_n$. 
\end{thm}
\begin{proof}
    The proof follows closely that of \Cref{thm:T-cell}. As before, let $\mm \Gamma \in \rr^{n \times N}$ be the adjacency matrix of the part of $G^\circ$ that governs how $\cf{C}_j$'s regulate $\cf{X}_i$'s, and let $\mm \Xi \in  \rr^{n \times n}$ be the adjacency matrix for how the intracellular species $\cf{X}_i$ regulate each other. The full ODE system is 
        \eq{
        -\dot{[\cf{X}_i]} = \dot{[\cf{X}_{ip}]} &=  
            \gamma_i [\cf{X}_i] - \delta_i [\cf{X}_{ip}] 
            + \underbrace{\sum_{\substack{j=1, \\ [\mm \Gamma]_{i,j} = 1}}^N \tilde{\lambda}_{ji}([\cf{C}_j], [\cf{X}_i] )
                - \sum_{\substack{j=1, \\ [\mm \Gamma]_{i,j} = -1}}^N \tilde{\lambda}_{ji}([\cf{C}_j], [\cf{X}_{ip}] )
                }_{\text{Regulation by $\cf{C}_j$ on $\cf{X}_i$}}
            \\&\qquad + \underbrace{\sum_{\substack{j=1, \\ [\mm \Xi]_{i,j} = 1}}^n \lambda_{ji}([\cf{X}_{jp}], [\cf{X}_i] )
                - \sum_{\substack{j=1, \\ [\mm \Xi]_{i,j} = -1}}^n \lambda_{ji}([\cf{X}_{jp}], [\cf{X}_{ip}] )
                }_{\text{Regulation by $\cf{X}_j$ on $\cf{X}_i$}},
    }
    where $\gamma_i$, $\delta_i > 0$ are constants, and $\tilde{\lambda}_{ji}$, $\lambda_{ji}$ satisfy \eqref{eq:activation_fcn-appendix}. 
    Using the conservation laws $[\cf{X}_i] + [\cf{X}_{ip}] = T_i$ to eliminate the variables of the inactive forms, we reduce the system to one on $(0,T_i)^n \subseteq \rrpp^n$, where $x_i =  [\cf{X}_{ip}]$: 
    \eq{
    \begin{split}
        \dot{x}_i &= \gamma_i (T_i - x_i) - \delta_i x_i
            + \sum_{j \colon [\mm \Gamma]_{i,j} = 1} \tilde{\lambda}_{ji} ([\cf{C}_j], T_i -x_i )
                - \sum_{j \colon [\mm \Gamma]_{i,j} = -1} \tilde{\lambda}_{ji} ([\cf{C}_j] ,   x_i )
           \\&\qquad + \sum_{j \colon [\mm \Xi]_{i,j} = 1} \lambda_{ji}( x_j , T_i-x_i)
                - \sum_{j \colon [\mm \Xi]_{i,j} = -1} \lambda_{ji} (x_j ,x_i) .  
    \end{split}
    }

    In terms of the partial derivatives, for $r \in \{1,2\}$, we write $\deriv{\lambda_{ji}}{r}$ for the derivative with respect to the $r$th variable, e.g., $\deriv{\lambda_{ji}}{1} = \frac{\partial \lambda_{ji}}{\partial [\cf{C}_j]}$. A similar notation is used for $\tilde{\lambda}_{ji}$. 
    Then the augmented Jacobian matrix of this reduced system has as its diagonal entries $[\mm A]_{0,0} = -1$, and for $1 \leq i \leq n$, 
    \eq{\begin{split}
        [\mm A]_{i,i} &= -\gamma_i - \delta_i 
            - \sum_{j \colon [\mm \Gamma]_{i,j} = 1}  \deriv{\tilde{\lambda}_{ji}}{2}([\cf{C}_j], T_i - x_i)
            - \sum_{j \colon [\mm \Gamma]_{i,j} = -1} \deriv{\tilde{\lambda}_{ji}}{2}([\cf{C}_j], x_i)
            \\&\qquad  
            - \sum_{j \colon [\mm \Xi]_{i,j} = 1} \deriv{\lambda_{ji}}{2} (x_j, T_i - x_i)
            - \sum_{j \colon [\mm \Xi]_{i,j} = -1} \deriv{\lambda_{ji}}{2} (x_j, x_i) < 0. 
    \end{split}}
    For $i\neq j \in [n]$, we have $[\mm A]_{0,j} = 0$ and 
    \eq{
        [\mm A]_{i,j} &= \left\{
            \begin{array}{cl}
                \deriv{\lambda_{ji}}{1}(x_j, T_i - x_i)  & \text{ if } [\mm\Xi]_{i,j} = 1, \\
                -\deriv{\lambda_{ji}}{1}(x_j, x_i) & \text{ if } [\mm\Xi]_{i,j} = -1, \\
               0  & \text{ if } [\mm\Xi]_{i,j} = 0 .
            \end{array}
            \right. 
    }
    Finally, using the equilibrium relations $[\cf{C}_j] = A_j U(L_T)$ from \eqref{eq:Tcell-Ci-ss}, we have 
    \eq{
        [\mm A]_{i,0} &= \left\{ 
        \begin{array}{cl}
            0 & \text{ if } [\mm \Gamma]_{i,\cdot} = \vv 0, \\ 
             U'(L_T) \sum_{j\colon [\mm \Gamma]_{i,j} \neq 0}  A_j \deriv{\tilde{\lambda}_{ji}}{1}  & \text{ if } [\mm \Gamma]_{i,\cdot} \neq \vv 0 ,  \, [\mm \Gamma]_{i,\cdot} \geq \vv 0,   \\
             -U'(L_T) \sum_{j\colon [\mm \Gamma]_{i,j} \neq 0}  A_j \deriv{\tilde{\lambda}_{ji}}{1}   & \text{ if }  [\mm \Gamma]_{i,\cdot} \neq \vv 0 , \,  [\mm \Gamma]_{i,\cdot} \leq \vv 0  .
        \end{array}
        \right. 
    }
    Thus, the proof goes through just as in \Cref{cor:main}. 
\end{proof}

\begin{rmk}
\label{rmk:application-T-cell-minimal}
    There are only two ingredients that are absolutely needed for the proofs of \hyperref[thm:T-cell]{Theorems~\ref{thm:T-cell}} and \ref{thm:T-cell-appendix}.  First, the mandatory contribution to the augmented Jacobian matrix $\mm A$ by a regulatory edge $\cf{X}_j \neutedge \cf{X}_i$ is only in the $(i,j)$ entry. Moreover, the sign of the entry matches that of the regulatory edge. It is possible that this edge also contributes a negative term to the diagonal of $\mm A$. Note that for some species, the spontaneous conversion between forms may not be necessary for the relevant diagonal entry to be negative.  
\end{rmk}

Finally, we simply note that other mechanisms are also permitted. For example,  $\cf{X} \posedge \cf{Y}$ can represent the two reactions  
    \eq{
        \cf{X}_p + \cf{Y} \FRxn{} \cf{X}_p + \cf{Y}_p 
        \quad \text{and} \quad 
        \cf{X} + \cf{Y}_p \FRxn{} \cf{X} + \cf{Y}, 
    }
where the rate function for each reaction satisfies \eqref{eq:activation_fcn-appendix}. Similarly $\cf{X} \negedge \cf{Y}$ can refer to the two reactions  
    \eq{
        \cf{X}_p + \cf{Y}_p \FRxn{} \cf{X}_p + \cf{Y} 
        \quad \text{and} \quad 
        \cf{X} + \cf{Y} \FRxn{} \cf{X} + \cf{Y}_p .
    }

\section{Details about \Cref{ex:appl-Tcell-noIFFL}}
\label{app:ex:appl-Tcell-noIFFL}

The parameters for the model of T-cell activation are as follow. For the kinetic proofreading model, the rate constants are all set to $1$, and the total amount of TCRs is $\rho = 5$, while $L_T \in (0.01,80)$. 
The total amount of each intracellular protein is set to $1$, and the rate constants for the spontaneous conversions are 
    $\gamma_1 = 2$, $  \delta_1 = 1$, 
    $\gamma_2 = 0.01$, $  \delta_2 = 0.05$, 
    $\gamma_3 = 0.02$, $  \delta_3 = 0.02$, 
    $\gamma_4 = 0.001$, $ \delta_4 = 1$.  
Finally, for the regulations, the rate constants are listed in \Cref{tab:ex-noIFFL-scheme}.  
    
    \begin{table}[h!]
    \centering 
    \begin{tabular}{cc}
    Regulatory edge & Rate constant \\ 
    \hline
    $\cf{C}_1 \negedge \cf{X}_1$ & $10$ \\
    $\cf{X}_1 \posedge \cf{X}_2$ & $4$ \\
    $\cf{X}_2 \negedge \cf{X}_1$& $25$ \\
    $\cf{X}_2 \negedge \cf{X}_3$& $5$ \\
    $\cf{X}_3 \negedge \cf{X}_2$& $420$ \\
    $\cf{X}_2 \posedge \cf{X}_4$& $1.25$ \\
    $\cf{X}_4 \posedge \cf{X}_2$&  $5$\\
    \hline
    \end{tabular}
    \caption{Rate constants for the regulatory network in \Cref{fig:ex-noIFFL-scheme}.}
    \label{tab:ex-noIFFL-scheme}
    \end{table}

\Cref{fig:ex-noIFFL-evals} contains a plot of the eigenvalues $\lambda_i$ of the reduced system, indicating that the system is asymptotically stable for each $L_T$ used. The top contains plots of the real parts; the top-right plot zooms in to close to the zero. The bottom contains plots of the eigenvalues in the complex plane; the bottom-right zooms in close to the origin. The code for the simulation is available at  \cite[Example~5.3 simulation]{github}.

\begin{figure}[h!]
\centering 
    \includegraphics[scale=0.75]{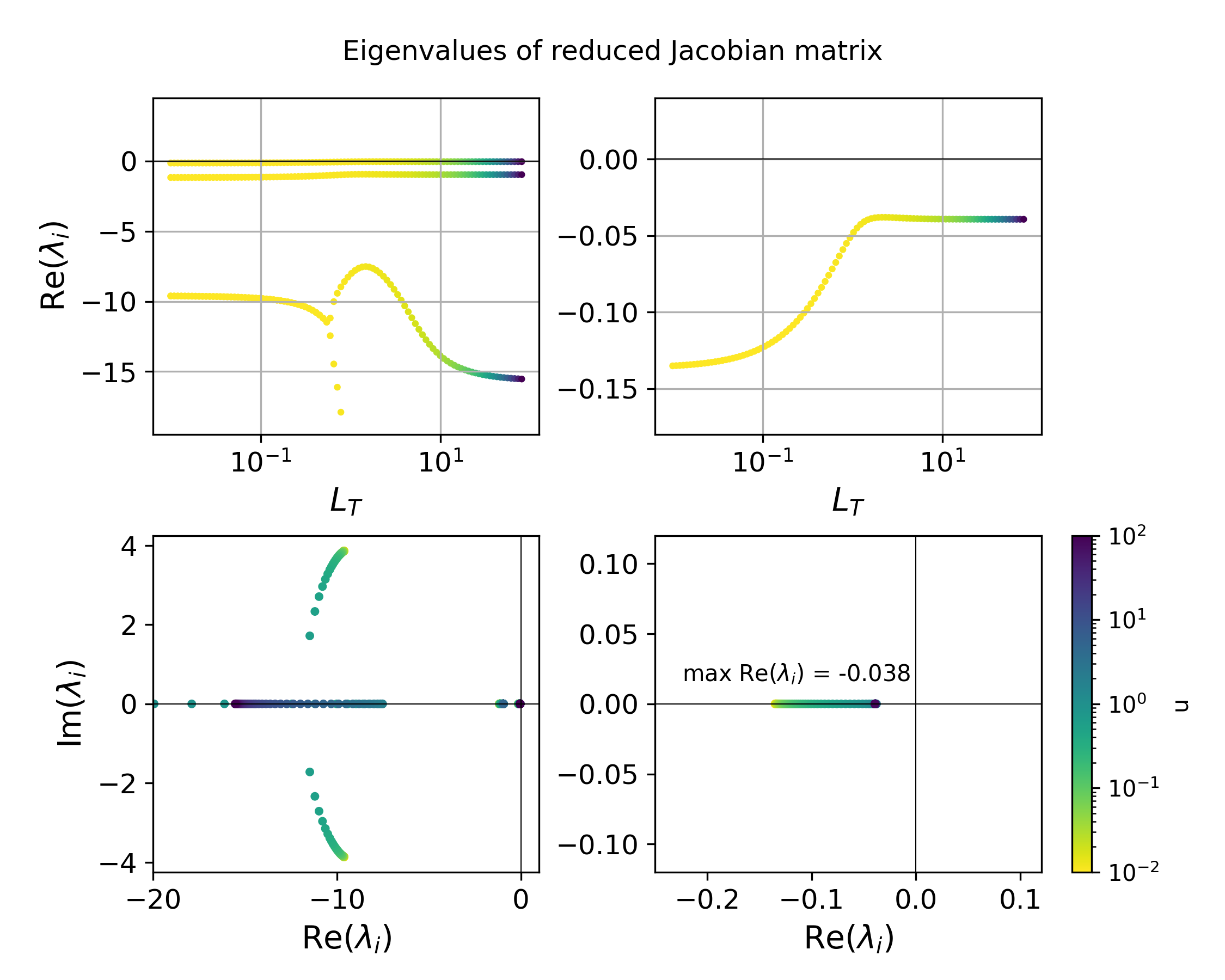} 
\caption{
    The eigenvalues of the reduced system for $\cf{X}_i$, $\cf{X}_{ip}$, $i\in[4]$ from \Cref{ex:appl-Tcell-noIFFL}. ({\sc Top}) The real parts of the eigenvalues. ({\sc Bottom}) Plots of eigenvalues in the complex plane. ({\sc Right}) Zoomed in views. The system is never singular, where the largest real part of the eigenvalues is less than $-0.03$. 
}
\label{fig:ex-noIFFL-evals}
\end{figure}

\clearpage 
\section{Details about \Cref{ex:GRN-4node}}
\label{app:ex:GRN-4node}

The parameters for \Cref{ex:GRN-4node} are as follow:  $\gamma_2 = \gamma_3 = 0.02$, $\delta_1 = 1$, $\delta_2 = \delta_3 = 0.1$, 
    $(k_1, V_1,m_1) = (1,1,4)$, 
    $(k_2, V_2,m_2) = (1,10,4)$, 
    $(k_3, V_3,m_3) = (0.5,3,2)$,  and 
    $(k_4, V_4,m_4) = (0.5,1,5)$. 
The eigenvalues, shown in \Cref{fig:ex:GRN-4node}, are always negative. In other words, the system is stable. The code for the simulation is available at  \cite[Example~5.4 simulation]{github}.

\begin{figure}[h!]
\centering 
    \includegraphics[scale=0.8]{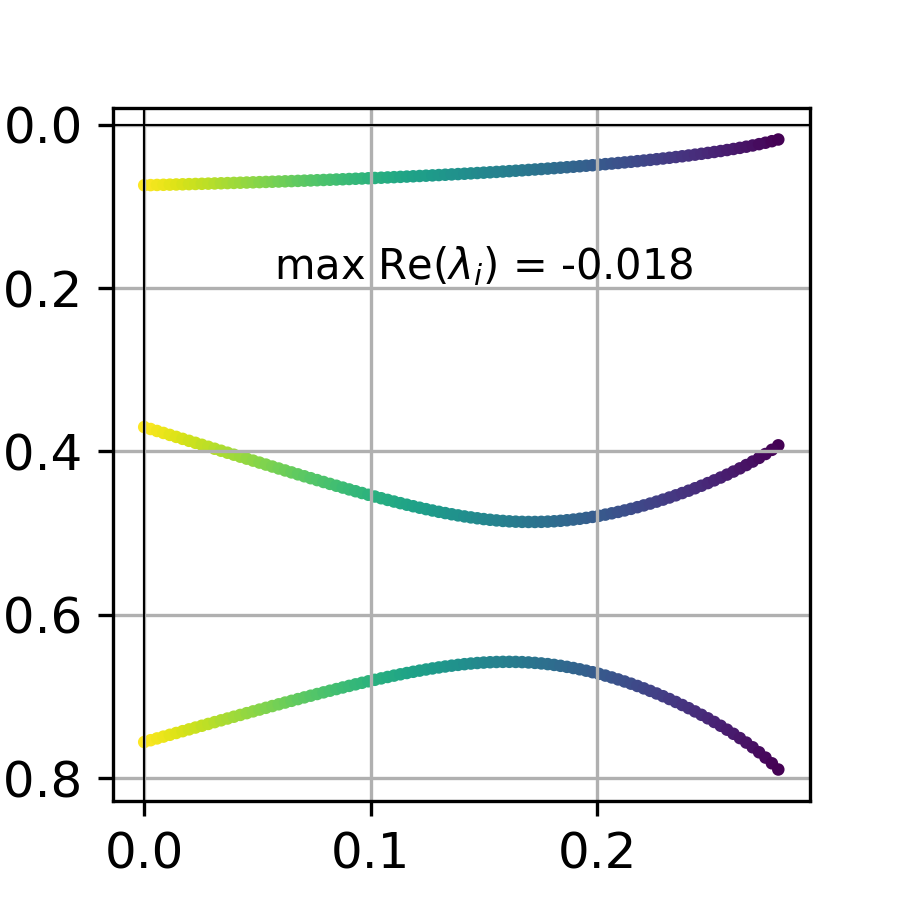} 
\caption{
    The eigenvalues of the system from \Cref{ex:appl-Tcell-noIFFL}. 
}
\label{fig:ex:GRN-4node}
\end{figure}

\section{Jacobian matrix of the CRM in \Cref{sec:application-CRM}}
\label{sec:application-CRM-appendix}

\setcounter{equation}{0} 

We prove here that the Jacobian matrix of \eqref{eq:application-CRM-ODE-R}--\eqref{eq:application-CRM-ODE-C} at a positive equilibrium $(\vv x^*, \vv y^*)$ is 
\eq{
    \mm J^* = \begin{pNiceArray}{r|cr}
            - \mm D & \tilde{\mm \Gamma} - \mm I_1 \\ 
            \hline 
            \mm I_2 & \mm 0
        \end{pNiceArray},
}
where 
\eq{
    \mm D &= \diag\left(\vv \mu + \left[ \tilde{\mm A} \odot \tilde{\mm F}'(\vv x^*) \right] \vv y^* \right) \in \rr^{M \times M}, \\ 
    \mm I_1 &= \tilde{\mm A} \odot \tilde{\mm F}(\vv x^*) \in \rr^{M\times N}, \\ 
    \mm I_2 &= \diag(\vv y^*) \left[ \tilde{\mm B} \odot \tilde{\mm F}'(\vv x^*) \right]^\top \in \rr^{N\times M}. 
}

The upper-left $M\times M$ block comes from the derivatives of $\dot x_i$ with respect to $x_{i'}$; similarly the lower-right $N\times N$ block comes from the equations for $y_j$'s. 
The off-diagonal entries in both the top-left and bottom-right corners are zero, due to the bipartite structure of the community network. The diagonal entries of the top-left corner are given by $- \left[\mu_i + \sum_j \alpha_{ij} f_{ij}'(x_i) y_j  \right]$ from \eqref{eq:application-CRM-ODE-R}. The diagonal of the bottom-right is given by $\left[ \tilde{\mm B}\odot \tilde{\mm F}(\vv x)\right]^\top \ones  - \vv \delta $, which is zero at $\vv x^*$; see \eqref{eq:application-CRM-ss}.

The off-diagonal blocks of $\mm J^*$ refer to the interactions between resources and consumers. For the bottom-left block, we consider the dependence of $\dot y_j$ on $x_i$, where $[\mm J^*]_{M+j,i} =  y_j \beta_{ij} f'_{ij}$. For the upper-right block, look to how $\dot x_i$ depends on $y_j$: 
\eq{
    [\mm J^*]_{i,M+j} = \gamma_{ij} - \alpha_{ij} f_{ij}
    = \left\{ 
    \begin{array}{cl}
         -\alpha_{ij} f_{ij} & \text{ if } \alpha_{ij} > 0 ,  \\
         \gamma_{ij}  & \text{ if } \gamma_{ij} > 0, \\
         0 & \text{ if } \alpha_{ij} = \gamma_{ij} = 0.
    \end{array}\right. 
}
Note that by assumption, $\gamma_{ij} > 0$ requires $\alpha_{ij} = 0$. Either way, the upper-right block is given by $\tilde{\mm \Gamma} - \tilde{\mm A} \odot \tilde{\mm F}$.

\section{Reducing to a subgraph when calculating sensitivity and connection to works by Golubitsky et al.}
\label{sec:app-trimgraph}

\setcounter{equation}{0}

The purpose of this section is two-fold. First, in \Cref{rmk:trimgraph} we noted that not necessarily every edge in the J-graph contributes to $d_u{x^*_n}$. Below we discuss two manners in which this occurs (\Cref{sec:trimming,sec:edges-not-in-cover}). 

Second, an observation made in \Cref{sec:edges-not-in-cover} naturally leads to the factoring of the non-principal minor (\Cref{sec:factoring}). Since this closely echoes the framework developed in \cite{WangHuangAntoneliGolubitsky2021}, we discuss in \Cref{sec:lit-golubitsky} the connection between their approach and ours. In particular, we provide a graph-theoretic interpretation for the objects defined in \cite{WangHuangAntoneliGolubitsky2021} and related works.

\subsection{Removing nodes that are not downstream of input and upstream of output}
\label{sec:trimming}

The claim of this section is that the only nodes that contribute to the sensitivity $d_u{x^*_n}$ via \Cref{lem:Cramer} are those that are downstream of the input $u$ and upstream of the output $x_n$. Terms associated to other nodes may appear in the non-principal minor $\mm A[\hat 1, \widehat{n+1}]$ and the Jacobian $\det(\mm J)$ simultaneously, but these terms are canceled out in the rational expression for $d_u{x^*_n}$. 

The proof follows \cite[Theorem~1]{AscensaoDattaHanciogluSontagEtAl2016}; this process is analogous to identifying the observable and reachable nodes in the \lq\lq Kalman decomposition\rq\rq\ reduction in linear control theory. 

First, partition the variables in terms of their reachability from $u$ and to $x_n$. Let $X$ be the (ordered) set of nodes downstream of $u$ and upstream of $x_n$; without loss of generality, we may assume $x_n$ belongs to $X$. Let $Y$ be the set of nodes upstream of $x_n$ but not downstream of $u$; $W$ for those downstream of $u$ but not upstream of $x_n$, and $V$ for those neither downstream of $u$ nor upstream of $x_n$. Rearranging the variables, \eqref{eq:syst_main} can be written as 
\eq{
    \dot Y &= G(Y) \\
    \dot V &= H(Y, V) \\ 
    \dot X &= F(Y, X; u) \\
    \dot W &= K(Y,V,X,W; u), 
}  
with augmented Jacobian matrix 
\eq{
    \mm A = \begin{pNiceArray}{c|cccc}[first-row,first-col,
      code-for-first-row=\scriptstyle\color{gray},
      code-for-first-col=\scriptstyle\color{gray}]
      & u & Y & V & X & W \\
u   & -1 & \vv 0^\top & \vv 0^\top & \vv 0^\top & \vv 0^\top  \\ \hline
Y   & \vv 0 & \mm J_{11} & \mm{0} & \mm{0} & \mm{0}   \\
V   & \vv 0 & \mm J_{21} & \mm J_{22}  & \mm{0} & \mm{0} \\ 
X & \partial_u F & \mm J_{31} & \mm{0} & \mm J_{33} & \mm{0} \\
W & \partial_u K & \mm J_{41} & \mm J_{42} & \mm J_{43} & \mm J_{44}
\end{pNiceArray} . 
}
Note that $\det (\mm J) = \det(\mm J_{11}) \det(\mm J_{22}) \det(\mm J_{33}) \det(\mm J_{44})$.

Let $X_{j'} = x_n$, where $j \coloneqq |Y| + |V| + j'$, so column $j+1$ of $\mm A$ picks out the output variable. Denote by $\mm M \coloneqq \mm J_{33}(j' \leftarrow \deriv{F}{u})$ the matrix where we replaced column ${j'}$ in $\mm J_{33}$ by the vector $\deriv{F}{u}$. Then the relevant non-principal minor from \Cref{lem:Cramer} can be factored as $\mm A[\hat 1, \widehat{j+1}] = (-1)^{j-1} \det(\mm J_{11}) \det(\mm J_{22}) \det(\mm M)\det(\mm J_{44})$, leading to a simplification of \eqref{eq:Cramer}: 
\eqn{\label{eq:Cramer-2-after_pruning}
    d_u{x^*_{n}} 
        = (-1)^{j} \frac{\mm A^*[\hat 1, \widehat{j+1}]}{\det(\mm J^*) }
        = \frac{-\det(\mm M^*) }{\det(\mm J_{33}^*)}. 
}
In other words, for the purpose of computing the sensitivity $d_u{x_n^*}$, we do not need the entire J-graph $G$, but only the induced subgraph $G[X\cup\{u\}]$, where $X$ is the set of nodes that are both downstream of $u$ and upstream of $x_n$. The subgraph $G[X\cup\{u\}]$ is called the \lq\lq core network\rq\rq\ in \cite{WangHuangAntoneliGolubitsky2021}.

The \lq\lq Kalman decomposition\rq\rq\ reduction can be carried out for the symbolic matrix $\mbc A$ instead of $\mm A$.

\subsection{Removing edges that do not appear in any disjoint I/O cover}
\label{sec:edges-not-in-cover}

What we describe next concerns the determinant of signed matrices (and its minors). More accurately, it concerns only the non-zero pattern of the matrix and the results apply equally well to the permanent as to the determinant. 

In this section, suppose we have an $m \times m$ symbolic matrix $\mbc A$ where $m > 1$, with strictly non-zero diagonal. Any non-zero entry is an indeterminate, i.e., $[\mbc A]_{ij} = a_{ij}$,  where $a_{ij}$ is non-zero.  Assume that every node is downstream of $1$, which we take as input, and upstream of $m$, which we treat to be the output. Since we are concerned with $\mbc A[\hat 1, \hat m]$, we may assume without loss of generality that $[\mbc A]_{1,j} = 0$ for $2 \leq j \leq m$, and that $\det(\mbc A) \not\equiv 0$. Let $G$ be the J-graph of $\mbc A$, and let $\mc E$ be the union of disjoint I/O covers. 

We described in \Cref{rmk:cycles-permutation-cover} how every non-zero term in the expansion of the non-principal minor $\mbc A[\hat 1, \hat m]$ corresponds to a disjoint I/O cover of the J-graph $G$, and for any $a_{ij}$ that appears in $\mbc A[\hat 1, \hat m]$, the corresponding edge $i \neutedge j$ must belong to an I/O cover. Therefore, the collection of non-zero terms in the expansion of the minor $\mbc A[\hat 1, \hat m]$ are in a one-to-one correspondence with the set of disjoint I/O covers of $\mc E$, which is in a one-to-one correspondence with the set of disjoint I/O covers of $G$. 
Note however, an edge in $G\setminus \mc E$ may be relevant when calculating $\det(\mbc A)$.

\begin{figure}[h]
\centering\setlength\tabcolsep{0.5em}
\medskip 
\begin{tabular}[b]{cc}
	\raisebox{6pt}{
		\subcaptionbox{\label{fig:factoring-Jgraph}}{
			\includegraphics{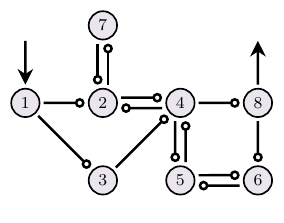}
		}
	}
	&
	\raisebox{0.8cm}{
		\multirow[b]{2}{*}[3.cm]{\subcaptionbox{\label{fig:factoring-covers}}{
			\includegraphics{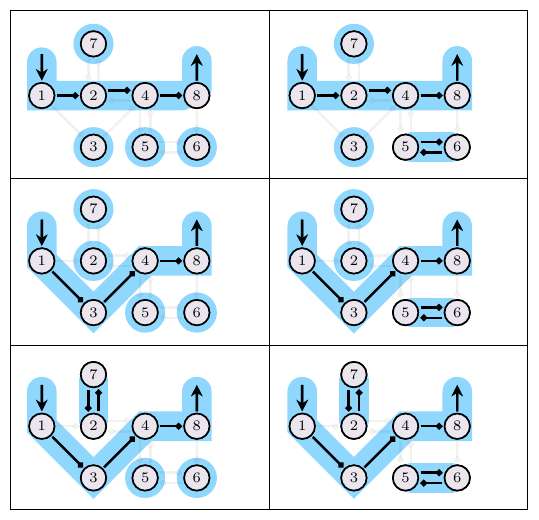}
		}}
	}
	\\[20pt]
	{
		\subcaptionbox{\label{fig:factoring-union}}{
			\includegraphics{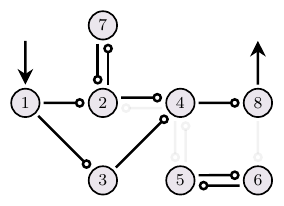}
		}
	}
  \end{tabular}
  \vspace{0.15cm}
\caption{(a) The J-graph from \Cref{ex:factoring} with self-loops omitted for simplicity. (b) The six disjoint I/O covers. (c) The union of the disjoint I/O covers $\mc E$, which has two connected components.  
}
\label{fig:factoring}
\end{figure}

\begin{ex}
\label{ex:factoring}
    Consider the minor $\mbc A [\hat 1, \hat 8]$ of the matrix 
    \eq{ 
        \mbc A = \begin{pNiceArray}{ccccccccc}
             a_{11} & 0 & 0 & 0 & 0 & 0 & 0 & 0 \\
              a_{21} & a_{22} & 0 &  a_{24} & 0 & 0 & a_{27}& 0 \\ 
             a_{31} & 0 & a_{33} & 0 & 0 & 0 & 0& 0 \\ 
             0 & a_{42} & a_{43} & a_{44} & a_{45} & 0 & 0& 0 \\ 
             0 & 0 & 0 & a_{54} &  a_{55} & a_{56} & 0& 0 \\ 
             0 & 0 & 0 & 0 & a_{65} & a_{66} & 0& a_{68} \\
             0 &  a_{72} & 0 & 0 & 0 & 0  & a_{77} & 0\\
             0 & 0 & 0 & a_{84} & 0 & 0 & 0 & a_{88} 
        \end{pNiceArray}. 
    }
    Its J-graph is shown in \Cref{fig:factoring-Jgraph} with self-loops omitted for simplicity. 
    Since the J-graph has six disjoint I/O covers (\Cref{fig:factoring-covers}), the non-principal minor has six terms: 
    \eq{
        \mbc A[\hat 1, \hat 8] 
            = ( a_{56} a_{65} - a_{55} a_{66} ) 
                \left[ \vphantom{\sum}\right. 
                a_{84}
                ( a_{27}a_{31}a_{43}a_{72} 
                    -a_{22}a_{31}a_{43}a_{77} 
                    -a_{21}a_{33}a_{42}a_{77})   
                \left. \vphantom{\sum}\right]
        .
    }  
    Notice that $a_{54}$ does not appear in $\mbc A[\hat 1, \hat 8]$ since $4 \neutedge 5$ is not in any of the disjoint I/O covers. In particular, $4\neutedge 5$ is not in the union of all disjoint I/O covers $\mc E$ (\Cref{fig:factoring-union}). A similar observation can be made for $a_{45}, a_{68}, a_{24}$. 
    
    One might notice that the irreducible factor $(a_{56}a_{65} - a_{55}a_{66})$ in $\mbc A[\hat 1, \hat 8]$ picks up precisely the strongly connected component $5 \neutFBL 6$ in $\mc E$. Meanwhile, the remaining factor, which is associated to the other connected component of $\mc E$, is reducible. In the next section, we discuss next how factors of the non-principal minor is related to $\mc E$. 
\end{ex}

\subsection{Factoring $\mbc A[\hat 1, \hat m]$ and the union of disjoint I/O covers}
\label{sec:factoring}

In \Cref{ex:factoring}, we saw that $\mbc A [\hat 1, \hat 8]$ has three irreducible factors. One of these factors is associated to the strongly connected component $5 \neutFBL 6$. The other two irreducible factors are associated to the connected component with the input and output nodes. We argue that this holds in general. 

A parallel effort at relating factors of $\mbc A[\hat 1, \hat m]$ and the graph has been carried out in \cite{WangHuangAntoneliGolubitsky2021} and related works. Their approach is more algebraic, while ours is graph-centric. The connection between what is discussed in this section and \cite{WangHuangAntoneliGolubitsky2021} is explored in more depth in \Cref{sec:lit-golubitsky}.

We first observe the graph-theoretic structure of the union of disjoint I/O covers $\mc E$.  
The input node $1$ and the output node $m$ must lie within the same connected component, which we call the \df{I/O component} $\mc I_\mrm{I/O}$. Although $\mc I_\mrm{I/O}$ is not strongly connected, $\mc I_\mrm{I/O} \cup \{ m \neutedge 1\}$ is strongly connected. Any other connected components in $\mc E$ must be strongly connected. 

We start with the strongly connected components, each of which corresponds to an irreducible factor of $\mbc A[\hat 1, \hat m]$. Moreover, such a factor is a principal minor of $\mbc A$.

\begin{lem}
\label{lem:appendix-factor-nonIOcomponents}
    Suppose $\mc I$ is a strongly connected component of $\mc E$. Let $\gamma$ be the set of nodes of $\mc I$. Then $\mc I$ corresponds to $\mbc A[\gamma]$, which is an irreducible factor of $\mbc A[\hat 1, \hat m]$.
\end{lem}
\begin{proof}
    First, observe that in every disjoint I/O cover $H$ of the J-graph, the connected components are supported either on $\gamma$ or on $\gamma^c$. Those supported on $\gamma$ must be cycles in $H$. In other words, $\mc I$ is the union of vertex-disjoint cycles covering $\gamma$, and $\mc I$ corresponds to $\mbc A[\gamma]$.

    Next, we argue that $\mbc A[\gamma]$ divides $\mbc A[\hat 1, \hat m]$. 
    Let $S_m^*$ be the subset of permutations for which the corresponding term $T_\sigma \neq 0$ in $\mbc A[\hat 1, \hat m]$. Similarly, let $S_\gamma^*$ be the set of permutations in $S_\gamma$ for which the corresponding term $T_\tau \neq 0$ in $\mbc A[\gamma]$. Every $\sigma \in S_m^*$ can be factored as $\sigma = \tau \rho$, where $\tau \in S_\gamma$ and $\rho \in S_{\gamma^c}$. Algebraically, we have $T_\sigma = T_\tau T_\rho$ for $\sigma \in S_m^*$, so $\tau \in S_\gamma^*$. 
    Because every component of $H$ is supported on either $\gamma$ or $\gamma^c$, we have independence between $\tau$, $\rho$, in the sense that if $\tau_1 \rho_1$, $\tau_2 \rho_2$ are two permutations in $S_m^*$, then $\tau_i \rho_j$ for $i,j \in \{1,2\}$ produce four elements of $S_m^*$. Therefore, $S_m^* = S_\gamma^* \times S_{\gamma^c}^*$ for some subset $S_{\gamma^c}^*$, and $\mbc A[\gamma]$ divides $\mbc A[\hat 1, \hat m]$. 

    Finally, that $\mbc A[\gamma]$ is irreducible as a polynomial of $a_{ij}$ follows from \cite[Corollary~7.2]{Schneider1977}, because $\mbc A(\gamma)$ has strictly non-zero diagonal, and $\mbc A(\gamma)$ is irreducible as a matrix since $\mc I$ is strongly connected. 
\end{proof}

We remark that Corollary~7.2 of \cite{Schneider1977} truly requires $\mbc A(\gamma)$ to have strictly non-zero diagonal. A simple counterexample is 
    \[ \begin{pmatrix}
        0 & b \\ 
        c & 0
    \end{pmatrix}. 
    \]
Therefore, we cannot directly apply \cite[Corollary~7.2]{Schneider1977} to the non-principal minors of $\mbc A$. 

The full argument behind \cite[Corollary~7.2]{Schneider1977} for an arbitrary symbolic matrix $\mbc B$ is as follow. First, $\det(\mbc B)$ is irreducible as a polynomial if and only if $\mbc B$ is a \emph{fully indecomposable}%
    \footnote{A matrix $\mbc B$ is said to be \emph{fully indecomposable} if there do not exist permutation matrices $\mm P, \mm Q$ such that $\mm P \mbc B \mm Q$ is upper block-triangular.}
matrix~\cite[Theorem~9.2.4]{BrualdiRyser1991}. Full indecomposability is equivalent to \emph{irreducibility}%
    \footnote{A matrix $\mbc B$ is said to be \emph{irreducible} if there does not exist a permutation matrix $\mm P$ such that $\mm P \mbc B \mm P^\top$ is upper block-triangular. Note that full indecomposability implies  irreducibility.} 
only when the matrix has strictly non-zero diagonal~\cite[Theorem~4.2.3]{BrualdiRyser1991}. Finally, $\mbc B$ is irreducible if and only if its directed graph is strongly connected~\cite[Theorem~3.2.1]{BrualdiRyser1991}.

The discussion concerning factors associated to $\mc I_\mrm{I/O}$ can be quite complicated. We use an example to illustrate the key idea, and state the result in \Cref{lem:app-factoring-IOcomponents}.

\begin{figure}[h]
\centering
    \includegraphics[]{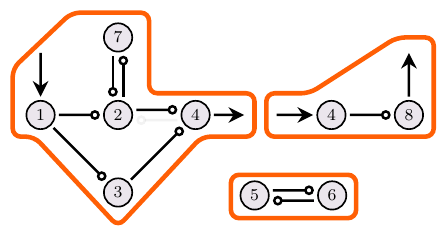}
\caption{Decomposing the graph $\mc E$ from \Cref{fig:factoring-union} into subgraphs corresponding to the irreducible factors of $\mbc A[\hat 1, \hat 8]$. 
}
\label{fig:factoring2}
\end{figure}

\begin{ex}
\label{ex:factoring2}
    Consider \Cref{ex:factoring} again. In the I/O component $\mc I_\mrm{I/O}$ in \Cref{fig:factoring-union}, notice that every I/O path must pass through the nodes $1$, $4$, and $8$. We cleave $\mc I_\mrm{I/O}$ at these nodes, thus splitting it into two subgraphs: one that lies between $1$ and $4$, and another between $4$ and $8$. 
    
    The latter can be seen as its own J-graph with input $4$ and output $8$, which corresponds to the non-principal minor $\mbc B[\hat 1, \hat 2]$, where 
    \eq{
        \mbc B = \mbc A(\{4,8\}) = 
        \begin{pNiceArray}{cc}[first-row,first-col,
            code-for-first-row=\scriptstyle\color{gray},
            code-for-first-col=\scriptstyle\color{gray}]
            & 4 & 8 \\
        4&    a_{44} & 0 \\
        8&    a_{84} & a_{88} 
        \end{pNiceArray}. 
    }
    This returns precisely the irreducible factor $a_{84}$ of $\mbc A[\hat 1, \hat 8]$. 
    The subgraph between $1$ and $4$ can also be seen as its own J-graph with input $1$ and output $4$. This corresponds to the non-principal minor $\mbc C[\hat 1, \hat 4]$, where  
    \eq{
        \mbc C = \mbc A(\{1,2,3,4,7\}) = \begin{pNiceArray}{ccccc}[first-row,first-col,
            code-for-first-row=\scriptstyle\color{gray},
            code-for-first-col=\scriptstyle\color{gray}]
            & 1 & 2 & 3 & 4 & 7 \\
        1&     a_{11} & 0 & 0 & 0 & 0  \\
        2&      a_{21} & a_{22} & 0 &  a_{24} & a_{27} \\ 
        3&     a_{31} & 0 & a_{33} & 0 & 0  \\ 
        4&     0 & a_{42} & a_{43} & a_{44} &  0  \\ 
        7&     0 &  a_{72} & 0 & 0  & a_{77}
        \end{pNiceArray}, 
    }
    This returns the other irreducible factor of $\mbc A[\hat 1, \hat 8]$. Together with $5 \neutFBL 6$, these correspond to the three irreducible factors of $\mbc A[\hat 1, \hat 8]$ (\Cref{fig:factoring2}).
\end{ex}

The basic idea is to cleave $\mc I_\mrm{I/O}$ at nodes through which every I/O path passes. This decomposes $\mc I_\mrm{I/O}$ into an ordered sequence of subgraphs, each of which can be treated as the J-graph of an input--output system. These subgraphs correspond to the irreducible factors of $\mbc A[\hat 1, \hat m]$.

\begin{lem}
\label{lem:app-factoring-IOcomponents}
    Let $\mc I_\mrm{I/O}$ be the I/O component of $\mc E$. Let $V_*$ be the set of nodes through which every I/O path passes through. 
    \begin{enumerate}[label={(\alph*)}]
        \item\label{lem:app-factoring-IOcomponents-a} 
            There is a total ordering on $V_*$ using the reachability relation, i.e., $i \lesssim j$ if and only if there is a (possibly empty) path from $i$ to $j$ in $\mc I_\mrm{I/O}$. 
    \end{enumerate} 
    Let $V_* = \{ i_1 \coloneqq 1 \lesssim i_2 \lesssim \cdots \lesssim i_\ell \coloneqq m\}$, where $\ell \geq 2$. 
    \begin{enumerate}[label={(\alph*)}]
    \setcounter{enumi}{1}
        \item\label{lem:app-factoring-IOcomponents-b} 
            The relation $\lesssim$ defines a pre-order on all the nodes of $\mc I_\mrm{I/O}$. In particular, any node $i \not\in V_*$ lies strictly between $i_k$ and $i_{k+1}$ for a unique $1 \leq k < \ell$. Similarly, any edge in $\mc I_\mrm{I/O}$ lies strictly between $i_k$ and $i_{k+1}$ for some $k$. 

        \item\label{lem:app-factoring-IOcomponents-c}  
            For each $1 \leq k < \ell$, let $\mc I_k$ be the subgraph induced by $\{i \colon i_k \lesssim i \text{ and } i \lesssim i_{k+1}\}$. Then each $\mc I_k$ corresponds to a factor of $\mbc A[\hat 1, \hat m]$; moreover, the factor is a non-principal minor of $\mbc A$. 
    \end{enumerate} 
\end{lem}

Since the proof mirrors that of \cite[Section~1.7.4]{WangHuangAntoneliGolubitsky2021}, we omit it for simplicity.

\subsection{Connection to literature}
\label{sec:lit-golubitsky}

Since versions of \Cref{lem:app-factoring-IOcomponents,lem:appendix-factor-nonIOcomponents} appear in \cite{WangHuangAntoneliGolubitsky2021} and related works, we take this opportunity to discuss their works in relation to \Cref{sec:factoring}. The main difference lies in how we approached our result. Throughout this work, we repeatedly employ the Leibniz formula for determinant expansion, and factoring permutations according to its cycle decomposition. In contrast, \cite{WangHuangAntoneliGolubitsky2021} appeals to Frobenius-K\"onig theory and applies permutation matrices $\mm P, \mm Q$ such that $\mm P \mbc A(\hat 1, \hat m)\mm Q$ is block upper-triangular; see the discussion on full indecomposability vs. irreducibility in the footnote after \Cref{lem:appendix-factor-nonIOcomponents}. 

We note that while \cite{GolubitskyStewartAntoneliHuangWang2020, WangHuangAntoneliGolubitsky2021, deOliveiraMadeiraAntoneli2024} were concerned with quasi-adaptation, parallel work was done for perfect adaptation~\cite{AraujoLiotta2018}. 

As is done in \cite{WangHuangAntoneliGolubitsky2021}, we assume the control parameter only influences the variable $x_1$. In this case, there is a simple relation between treating $u$ vs. $x_1$ as input; therefore, we use the version of J-graph without an explicit node for $u$. On the general case where $u$ may affect multiple variables, we briefly comment on it in \Cref{rmk:multiple-controls}. 

Nodes that appear in any I/O path are called \emph{simple nodes} in \cite{WangHuangAntoneliGolubitsky2021}, and those that appear in \emph{every} I/O path are called \emph{super-simple}. In \Cref{ex:factoring}, the nodes $2,3$ are simple but not super-simple, while $V_* = \{1,4,8\}$ is the set of super-simple nodes. The nodes $5,6,7$ are \emph{appendage} nodes. In particular, $\mc I_\mrm{I/O}$ may contain both simple and appendage nodes; however, the strongly connected components of $\mc E$ contain only appendage nodes. 

The \lq\lq appendage\rq\rq\ language also applies to subgraphs of $\mc E$\footnote{Technically subgraphs of the J-graph.}. The \emph{appendage subnetwork} $H$ is the subnetwork consisting of all the appendage nodes and their incident edges. Each strongly connected component of $H$ is called an \emph{appendage path component}. In \Cref{ex:factoring}, $5 \neutFBL 6$ is an appendage path component of the appendage subnetwork, so is $\{7\}$.  

Looking at \Cref{fig:factoring-union}, we would like to say the node $7$ is different from $5$ and $6$, since $7$ belongs to I/O component $\mc I_\mrm{I/O}$. 
This distinction is accomplished by asking whether there exists a cycle that contains the appendage node and a simple node that is not super-simple%
    \footnote{Formally, let $\mc K$ be an appendage path component. Then a cycle is \emph{$\mc K$-bad} if it contains both a node in $\mc K$ and a simple node that is not super-simple. The component $\mc K$ is said to satisfy the \emph{no cycle condition} if there are no $\mc K$-bad cycles containing nodes in $\mc K$~\cite{WangHuangAntoneliGolubitsky2021}. In \Cref{ex:factoring}, the cycle $2 \neutFBL 7$ is said to be $7$-bad, while the cycle $5 \neutFBL 6$ is said to satisfy the no-cycle condition since the only cycles containing the node $5$ (and $6$) either uses super-simple nodes, or no simple nodes at all.}%
. In the language of \Cref{sec:factoring}, this is equivalent to asking whether the appendage node belongs to the I/O component $\mc I_\mrm{I/O}$ or a strongly connected component of $\mc E$. 
The former (e.g., $G[\{7\}]$), they call a \emph{cycle appendage path component}; the latter (e.g., $G[\{5,6\}]$) is a \emph{no cycle appendage path component}.

Finally, the subgraph induced by the simple nodes and nodes in a cycle appendage path component (i.e., $\mc I_\mrm{I/O}$) is the \emph{structural subnetwork}~\cite[Definition~1.18]{WangHuangAntoneliGolubitsky2021}. This subgraph is associated to non-principal minors of $\mbc A$ (\Cref{lem:app-factoring-IOcomponents}). In contrast, the no cycle appendage path components (i.e., the strongly connected components of $\mc E$) are associated to principal minors (\Cref{lem:appendix-factor-nonIOcomponents}). 

The authors then classify the type of quasi-adaptation based on the degree of the irreducible factor as well as whether it is associated to a structural subnetwork or no cycle appendage path component~\cite[Section~1.9]{WangHuangAntoneliGolubitsky2021}. 
A degree-1 appendage block (corresponding to an irreducible factor $f_{ii}$) is said to be a \emph{null degradation}. A degree-1 structural block (corresponding to an irreducible factor $f_{ij}$) is called \emph{Haldane}, e.g., $f_{84}$ in \Cref{ex:factoring}. Higher degree appendage blocks are simply called \emph{appendage homeostasis of degree-$k$}, e.g., $(f_{56}f_{65}-f_{55}f_{66})$ in \Cref{ex:factoring}. Finally, higher degree structural blocks are classified based on whether or not the subgraph contains a cycle. Without a cycle, it is simply called a \emph{feedforward loop}, but with a cycle, it is called \emph{degree-$k$ structural homeostasis}, e.g., in \Cref{ex:factoring}, the irreducible factor with three terms correspond to a \emph{degree-4 structural homeostasis}.

\begin{rmk}
\label{rmk:multiple-controls} 
    In \cite{MadeiraAntoneli2022}, the framework of \cite{WangHuangAntoneliGolubitsky2021} was extended to systems where the control parameter directly influences more than one state variable. There, the theory and terminology largely remain the same. The only difference is that a degree-$k$ structural homeostasis where the first node is the control parameter\footnote{The node for the control parameter $u$ is technically not shown, so it appears as half-edges.} is called a \emph{input counterweight homeostasis}. 
\end{rmk}

\putbib[cit.bib]
\end{bibunit}

\end{document}